\documentclass[a4paper,english]{lipics-v2016}
\usepackage{microtype}%if unwanted, comment out or use option "draft"

\usepackage{color}%colored text with \textcolor{color}{text}

\usepackage{amsmath}
\usepackage{amsthm}
\usepackage{amssymb}
\usepackage{gensymb}%degree
\usepackage[detect-weight=true, binary-units=true]{siunitx}

\usepackage{graphicx}
\usepackage{caption}
\usepackage{subcaption}
\usepackage{tikz}
\usetikzlibrary{arrows}
\usepackage{wrapfig}
\usepackage{thmtools}
\usepackage{thm-restate}

\usepackage[colorinlistoftodos,prependcaption,textsize=small]{todonotes}

\newcommand{\old}[1]{{}}

\usepackage[linesnumbered, ruled, vlined]{algorithm2e}
\title{Covering Tours and Cycle Covers with Turn Costs: Hardness and Approximation}
\titlerunning{Covering Tours and Cycle Covers with Turn Costs} %optional, in case that the title is too long; the running title should fit into the top page column

\author[1]{S\'andor P.~Fekete}
\author[1]{Dominik Krupke}
\affil[1]{Department of Computer Science, TU Braunschweig, 38106 Braunschweig, Germany\\
  \texttt{\{s.fekete,d.krupke\}@tu-bs.de}}
\authorrunning{S.\,P.~Fekete and D.~Krupke} %mandatory. First: Use abbreviated first/middle names. Second (only in severe cases): Use first author plus 'et. al.'

\Copyright{S.\,P.~Fekete and D.~Krupke}%mandatory, please use full first names. LIPIcs license is "CC-BY";  http://creativecommons.org/licenses/by/3.0/

\subjclass{F.2.2 Nonnumerical Algorithms and Problems}
\keywords{Geometric optimization, optimal tours, optimal cycle covers, turn cost, NP-hardness, approximation}% mandatory: Please provide 1-5 keywords
\EventShortTitle{arXiv}
\begin{document}

\maketitle

\begin{abstract}
	We investigate a variety of problems of finding tours and cycle covers with
minimum turn cost.  Questions of this type have been studied in the past, with
complexity and approximation results, and open problems dating back to
work by Arkin et al.~in 2001. A wide spectrum of practical applications have
renewed the interest in these questions, and spawned variants: for {\em full
coverage}, every point has to be covered, for {\em subset coverage}, specific
points have to be covered, and for {\em penalty coverage}, points may be left
uncovered by incurring an individual penalty. 

\indent
We make a number of contributions. We first show that finding a minimum-turn
(full) cycle cover is NP-hard even in 2-dimensional grid graphs, solving the
long-standing open \emph{Problem~{53}} in \emph{The Open Problems Project}
edited by Demaine, Mitchell and O'Rourke. We also prove NP-hardness of finding
a {\em subset} cycle cover of minimum turn cost in {\em thin} grid graphs, for
which Arkin et al.~gave a polynomial-time algorithm for full coverage; this
shows that their boundary techniques cannot be applied to compute exact
solutions for subset and penalty variants.

%On the positive side, we establish the first constant-factor approximation algorithms for all considered subset and penalty
%problem variants, making use of LP/IP techniques.
%For {\em full} coverage in more general grid graphs (e.g., hexagonal grids), our approximation factors
%are better than the combinatorial ones of Arkin et al. 
%Our approach can also be extended to other geometric variants, such as scenarios with obstacles 
%%and linear combinations of turn and distance costs.
%\todo[inline]{This abstract is focusing very much on grid graphs. This is not critical but of course it does not match the paper perfectly. Only the last sentence even mentions the geometric variant.}

On the positive side, we establish the first constant-factor approximation
algorithms for all considered subset and penalty problem variants for very
general classes of instances, making use of LP/IP techniques. For these generalized
graph problems with many possible edge directions (and thus, turn angles, such as in
hexagonal grids or higher-dimensional variants), our approximation factors
also improve the combinatorial ones of Arkin et al.  Our approach can also be
extended to other geometric variants, such as scenarios with obstacles and
linear combinations of turn and distance costs.

\end{abstract}

\section{Introduction}
%Tours
Finding roundtrips of minimum cost is one of the classic problems of theoretical computer science.
In its most basic form, the objective of the \emph{Traveling Salesman Problem (TSP)} is to minimize 
the total length of a single tour that covers all of a given set of locations. If the tour
is not required to be connected, the result may be a {\em cycle cover}: a set of closed subtours
that together cover the whole set. This distinction makes a tremendous difference
for the computational complexity: while the TSP is NP-hard, computing a cycle cover of minimum
total length can be achieved in polynomial time, based on matching techniques.

Evaluating the cost for a tour or a cycle cover by only considering its length may not always 
be the right measure. Fig.~\ref{fig:drone} shows an example application, in which a 
drone has to sweep a given region to fight mosquitoes that may transmit dangerous diseases.
As can be seen in the right-hand part of the figure, by far the dominant part of the overall
travel cost occurs when the drone has to change its direction. (See our related video and
abstract~\cite{drone_vid} for more details, and the resulting tour optimization.)
There is an abundance of other related applied work, e.g., mowing lawns
or moving huge wind turbines~\cite{wind}.

\vspace*{-3mm}
\begin{figure}[h]
                \includegraphics[width=.30\textwidth]{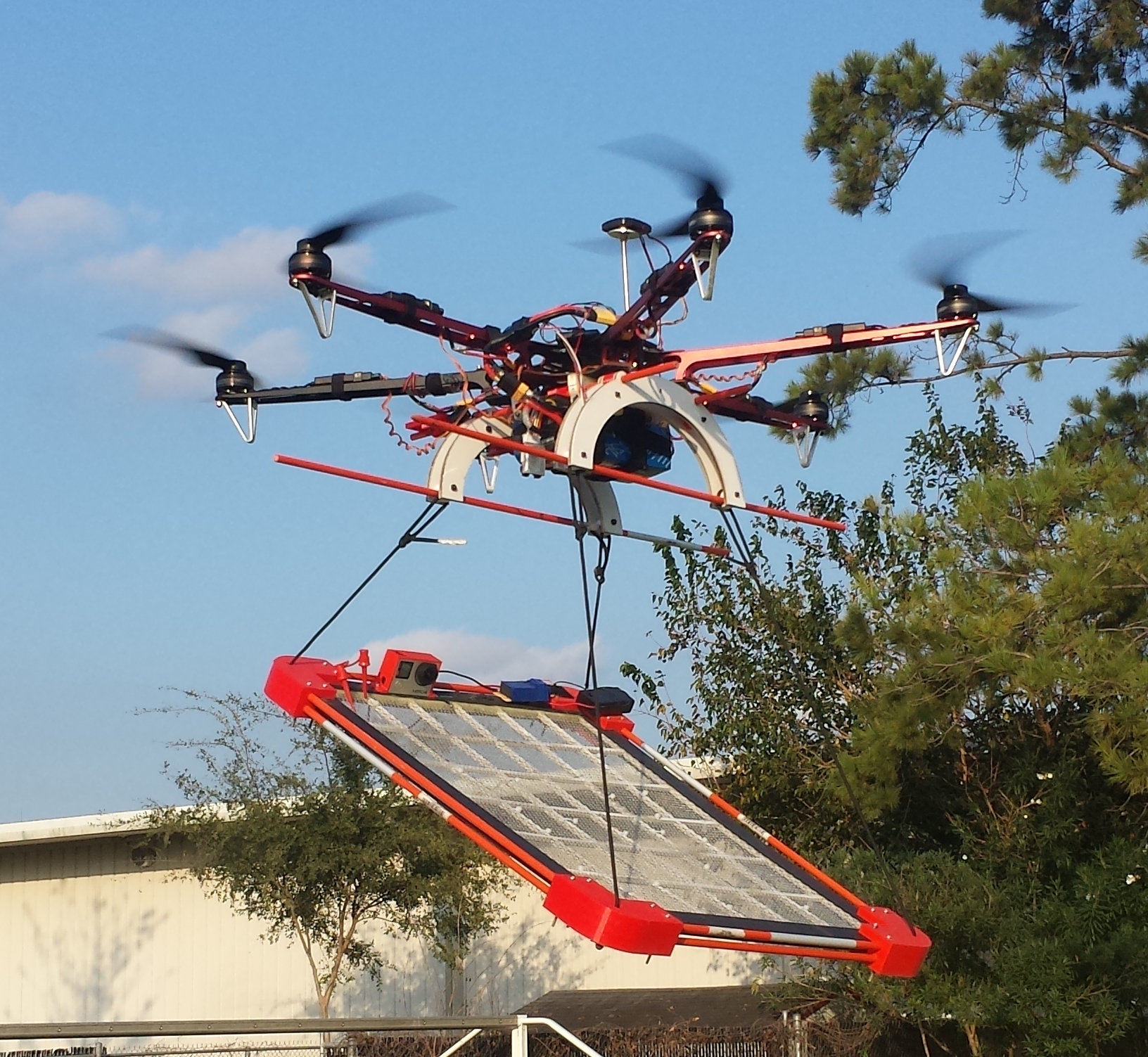}
                \includegraphics[width=.31\textwidth]{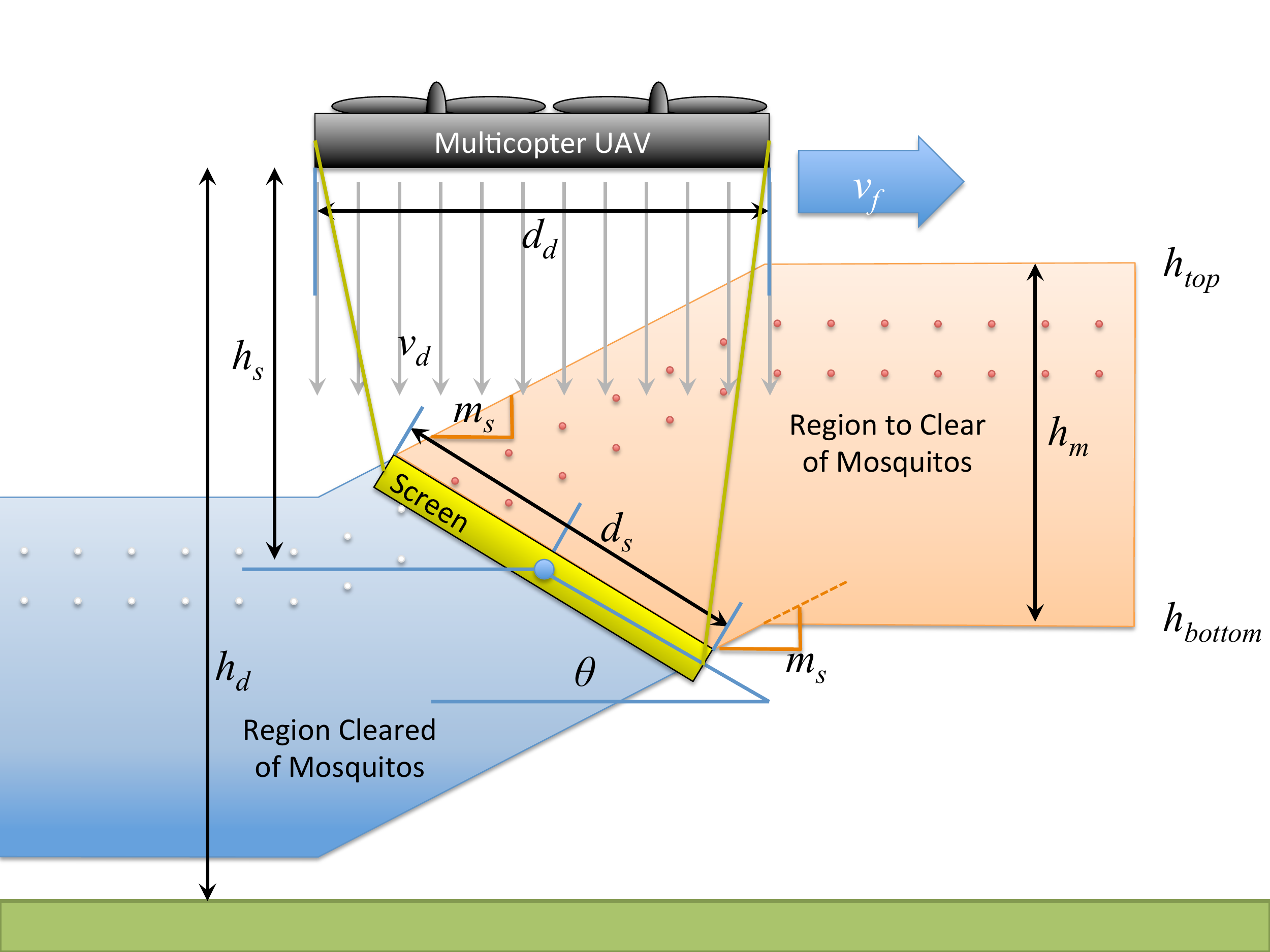}
                \includegraphics[width=.35\textwidth]{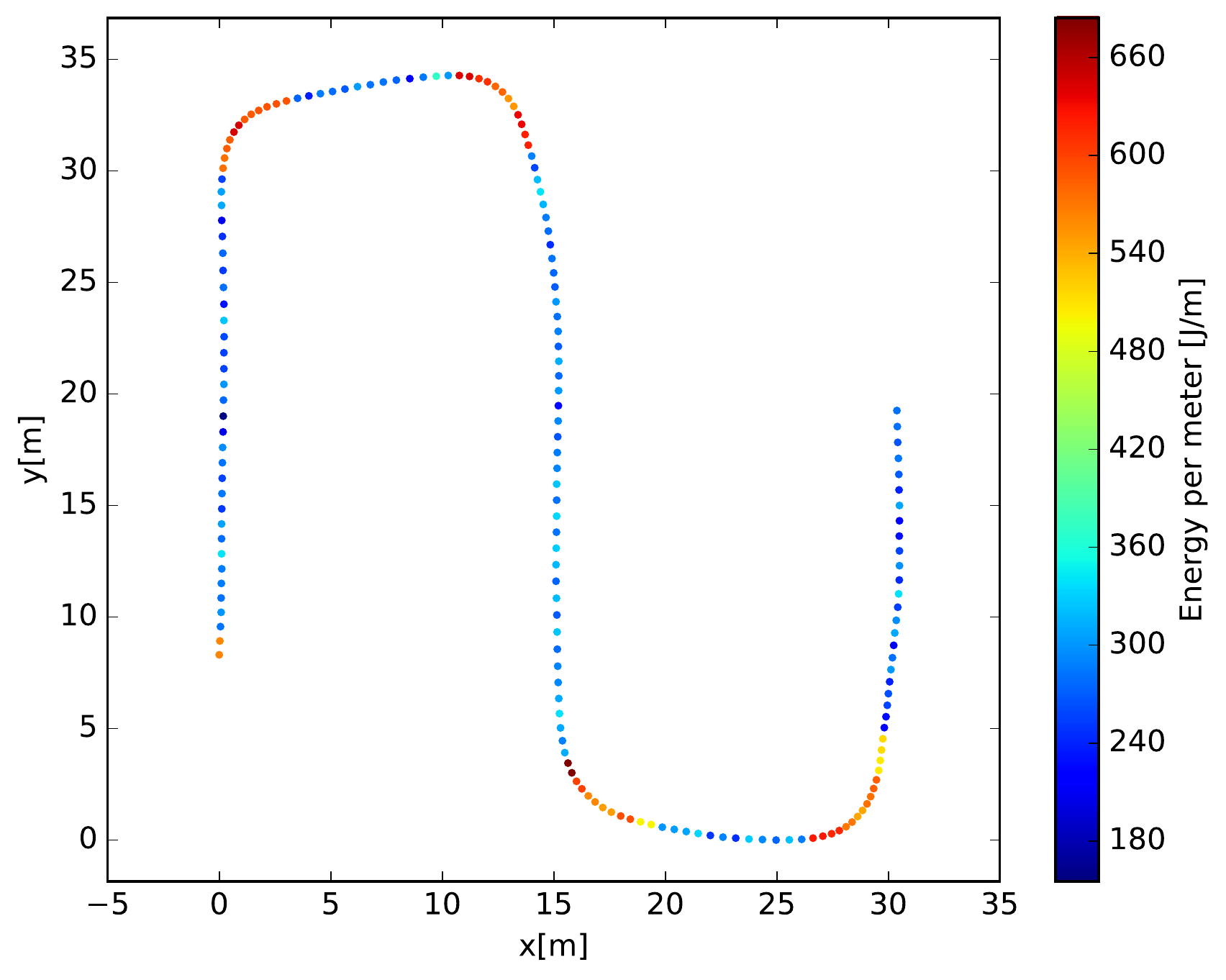}
	\caption[Mosquito hunting drone]{{\bf (Left)} A drone equipped with an
	electrical grid for killing mosquitoes. {\bf (Middle)} Physical aspects of the flying
	drone. {\bf (Right)} Making turns is expensive. See our related video at
	\url{https://www.youtube.com/watch?v=SFyOMDgdNao} for details, and
\cite{drone_vid} for an accompanying abstract.} \label{fig:drone}
\vspace*{-6mm}
\end{figure}

For many purposes, two other variants are also practically important: 
for \emph{subset coverage}, only a prespecified subset of locations needs to be visited, while for \emph{penalty coverage},
locations may be skipped at the expense of an individual penalty.
%(See Section~\ref{sec:prelim} for formal definitions.)
%\todo[inline]{The subset and penalty variant are not that formally defined, so one might save this insertion.}
From the theoretical side, Arkin et al.~\cite{arkin2005optimal} showed that finding minimum-turn tours in grid graphs is NP-hard,
even if a minimum-turn cycle cover is given. 
%The NP-hardness follows trivially for full coverage tours, subset coverage tours, and penalty coverage tours with weighted distance and turn costs.
The question whether a minimum-turn cycle cover can be computed in polynomial time (just like a minimum-length cycle cover) has been open for at least 17 years,
dating back to the conference paper~\cite{arkin2001optimal}; it has been listed for 15 years as \emph{Problem~{53}}
in \emph{The Open Problems Project} edited by Demaine, Mitchell, and O'Rourke~\cite{openproblemproject}.
%In the meantime, Maurer~\cite{maurer2009diploma} has shown that the \emph{minimum-turn cycle partition}, where every vertex has to be covered exactly once, can be calculated in polynomial time and stated the assumption that \emph{minimum-turn cycle cover} is NP-hard.
In Section~\ref{sec:gg:hardness} we resolve this problem by showing that computing a minimum-turn cycle cover in planar grid graphs is 
indeed NP-hard.

This makes it important to develop approximation algorithms.
In Section~\ref{sec:approx}, we present a general technique based on Integer Programming (IP) formulations and their Linear Programming (LP) relaxations. 
By the use of polyhedral results and combinatorial modifications, we can establish constant approximation factors of this technique for all problem variants.

\subsection{Related Work}

{\bf Milling with Turn Costs.}
Arkin et al.~\cite{arkin2001optimal,arkin2005optimal} introduce the problem of milling (i.e., ``carving out'') with turn costs.
They show hardness of finding an optimal tour, even in {\em thin} 2-dimensional grid graphs (which do not contain an induced $2\times 2$ subgraph)
with a given optimal cycle cover.
%If a cycle cover of cost $c$ in such a 2-dimensional grid graph is given, a tour of cost $1.5c$ can be derived by connecting the cycles.
They give a $2.5$-approximation algorithm for obtaining a cycle cover, resulting in a $3.75$-approximation algorithm for tours.
The complexity of finding an optimal cycle cover in a 2-dimensional grid graph was established as \emph{Problem~{53}} in \emph{The Open Problems Project}~\cite{openproblemproject}.

Maurer~\cite{maurer2009diploma} proves %in his (German) diploma thesis 
that a cycle {\em partition} with a minimum number of turns in grid graphs can be computed in polynomial time
and %He also 
performs practical experiments for optimal cycle covers. %his IP-formulation is a stepping stone for one of our formulations.
%this context, also %  with instances of size less than $50*50$ for grid graphs and less than $50$ points for general points in $\mathbb{R}^2$.
%His IP-formulation equals the core of one of our formulations.
%.For the problem in 2-dimensional grid graphs, 
De Assis and de Souza~\cite{de2011experimental} computed a provably optimal solution for
an instance with 76 vertices. %provide an optimal IP-formulation and a heuristic.
%They performed experiments with instances of size less than \num{76} vertices for the exact IP-formulation.
%For the heuristic, instances with less than \num{500} vertices were tested.
%Note that the exact IP-formulation proposed in this paper can solve many instances with more than \num{700} vertices in reasonable time;
%and our forthcoming paper~\cite{dom3}; see~\cite{eurocg} for a preliminary abstract on the geometric part.
For the abstract version on graphs, Fellows et al.~\cite{fellows2009abstract} %study the parameterized complexity and 
show that the problem is fixed-parameter tractable by the number of turns, tree-width, and maximum degree.
Benbernou~\cite{benbernou2011diss} considered milling with turn costs on the surface of polyhedrons in the 3-dimensional grid.
She gives a corresponding 8/3-approximation algorithm for tours.

Note that the theoretical work presented in this paper has significant practical implications. As described in our 
forthcoming conference paper~\cite{ALENEX19}, the IP/LP-characterization presented in Section~\ref{sec:approx} can be modified and 
combined with additional algorithm engineering techniques to allow solving instances with more than 1000 pixels to 
provable optimality (thereby expanding the range of de Assis and de Souza~\cite{de2011experimental} by a factor of 15),
and computing solutions for instances with up to 300,000 pixels within a few percentage points (thereby showing that
the practical performance of our approximation techniques is dramatically better than the established worst-case bounds).

For mowing problems, i.e., covering a given area with a moving object that may leave the region,  %it is allowed to leave the considered polygon.
%When this problem is discretized to a set of points to be visited, it is called the \emph{minimum bends traveling salesman problem}.
Stein and Wagner~\cite{wagner2001approximation} give a 2-approximation algorithm on the number of turns for the case of orthogonal movement.
% Milling/Mowing with Distance Costs
If only the traveled distance is considered, Arkin et al.~\cite{arkin2000approximation} provide approximation algorithms for milling and mowing.
%They give a 2.5-approximation for minimum length milling in orthogonal (not necessarily integral) polygons with a unit square cutter.
%For simple orthogonal polygons, an $11/5$-approximation is given and, in case we can reduce the problem to finding a covering tour in a grid graph, a $6/5$-approximation algorithm is given (which improves a previous result of Ntafos~\cite{ntafos1992watchman} of a $4/3$-approximation algorithm).
%For the mowing variant, a $3+\epsilon$-approximation is provided that internally uses a PTAS for the Euclidean TSP; this can also be replaced by other approximation
%algorithms, e.g.,
%the common algorithm of Christofides~\cite{christofides1976worst}.
%This problem has also been considered by Iwano et al.~\cite{iwano1994traveling} but with the name `the traveling cameraman problem'.
%TODO possibly TSP in GG
%\todoi{AFAIK there is a LTAS for TSP in GG}

%\todoi{Some reviewer complained that we did not cite results of angular tsp. Seemed to be an author of \cite{oswin2017minimization}. "The angular traveling salesman searches for a tour on points that minimizes the overall turn angles. Hardness and approximation is given by Aggarwal et al.~\cite{TODO}, this problem seems to be very hard to solve optimally with integer programming as shown by Oswin et al.~\cite{oswin2017minimization}." More can still be found below but we should waste much more space.}

\old{
\subsection{Milling/covering in practice}

%robot
An extensive survey on (mostly heuristic) covering paths/tours methods for robots has been made by Galceran and Carreras \cite{galceran2013survey}.
A possible technique is to partition the area by a grid as, e.g., done by Zelinsky et al.~\cite{zelinsky1993planning}, or Gabriely and Rimon~\cite{gabriely2002spiral}.
The grid graphs considered in this thesis are assumed to origin from such a discretization.
A different approach is to partition the area into simpler geometric areas (e.g. trapezoidal map) and then use simple patterns to cover the simple areas.
For example, Huang~\cite{huang2001optimal} does such a partition and then calculates the optimal orientation for a zig-zag coverage of these simple areas.
This correlates to the number of turns needed for the zig-zag coverage.

%Milling
Covering tour optimization is also important for cutter.
Yao et al.~\cite{yao2002hybrid} analyze different cutter path generation methods and propose a greedy and a branch and bound algorithm for generating so-called hybrid cutter paths that use different path types for different parts.
They also mention the increase of the costs due to sharp turns.

%Harvesting
Covering path planning is also important on larger scales, such as for autonomous agricultural machines on pixels, whose high precision GPS allow them to follow previously optimized tours.
Corresponding work has been done by, e.g., Ta{\"\i}x et al.~\cite{taix2003path} and Ali et al.~\cite{ali2009infield}.
Ali et al. ~\cite{ali2009infield} consider tour planning for tractors for crop harvesting operations including turn penalties.

%UAV
Covering is also done by UAVs for area observation.
Some UAVs have strong restrictions on the turn radius that usually cannot be ignored for covering tour planning.
Such optimizations have been considered by, e.g., Ahmadzadeh et al.~\cite{ahmadzadeh2008optimization} and Agarwal et al.~\cite{agarwal2005aco}.
}

\old{
\subsection{Quadratic TSP}

Chapter~\ref{chapter:geo} generalizes the problem from grid graphs to points in the $\mathbb{R}^2$ plane with polygonal obstacles and a limited discrete set of orientations a point can be passed through.
The cost consists of the weighted sum of turning angles and the length.
While in the grid graph only orthogonal movement ($90\degree$ and $180\degree$ turns) were allowed, now all turns ($0\degree$ - $180\degree$ clockwise or counterclockwise) are allowed and weighted linearly.
This is related to the \emph{quadratic traveling salesman problem (QTSP)} where the cost of an edge depends on the previous edge (the costs are not induced by geometry but can be arbitrarily selected).

J{\"a}ger and Molitor~\cite{jager2008algorithms} give two exact solution methods and seven heuristics for the asymmetric version.
They were only able to give exact solutions for instances with 24 points and heuristic solutions for 44 points.
They also noted that random instances seem to be simpler than real instances.

Fischer et al.~\cite{fischer2014exact} give similar but more advanced results on the symmetric version. Their branch and cut integer programming approach (with problem specific cutting planes) is able to calculate optimal solutions for real-world instances (from bioinformatics) of the size of up to 100 vertices in less than 10 minutes.

An extensive study of the QTSP polytope has been made by Fischer in her dissertation~\cite{fischer2013diss} (see also \cite{fischer2013symmetric}).

More computational results with integer programming containing also more variants (e.g., the maximization instead of minimization) follow in Aichholzer et al.~\cite{aichholzerminimization}.
The tested instances have sizes of less than 100 points.

Rostami et al.~\cite{rostami2013quadratic} also consider optimal solver and lower bounding procedures for the Q(A)TSP but the experiments only consider very small instances (less than 25 points).

In this thesis, we are able to obtain an $O(\omega)$-approximation if the ways to pass a point are discretized to $\omega$ orientations (e.g. horizontally and vertically).
They also show that optimal angular metric cycle covers and tours are NP-hard to calculate.
The NP-hardness proof in Section~\ref{sec:gg:hard:subset} for subset cycle cover in thin 2-dimensional grid graphs is adapted from this proof.
}

{\bf Angle and curvature-constrained tours and paths.}
If the instances are in the $\mathbb{R}^2$ plane and only the turning angles are measured, the problem is called the \emph{Angular Metric Traveling Salesman Problem}.
Aggarwal et al.~\cite{aggarwal2000angular} prove hardness and provide an $O(\log n)$ approximation algorithm for cycle covers and tours that works even 
for distance costs and higher dimensions. As shown by Aichholzer et al.~\cite{oswin2017minimization}, 
this problem seems to be very hard to solve optimally with integer programming.
Fekete and Woeginger~\cite{fekete1997angle} consider the problem of connecting a point set with a tour for which the angles between the two successive edges are constrained.
%If only $\pm 90\degree$ turns are allowed, the problem can be solved in polynomial time.
%In many scenarios, there are no sharp turns, but the maximum curvature of a path in a geometric environment has to be limited.
Finding a curvature-constrained shortest {\em path} with obstacles has been shown to be NP-hard by Lazard et al.~\cite{lazard1998complexity}.
Without obstacles, the problem is known as the \emph{Dubins path}~\cite{dubins1957curves}
that can be computed efficiently.
For different types of obstacles, Boissonnat and Lazard~\cite{boissonnat1996polynomial}, Agarwal et al.~\cite{agarwal2002curvature} 
and Agarwal and Wang~\cite{agarwal2001approximation} provide polynomial-time algorithms or $1+\epsilon$ approximation algorithms, respectively. 
Takei et al.~\cite{takei2010practical} consider the solution of the problem from a practical perspective.
The \emph{Dubins Traveling Salesman Problem} is considered by Le Ny et al.\cite{le2012dubins}

%\old{
%\todoi{The following part could possibly be removed to save space as we use it for connecting the penalty cycles but it is not essential for the primary part of this paper.}

{\bf Related combinatorial problems.}
Goemans and Williamson~\cite{goemans1995general} provide an approximation technique for constrained forest problems and similar problems that deal with penalties.
In particular, they provide a 2-approximation algorithm for \emph{Prize-Collecting Steiner Trees} in general symmetric graphs and the \emph{Penalty Traveling Salesman Problem} in graphs that satisfy the triangle inequality. 
An introduction into approximation algorithms for prize-collecting/penalty problems, k-MST/TSP, and minimum latency problems is given by Ausiello et al.~\cite{ausiello2007prize}.
%}

\subsection{Preliminaries}
\label{sec:prelim}
% Definition of discr. angular metric as approximation of the angular-metric
The angular metric traveling salesman problem resp.\ cycle cover problem ask for a cycle resp.\ set of cycles such that a given set $P$ of $n$ points in $\mathbb{R}^d$ is covered and the sum of turn angles in minimized.
A cycle is a closed chain of segments and covers the points of the segments' joints (at least two of $P$).
The turn angle of a joint is the angle difference to $180^\circ$.
In the presence of polygonal obstacles, cycles are not allowed to cross them.
%\todo[inline]{Objective Functions and problem variants. This is not nicely integrated.}
We consider three coverage variants: Full, subset, and penalty.
In full coverage, every point has to be covered.
In subset coverage, only points in a subset $S\subseteq P$ have to be covered (which is only interesting for grid graphs).
In penalty coverage, no point has to be covered but every uncovered point $p\in P$ induces a penalty $c(p)\in \mathbb{Q}^+_0$ on the objective value.
Optionally, the objective function can be a linear combination of distance and turn costs.

In the following, we introduce the \emph{discretized angular metric}, by considering for every 
point $p\in P$ a set of $\omega$ possible orientations (and thus, $2\omega$ possible directions)
for a trajectory through $p$. We model this by considering for each $p\in P$ a set
$O_p$ of $\omega$ infinitely short segments, which we call \emph{atomic strips}; a 
point is covered if one of its segments is part of the cycle, see Fig.~\ref{fig:angtsp_to_atomic_strips}.
%With $\omega=n-1$, the angular metric can be exactly replicated by orienting the atomic strips to all neighbors.
The corresponding selection of atomic strips is called \emph{Atomic Strip Cover}, i.e., a selection of one $o\in O_p$ for every $p\in P$.
\begin{figure}[h]
\vspace*{-3mm}
	\centering
	\includegraphics[width=0.9\textwidth]{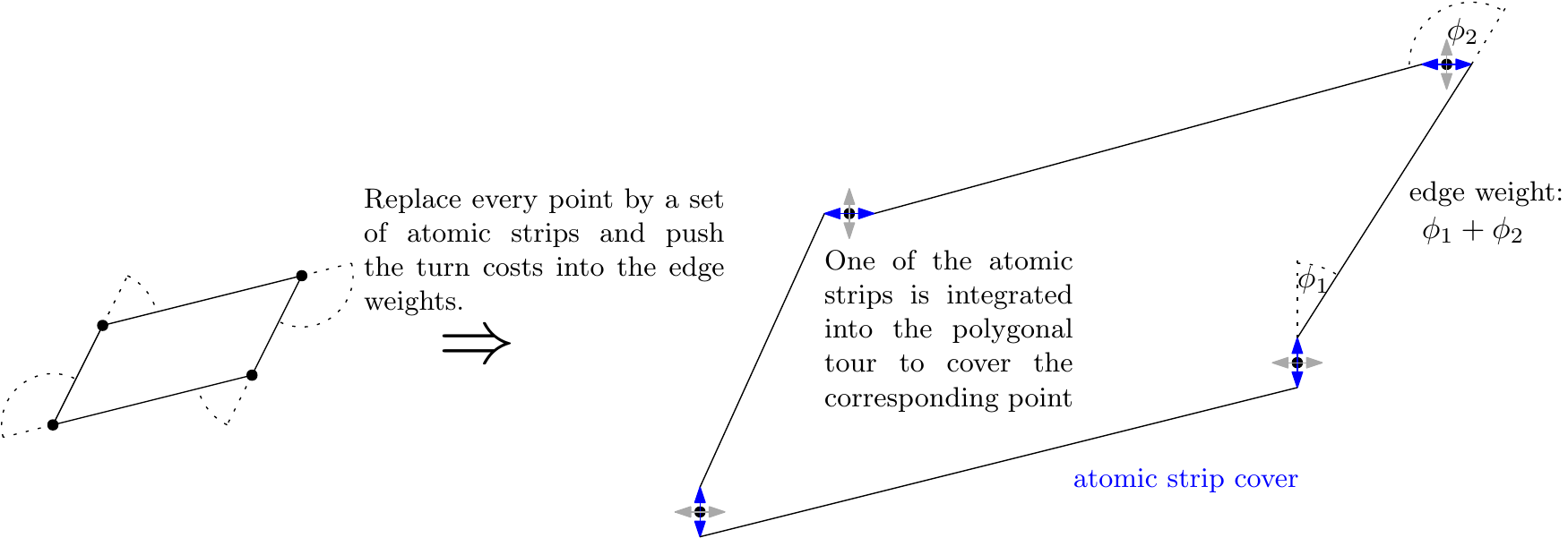}
	\caption{Transforming an angular metric TSP instance and solution to an instance based on atomic strips, which can be considered infinitely small segments.}
	\label{fig:angtsp_to_atomic_strips}
\vspace*{-3mm}
\end{figure}

% define graph and abstract problem: Semin-Quadratic CC
The atomic strips induce a weighted graph $G_O(V_O, E_O)$ with the endpoints of the atomic strips as vertices and the connections between the endpoints as edges.
The weight of an edge in $G_O$ equals the connection costs, in particular the turn costs on the two endpoints.
Thus, the cycle cover problem turns into finding an Atomic Strip Cover with the minimum-weight perfect matching on its induced subgraph.
As the cost of connections in it depends on {\em two} edges in the original graph,
we call this generalized problem (in which the edge weights do not have to be induced by geometry) the 
\emph{semi-quadratic cycle cover problem}.

It is important to note that the weights do not satisfy the triangle inequality; however, a direct connection is not more expensive than a connection that includes another atomic strip, giving rise to the following {\em pseudo-triangle inequalities}.
\begin{equation}
		\forall v_1, v_2\in V_O, w_1w_2\in O_p, p\in P:
		\begin{matrix}
			\text{cost}(v_1v_2)\leq \text{cost}(v_1w_1)+\text{cost}(w_2v_2)  \\
			\text{cost}(v_1v_2)\leq \text{cost}(v_1w_2)+\text{cost}(w_1v_2)
		\end{matrix}
		\label{eq:triang}
	\end{equation}

% Grid Graphs
Our model allows the original objective function to be a linear combination of turn and distance costs,
as it does not influence Eq.~(\ref{eq:triang}).
Instances with polygonal obstacles for 2-dimensional geometric instances are also possible 
(however, for 3D, the corresponding edge weights can no longer be computed efficiently).
A notable special case %that is especially useful optimizing area coverage such as milling or lawn mowing 
are {\em grid graphs} that arise as vertex-induced subgraph of the infinite integer orthogonal grid.
%points lie on the integral grid; a cycle is only allowed to consist of segments between two points with distance one.
In this case, a point can only be covered straight, by a simple $90^\circ$ turn, or by a $180^\circ$ u-turn.
We show grid graphs as polyominoes in which vertices are shown as {\em pixels}.
We also speak of the number of {\em simple turns} (u-turns counting as two) instead of turn angles.
More general grid graphs can be based on other grids, such as 3-dimensional integral or hexagonal grids.

% modelling the grid graph
Minimum turn cycle covers in grid graphs can be modelled as a semi-quadratic cycle cover problem with $\omega=2$ and edge weights satisfying Eq.~(\ref{eq:triang}).
One of the atomic strips represents being in a horizontal orientation (with an east and a west heading vertex) and the other being in a vertical orientation (with a north and a south heading vertex).
The cost of an edge is as follows: Every vertex is connected to a position and a direction.
The cost is the cheapest transition from the position and direction of the first vertex to the position and opposite heading of the second vertex (this is symmetric and can be computed efficiently).
We can easily transform a cycle cover in a grid graph into one based on atomic strips and vice versa, see Fig.~\ref{fig:convert} (left).
\begin{figure}[t]
\vspace*{-3mm}
	\centering
	\includegraphics[width=0.5\columnwidth]{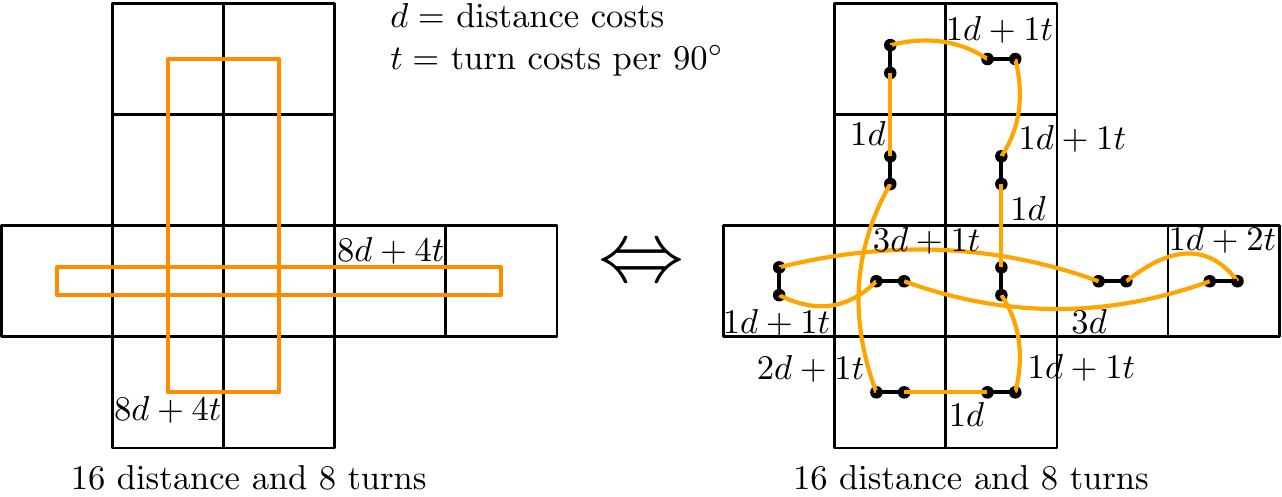}
	%\hspace*{0.04\columnwidth}
	%\includegraphics[width=0.15\columnwidth]{./figures/transformation}
	\hspace*{0.04\columnwidth}
	\includegraphics[width=0.25\columnwidth]{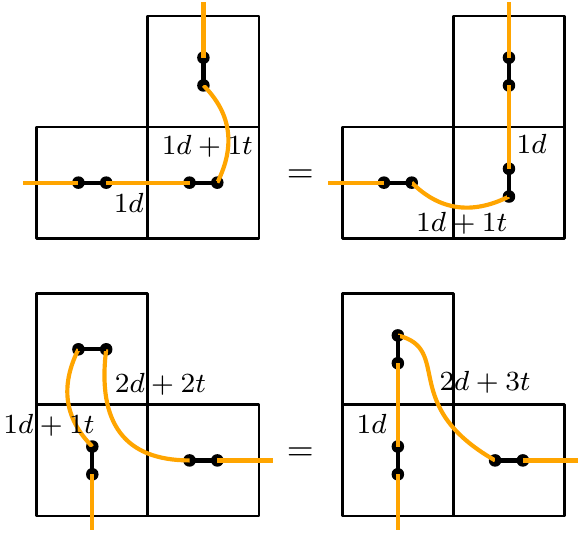}
	\caption{{\bf (Left)} From an optimal cycle cover (dotted) we can extract an
	Atomic Strip Cover (thick black), such that the matching (orange) induces an optimal
solution. %(or entering from $d(v)$). This
%cost is actually symmetric. 
	%(Middle) A simple transformation allows to efficiently calculate the
	%cheapest transition between two configurations via Dijkstra. $\tau$
	%are the turn costs for a simple turn, $\gamma$ are the distance costs. 
	{\bf (Right)} For turns it does not matter if we choose the horizontal or vertical atomic strip.}
	\label{fig:strip_orientation_turns}
	\label{fig:convert}
\vspace*{-3mm}
\end{figure}
For each pixel we choose one of its transitions.
If it is straight, we select the equally oriented strip; otherwise it does not matter, see Fig.~\ref{fig:strip_orientation_turns}~(right).
With more atomic strips we can also model more general grid graphs such as hexagonal or 3-dimensional grid graphs with three atomic strips.
%\todo[inline]{Maybe add that $\omega$ is the number of ways a point can be covered straight (undirected)?}

%\todo[inline]{Obstacles}
%The DACC and DATSP in the plane can also work with polygonal obstacles.
%The edge costs simply have to take the necessary detours into consideration.
%These can be computed similar to the euclidean shortest path.
%First only additional turns at the vertices of obstacles may be necessary (use the same local optimization argument as for euclidean shortest path just with turn angles).
%Turn costs can be integrated into Dijkstra's algorithm by simply replacing every vertex by a set of heading dependent vertices (for the linear amount of incoming and outgoing headings) and add edges with the turn costs to switch between the headings.
%Because the euclidean shortest path becomes NP-hard in $3D$, this is only possibly in $2D$.

\subsection{Our Contribution}
We provide the following results.

\begin{itemize}
	\item We resolve \emph{Problem~{53}} in \emph{The Open Problems Project}~\cite{openproblemproject} by proving that finding a cycle cover of minimum turn cost is NP-hard, even in the restricted case of grid graphs.
		%As a consequence, all relevant problem variants are NP-hard.
		We also prove that finding a subset cycle cover of minimum turn cost is NP-hard, even in the restricted case of {\em thin} grid graphs, in which no induced $2\times 2$ subgraph exists.
		This differs from the case of full coverage in thin grid graphs, which is known to be polynomially solvable~\cite{arkin2005optimal}.
	\item We provide a general IP/LP-based technique for obtaining $2*\omega$ approximations for the semi-quadratic (penalty) cycle cover problem if Eq.~(\ref{eq:triang}) is satisfied, where $\omega$ is the maximum number of atomic strips per vertex.
	\item We show how to connect the cycle covers to minimum turn tours to obtain a $6$ approximation for full coverage in regular grid graphs, $4\omega$ approximations for full tours in general grid graphs, $4\omega+2$ approximations for (subset) tours, and $4\omega+4$ for penalty tours.
\end{itemize}

To the best of our knowledge, this is the first approximation algorithm for the subset and penalty variant with turn costs. 
For general grid graphs our techniques yields better guarantees than than the techniques of Arkin et al.\ who give a factor of $6*\omega$ for cycle covers and $6*\omega+2$ for tours.
In practice, our approach also yields better solutions for regular grid graphs, see~\cite{ALENEX19}.

\section{Complexity}
\label{sec:gg:hardness}
%Intro
\emph{Problem~{53}} in \emph{The Open Problems Project}
%edited by Demaine, Mitchell, and O'Rourke~\cite{openproblemproject}
asks for the complexity of finding a minimum-turn (full) cycle cover in a 2-dimensional grid graph. This is by no means 
obvious: large parts of a solution can usually be deduced by local information and matching techniques.
In fact, it was shown by Arkin et al.~\cite{arkin2001optimal,arkin2005optimal} that the full coverage variant in {\em thin} grid graphs (which do not contain a $2\times 2$ square,
so every pixel is a boundary pixel) is solvable in polynomial time.
%Also our approximation algorithm, see next section, returned surprisingly often optimal solutions.
%A further problem is that there is an interference by neighbored pixels that have to be covered too.
In this section, we prove that finding a {\em full} cycle cover in 2-dimensional grid graphs with minimum turn cost is NP-hard,
resolving \emph{Problem~{53}}. We also show that {\em subset} coverage is NP-hard even for {\em thin} grid graphs, so the boundary
techniques by Arkin et al.~\cite{arkin2001optimal,arkin2005optimal} do not provide a polynomial-time algorithm.

\subsection{Full Coverage in Grid Graphs}
\begin{theorem}
	\label{th:gg:hard:nphard}
	It is NP-hard to find a cycle cover with a minimum number of $90^\circ$ turns ($180^\circ$ turns counting as two) in a grid graph.
\end{theorem}

%1in3 SAT
The proof is based on a reduction from~\emph{One-in-three 3SAT} (1-in-3SAT), which was shown to be NP-hard by Schaefer~\cite{schaefer1978complexity}:
for a Boolean formula in conjunctive normal form with only three literals per clause, 
decide whether there is a truth assignment that makes exactly one literal per clause {\tt true} (and exactly two literals {\tt false}).
For example, $(x_1\vee x_2 \vee x_3) \wedge (\overline{x_1}\vee \overline{x_2} \vee \overline{x_3})$ is not (1-in-3) satisfiable,
whereas $(x_1\vee x_2 \vee x_3) \wedge (\overline{x_1}\vee \overline{x_2} \vee \overline{x_4})$ is satisfiable.

%----------------------------------------------------------------------------------------
% Example
%----------------------------------------------------------------------------------------
%comment the following variable to switch between the two possibilities.
%\def\useclauseinsteadofexample{}
\ifx\useclauseinsteadofexample\undefined
\begin{figure}
	\centering
	\includegraphics[width=1.0\columnwidth]{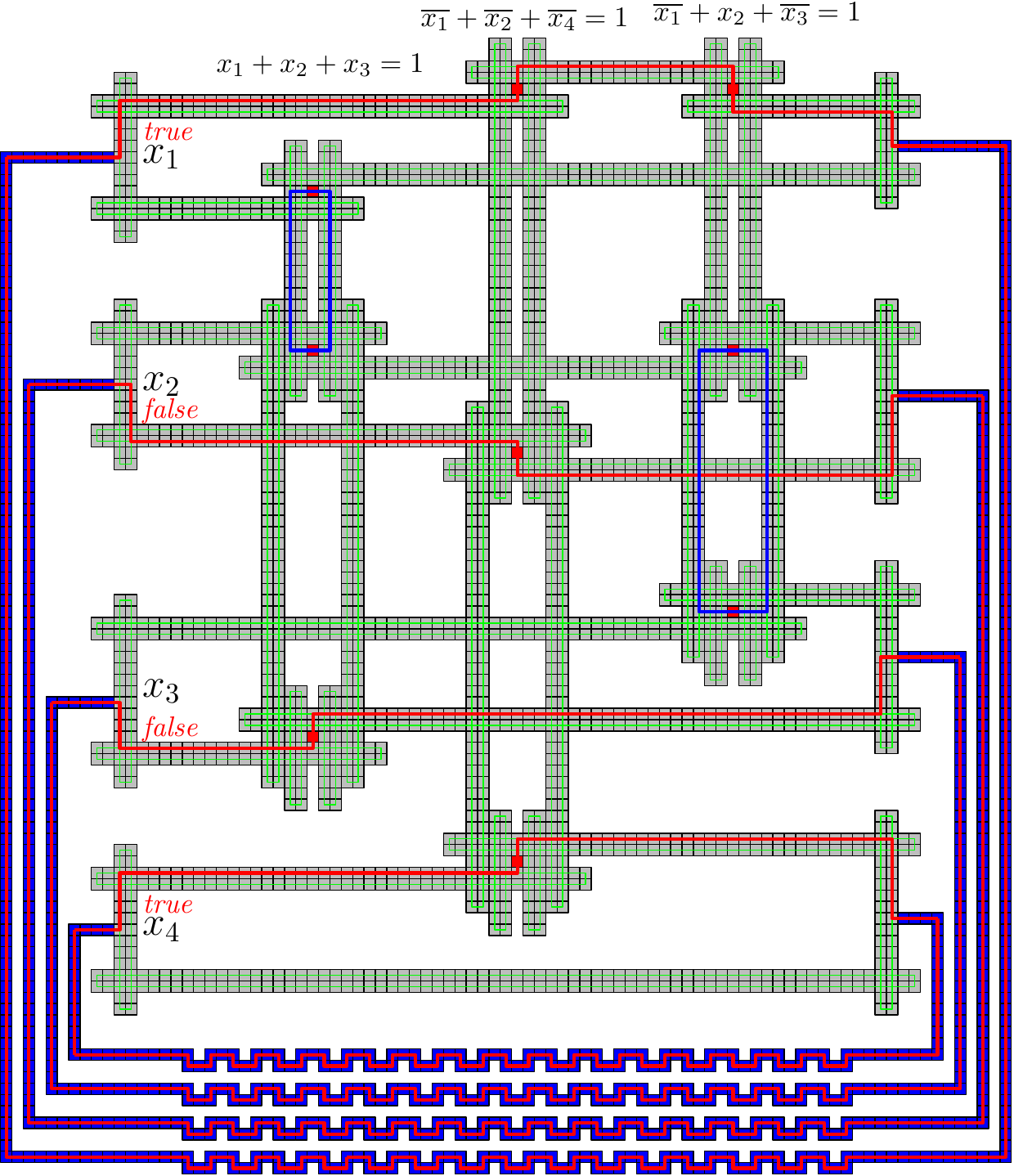}
	\caption{Representing the \emph{1-in-3SAT}-formula $x_1+x_2+x_3=1\wedge \overline{x_1}+\overline{x_2}+\overline{x_4}=1 \wedge \overline{x_1}+x_2+\overline{x_3}=1$.}
	\label{fig:nphard:simple:multiclause}
	\vspace*{-6mm}
\end{figure}
\else
\ifx\auxiliaryhardnessconstructionsymbols\undefined
\newcommand{\variablepoundsymbol}{%
	\resizebox{11pt}{!}{
		\begin{tikzpicture}[font=\tiny, every node/.style={draw=gray, text=black, inner sep=0.5pt, outer sep=0pt, minimum size=8pt} ]
			\draw (0,0.1) -- (1,0.1) -- (1,0.2) -- (0, 0.2) -- cycle;
			\draw (0.9,0) -- (0.9,0.8) -- (0.8, 0.8) -- (0.8, 0) -- cycle;
			\draw (0,0.7) -- (1,0.7) -- (1, 0.6) -- (0, 0.6) -- cycle;
			\draw (0.1,0) -- (0.1,0.8) -- (0.2, 0.8) -- (0.2, 0) -- cycle;
		\end{tikzpicture}
	}
}

\newcommand{\variableconstrsymbol}{%
	\resizebox{11pt}{!}{
		\begin{tikzpicture}[font=\tiny, every node/.style={draw=gray, text=black, inner sep=0.5pt, outer sep=0pt, minimum size=8pt} ]
			\draw (0,0.1) -- (1,0.1) -- (1,0.2) -- (0, 0.2) -- cycle;
			\draw (0.9,0) -- (0.9,0.8) -- (0.8, 0.8) -- (0.8, 0) -- cycle;
			\draw (0,0.7) -- (1,0.7) -- (1, 0.6) -- (0, 0.6) -- cycle;
			\draw (0.1,0) -- (0.1,0.8) -- (0.2, 0.8) -- (0.2, 0) -- cycle;
			\draw (0.1, 0.4) -- (-0.1, 0.4) -- (-0.1, -0.1) -- (1.1, -0.1) -- (1.1, 0.4) -- (0.9, 0.4);
		\end{tikzpicture}
	}
}

\newcommand{\variablefalsesymbol}{%
	\resizebox{11pt}{!}{
		\begin{tikzpicture}[font=\tiny, every node/.style={draw=gray, text=black, inner sep=0.5pt, outer sep=0pt, minimum size=8pt} ]		
			\draw (0,0.1) -- (1,0.1) -- (1,0.2) -- (0, 0.2) -- cycle;
			\draw (0.9,0) -- (0.9,0.8) -- (0.8, 0.8) -- (0.8, 0) -- cycle;
			\draw (0,0.7) -- (1,0.7) -- (1, 0.6) -- (0, 0.6) -- cycle;
			\draw (0.1,0) -- (0.1,0.8) -- (0.2, 0.8) -- (0.2, 0) -- cycle;

			\draw (0.1, 0.4) -- (-0.1, 0.4) -- (-0.1, -0.1) -- (1.1, -0.1) -- (1.1, 0.4) -- (0.9, 0.4);
			\draw[color=red, thick] (0.1, 0.4) -- (-0.1, 0.4) -- (-0.1, -0.1) -- (1.1, -0.1) -- (1.1, 0.4) -- (0.9, 0.4) -- (0.9, 0.1) -- (0.1, 0.1) -- (0.1, 0.4);
		\end{tikzpicture}
	}
}

\newcommand{\variabletruesymbol}{%
	\resizebox{11pt}{!}{
		\begin{tikzpicture}[font=\tiny, every node/.style={draw=gray, text=black, inner sep=0.5pt, outer sep=0pt, minimum size=8pt} ]		
			\draw (0,0.1) -- (1,0.1) -- (1,0.2) -- (0, 0.2) -- cycle;
			\draw (0.9,0) -- (0.9,0.8) -- (0.8, 0.8) -- (0.8, 0) -- cycle;
			\draw (0,0.7) -- (1,0.7) -- (1, 0.6) -- (0, 0.6) -- cycle;
			\draw (0.1,0) -- (0.1,0.8) -- (0.2, 0.8) -- (0.2, 0) -- cycle;

			\draw (0.1, 0.4) -- (-0.1, 0.4) -- (-0.1, -0.1) -- (1.1, -0.1) -- (1.1, 0.4) -- (0.9, 0.4);
			\draw[color=red, thick] (0.1, 0.4) -- (-0.1, 0.4) -- (-0.1, -0.1) -- (1.1, -0.1) -- (1.1, 0.4) -- (0.9, 0.4) -- (0.9, 0.7) -- (0.1, 0.7) -- (0.1, 0.4);
		\end{tikzpicture}
	}
}
\def\auxiliaryhardnessconstructionsymbols{}
\fi

\begin{figure}
	\centering
	\includegraphics[width=0.32\columnwidth]{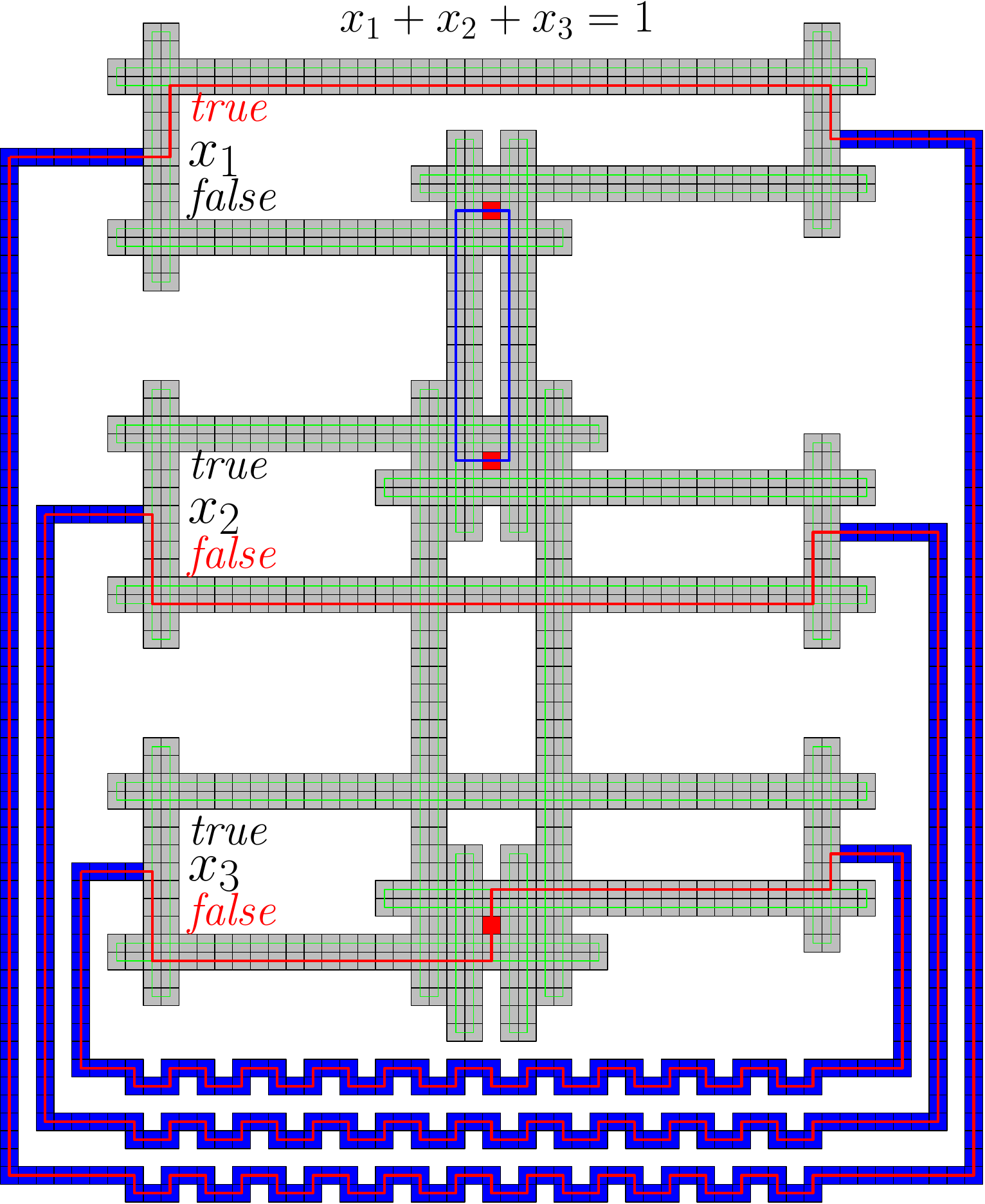}
	\includegraphics[width=0.32\columnwidth]{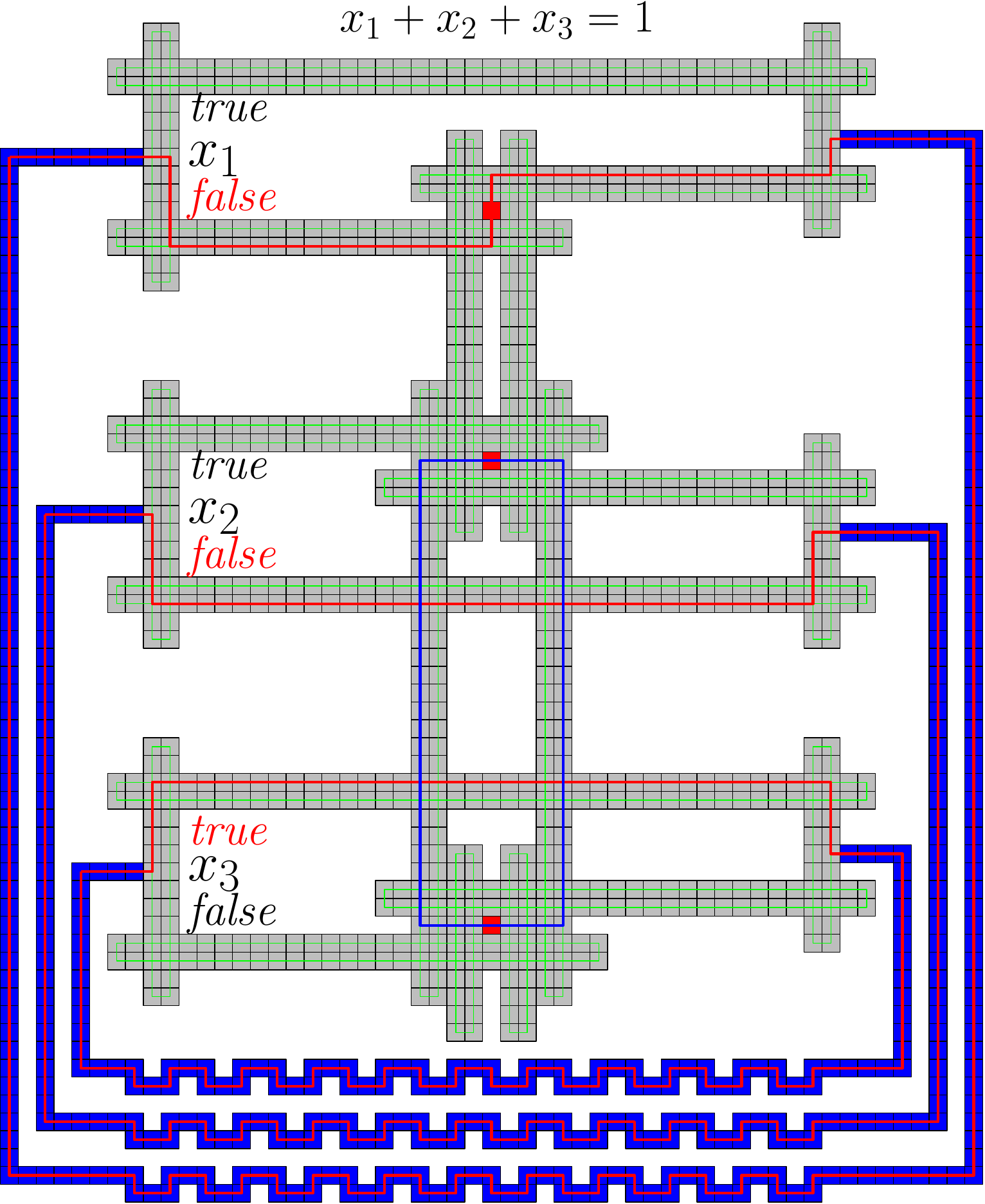}
	\includegraphics[width=0.32\columnwidth]{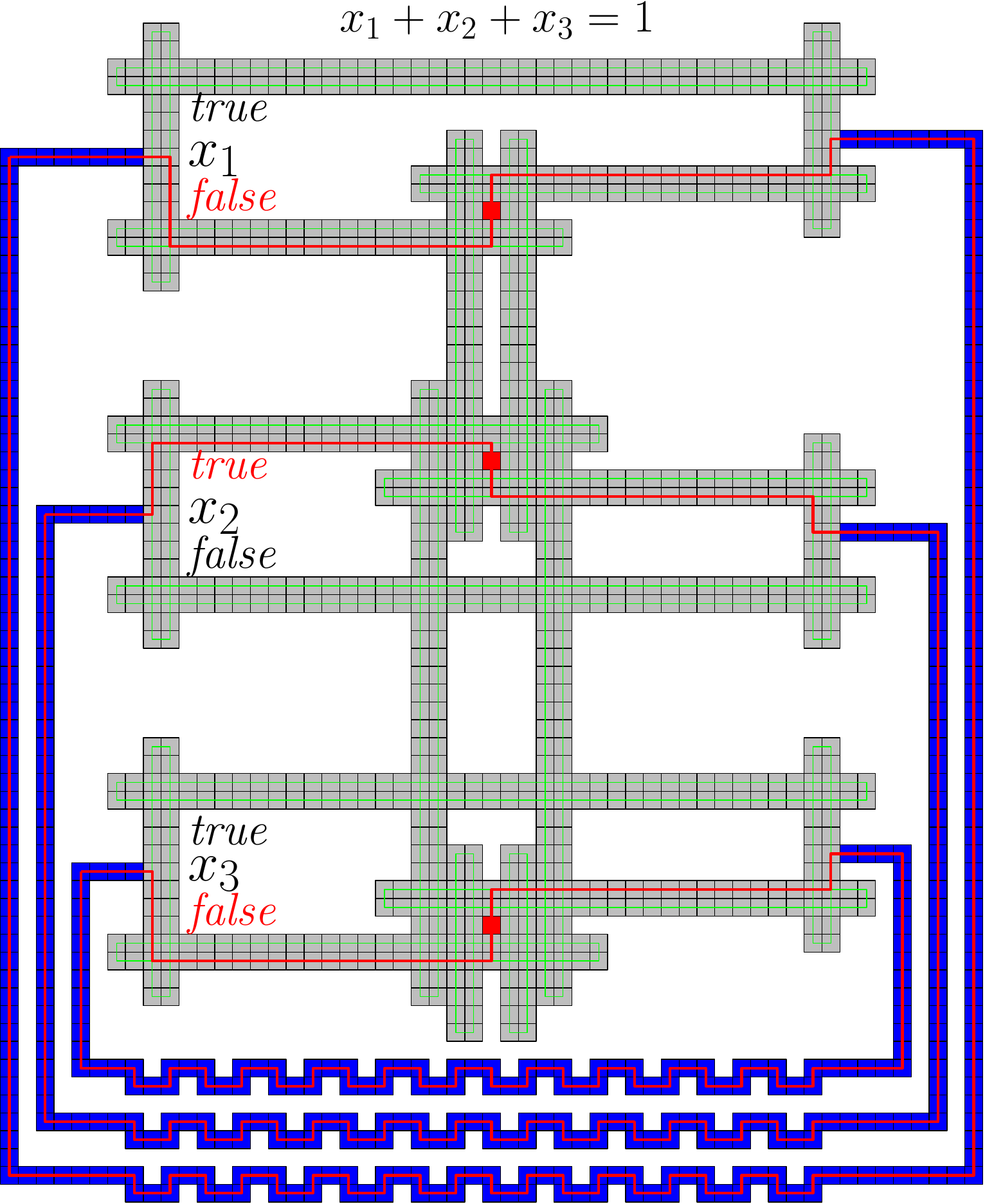}
	\caption{Construction for the one-clause formula $x_1+x_2+x_3=1$ and
	three possible solutions. Every variable has a cycle traversing 
	a zig-zagging path of blue pixels. A variable is {\tt true} if its cycle uses the upper path 
	(\variabletruesymbol) through green/red pixels, {\tt false} if it takes the lower
	path (\variablefalsesymbol). For covering the red pixels, we may
	use two additional turns. This results in three classes of
	optimal cycle covers, shown above. If we use the blue 4-turn cycle to cover
	the upper two red pixels, we are forced to cover the lower red pixel by the
	$x_3$ variable cycle, setting $x_3$ to {\tt false}. The variable cycles of
	$x_1$ and $x_2$ take the cheapest paths, setting them to {\tt true}
	or {\tt false}, respectively. The alternative to a blue cycle 
	is to cover all three red pixel by the variable cycles,
	as in the right solution.}
	% DO NOT CHANGE ABOVES TEXT! CURRENTLY THE FIGURE IS INCLUDED VIA \input{./figures/nphard_simple_clause.tex}
	\vspace*{-3mm}
	\label{fig:nphard:simple:singleclause}
\end{figure}
\fi

\ifx\fignphardsimpleincluded\undefined
\def\fignphardsimpleincluded{}

\begin{figure}
	\centering
	\includegraphics[width=0.32\columnwidth]{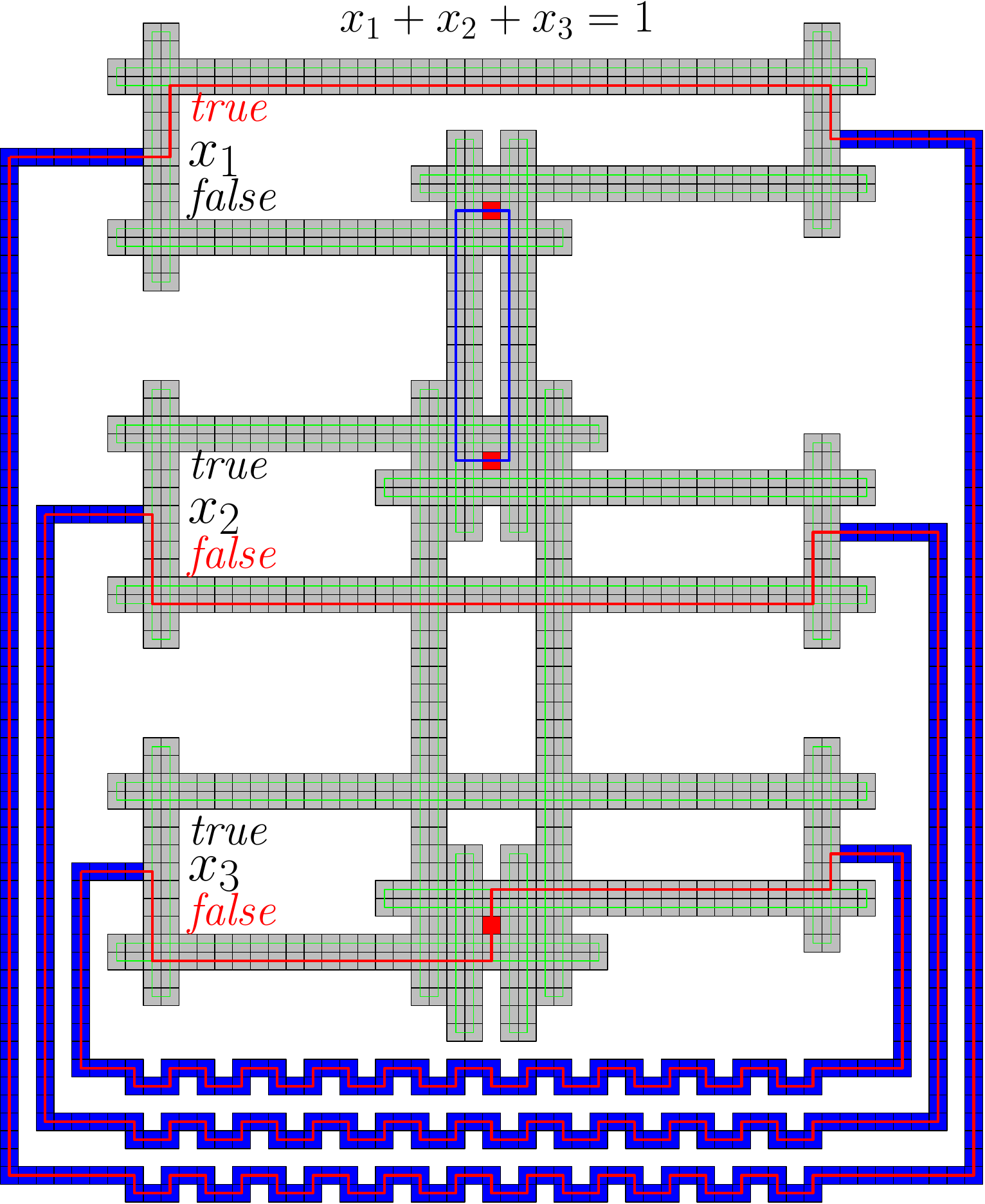}
	\includegraphics[width=0.32\columnwidth]{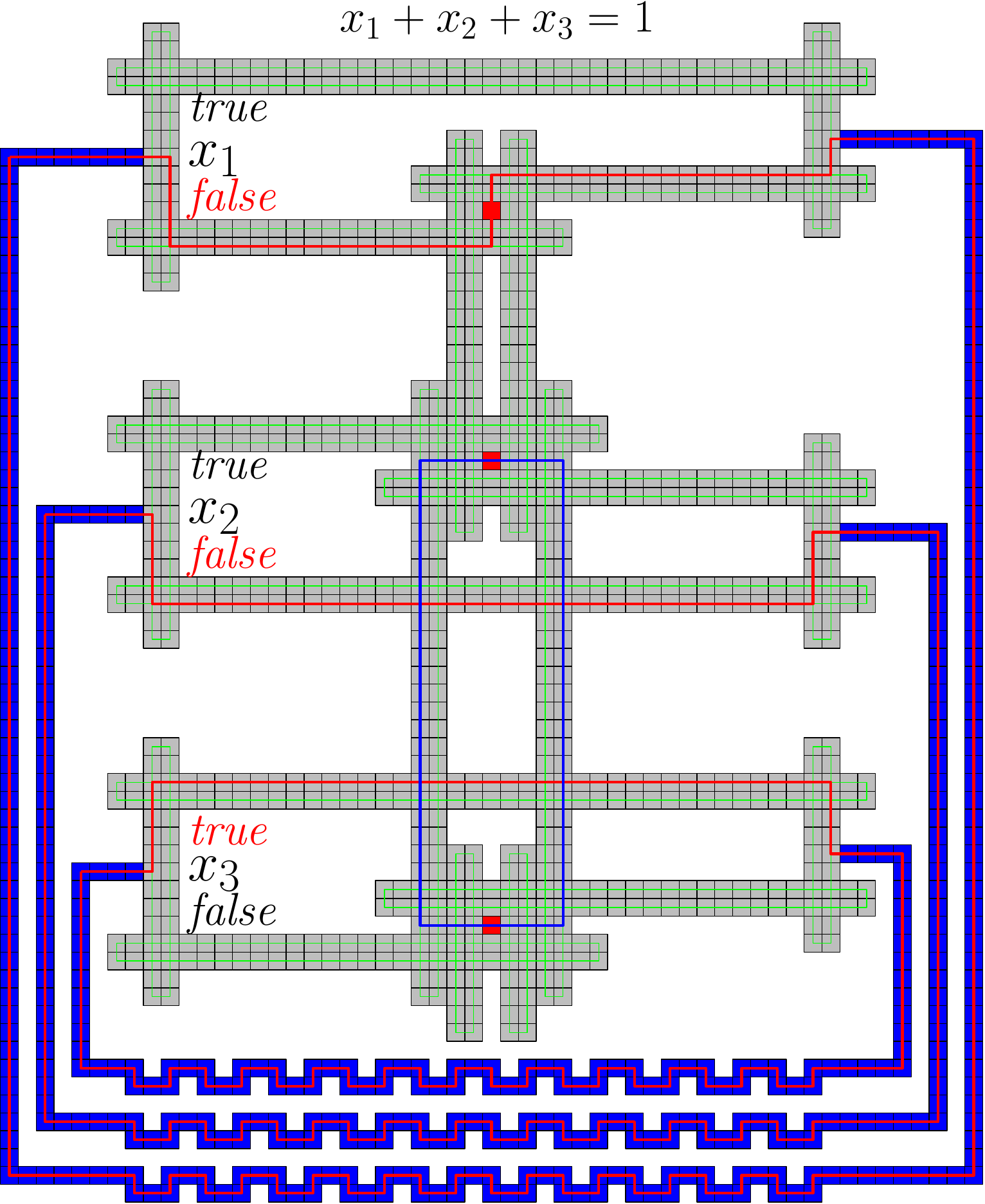}
	\includegraphics[width=0.32\columnwidth]{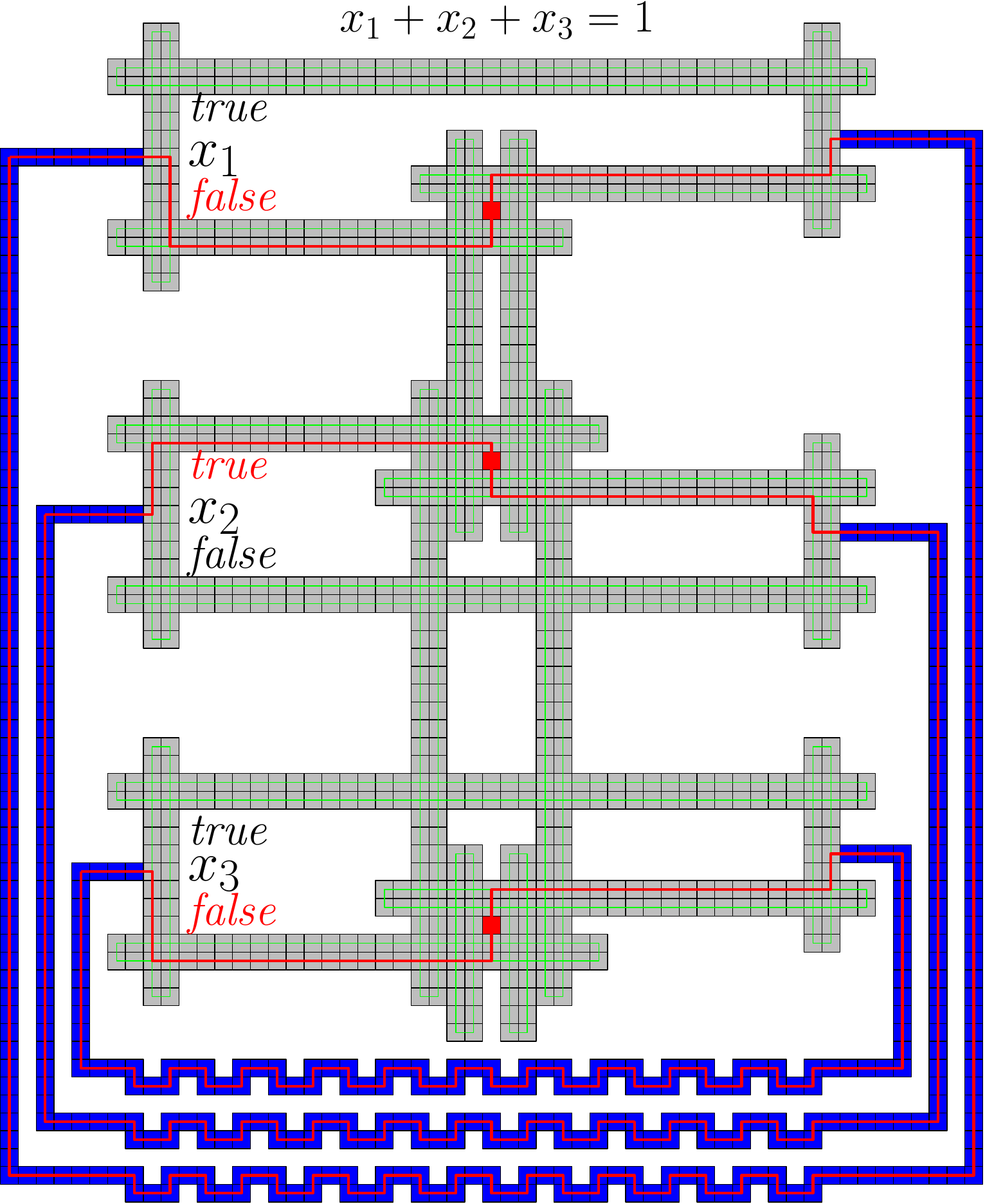}
	\caption{Construction for the one-clause formula $x_1+x_2+x_3=1$ and three possible solutions.
	Every variable has a cycle traversing a zig-zagging path of blue pixels.
	A variable is {\tt true} if its cycle uses the upper path % (\variabletruesymbol)
	through green/red pixels, {\tt false} if it takes the lower path. % (\variablefalsesymbol).
	For covering the red pixels, we may use two additional turns. This results in three classes of optimal cycle covers, shown above.
	If we use the blue 4-turn cycle to cover the upper two red pixels, we are forced to cover the lower red pixel by the $x_3$ variable cycle, setting $x_3$ to {\tt false}.
	The variable cycles of $x_1$ and $x_2$ take the cheapest paths, setting them to {\tt true} or {\tt false}, respectively.
	The alternative to a blue cycle is to cover all three red pixel by the variable cycles, as in the right solution.}
	\label{fig:nphard:simple:singleclause}
\end{figure}
\fi
%an inclusion guard makes sure to add this figure only the first time. Uncomment to exclude it (will automatically be included in the appendix)
An example is given in Fig.~\ref{fig:nphard:simple:multiclause} for the instance $x_1+x_2+x_3=1\wedge \overline{x_1}+\overline{x_2}+\overline{x_4}=1 \wedge \overline{x_1}+x_2+\overline{x_3}=1$.
%See Appendix~\ref{app:A} for full details, Fig.~\ref{fig:nphard:simple:singleclause} for 
%representing the one-clause formula $x_1+x_2+x_3=1$ with its three possible 1-in-3 solutions, and Fig.~\ref{fig:nphard:simple:multiclause}
%for the instance $x_1+x_2+x_3=1\wedge \overline{x_1}+\overline{x_2}+\overline{x_4}=1 \wedge \overline{x_1}+x_2+\overline{x_3}=1$.
%The instance consists of blue, red, and gray pixels and the solutions of green, red, and blue cycles.
%An example is given in Fig.~\ref{fig:nphard:simple:multiclause} for the instance $x_1+x_2+x_3=1\wedge \overline{x_1}+\overline{x_2}+\overline{x_4}=1 \wedge \overline{x_1}+x_2+\overline{x_3}=1$.
%The instance consists of blue, red, and gray pixels and the solutions of green, red, and blue cycles.
For every variable we have a \variableconstrsymbol~gadget consisting of a gray \variablepoundsymbol~gadget and a zig-zagging, high-cost 
path of blue pixels. 
A cheap solution traverses a blue path once and connect the ends through the remaining construction of gray and red pixels.
Such {\em variable cycles} (highlighted in red) must either go through the upper (\variabletruesymbol) 
or lower (\variablefalsesymbol) lane of the variable gadget;
the former corresponds to a {\tt true}, the later to a {\tt false} assignment of the corresponding variable.
A {\em clause gadget} modifies a lane of all three involved variable gadgets.
This involves the gray pixels that are covered by the green cycles; we can show that they do
not interfere with the cycles for covering the blue and red pixels, and cannot
be modified to cover them.  Thus, we only have
to cover red and blue pixels, but can pass over gray pixels, too.

To this end, we must connect the ends of the blue paths; 
as it turns out, the formula is satisfiable if and only if we can perform this connection
in a manner that also covers one corresponding red pixel with at most two extra turns. 

\definecolor{ipe-red}{rgb}{1 0 0}
\newcommand{\fillfieldred}[0]{red!80}
\definecolor{ipe-gray}{gray}{0.745}
\newcommand{\fillfieldgray}[0]{ipe-gray}
\definecolor{ipe-blue}{rgb}{0,0,1}
\newcommand{\fillfieldblue}[0]{ipe-blue}
%\newcommand{\fillfieldblue}[0]{blue!30}

%We provide additional details of the proof of Theorem~\ref{th:gg:hard:nphard}.

\ifx\useclauseinsteadofexample\undefined
\ifx\fignphardsimpleincluded\undefined
\def\fignphardsimpleincluded{}

\begin{figure}
	\centering
	\includegraphics[width=0.32\columnwidth]{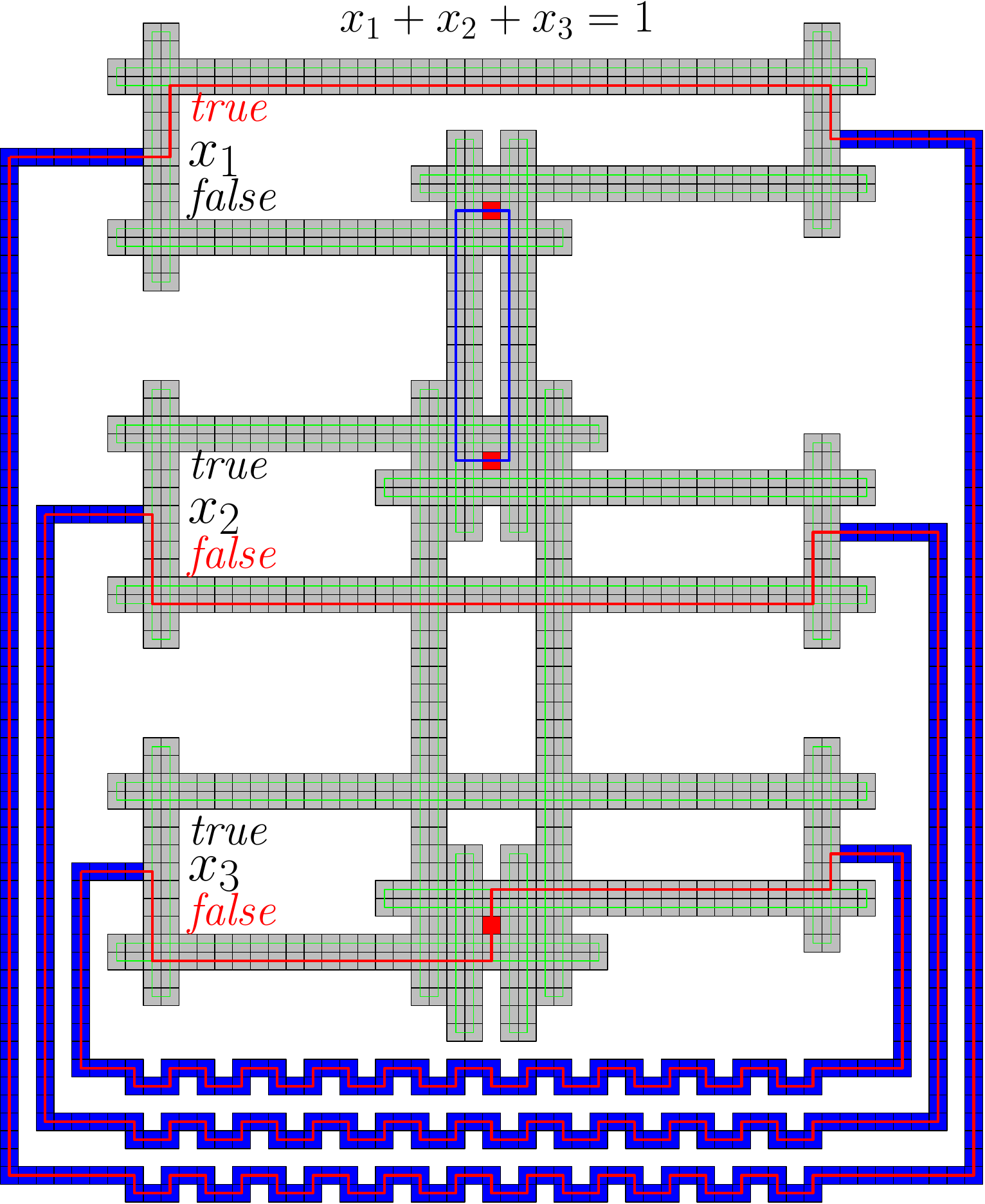}
	\includegraphics[width=0.32\columnwidth]{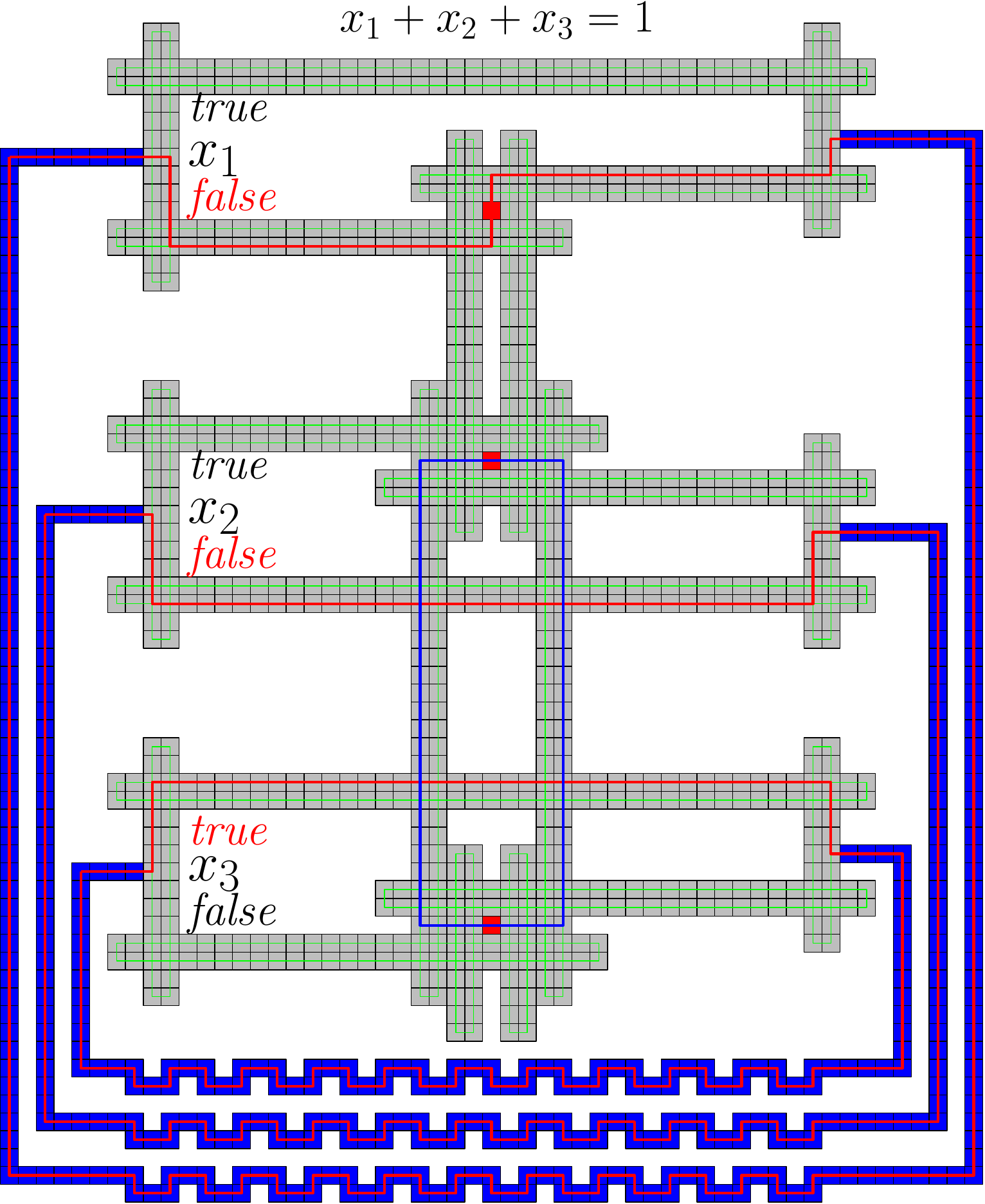}
	\includegraphics[width=0.32\columnwidth]{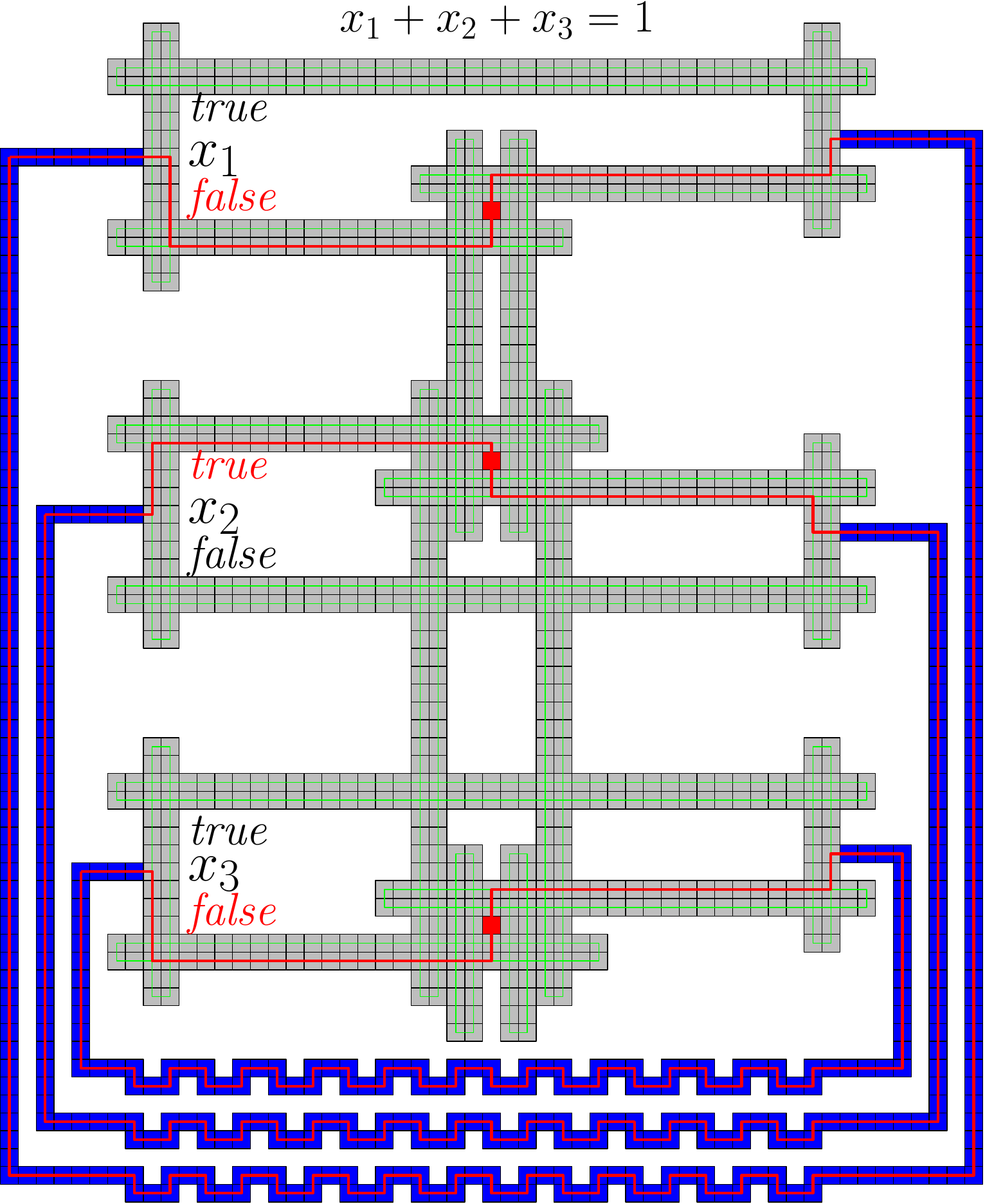}
	\caption{Construction for the one-clause formula $x_1+x_2+x_3=1$ and three possible solutions.
	Every variable has a cycle traversing a zig-zagging path of blue pixels.
	A variable is {\tt true} if its cycle uses the upper path % (\variabletruesymbol)
	through green/red pixels, {\tt false} if it takes the lower path. % (\variablefalsesymbol).
	For covering the red pixels, we may use two additional turns. This results in three classes of optimal cycle covers, shown above.
	If we use the blue 4-turn cycle to cover the upper two red pixels, we are forced to cover the lower red pixel by the $x_3$ variable cycle, setting $x_3$ to {\tt false}.
	The variable cycles of $x_1$ and $x_2$ take the cheapest paths, setting them to {\tt true} or {\tt false}, respectively.
	The alternative to a blue cycle is to cover all three red pixel by the variable cycles, as in the right solution.}
	\label{fig:nphard:simple:singleclause}
\end{figure}
\fi
%an inclusion guard makes sure to add this figure only the first time
\else
\begin{figure}
	\centering
	\includegraphics[width=\columnwidth]{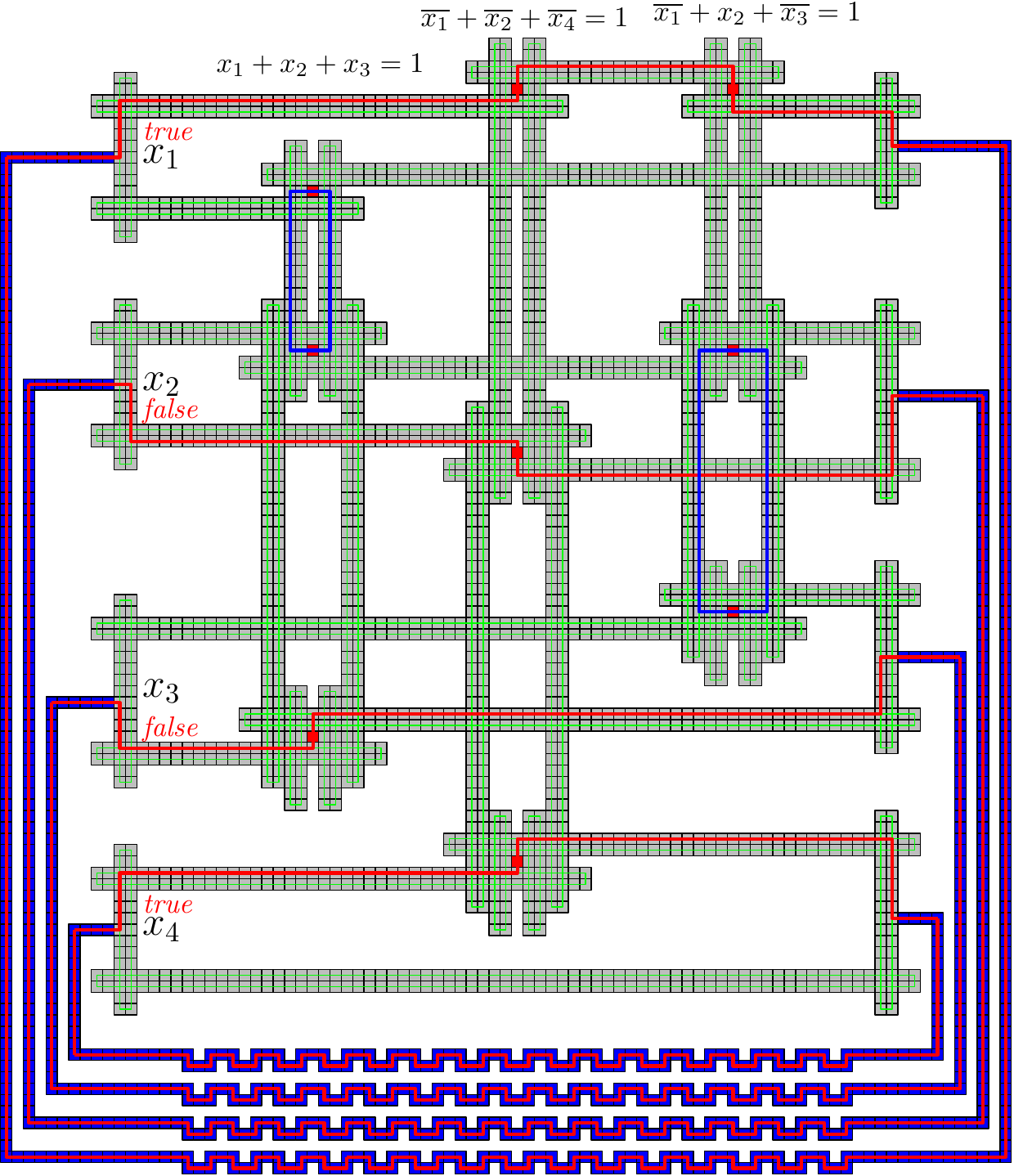}
	\caption{Representing the \emph{1-in-3SAT}-formula $x_1+x_2+x_3=1\wedge \overline{x_1}+\overline{x_2}+\overline{x_4}=1 \wedge \overline{x_1}+x_2+\overline{x_3}=1$.}
	\label{fig:nphard:simple:multiclause}
\end{figure}
\fi

%An example is given in Fig.~\ref{fig:nphard:simple:multiclause} for the instance $x_1+x_2+x_3=1\wedge \overline{x_1}+\overline{x_2}+\overline{x_4}=1 \wedge \overline{x_1}+x_2+\overline{x_3}=1$.
%%The instance consists of blue, red, and gray pixels and the solutions of green, red, and blue cycles.
%For every variable we have a \variableconstrsymbol~gadget consisting of a gray \variablepoundsymbol~gadget and a zig-zagging, high-cost 
%path of blue pixels. 
%A cheap solution will traverse a blue path once and connect the ends through the remaining construction of gray and red pixels.
%Such {\em variable cycles} (highlighted in red) must either go through the upper (\variabletruesymbol) 
%or lower (\variablefalsesymbol) lane of the variable gadget;
%the former corresponds to a {\tt true} assignment, the later to a {\tt false} assignment of the corresponding variable.
%Furthermore, a {\em clause gadget} modifies a lane of all three involved variable gadgets.
%This involves the gray pixels that are covered by the green cycles; we can show that they do
%not interfere with the cycles for covering the blue and red pixels, and cannot
%be modified to cover them.  Hence, we can assume that we only have
%to cover red and blue pixels, but can pass over gray pixels, too.
%
%To this end, we must connect the ends of the blue paths; 
%as it turns out, the formula is satisfiable if and only if we can perform this connection
%in a manner that also covers one corresponding red pixel with at most two extra turns. 

There are three important conclusions.
\begin{enumerate}
	\item The gray pixels are always covered by the green cycles or the solution can be modified without extra turns.
% such that they are if the upper bound for turns in satisfiable constructions
% is kept. 
Thus, we only need to focus on the logic of covering the blue and
 red pixels. This is based on the observation that for keeping tight local
 upper bounds on the number of turns, we are only allowed to make turns at
 very specific locations.
	\item There are only the red variable cycles that cover exactly the red
pixels on the {\tt true} or {\tt false} lane and the simple 4-turn blue cycles
for covering two vertically adjacent red pixels available for covering the red
and blue pixels within the upper bounds of turns for satisfiable constructions.
	\item A coverage with only two turns per red pixel is possible if and only if the underlying formula is satisfiable.
\end{enumerate}
The first item is the most challenging, but also most crucial part.

%\statement{Theorem}{th:gg:hard:nphard}
%{\em	It is NP-hard to find a cycle cover with a minimum number of turns in a 2-dimensional grid graph.}

%----------------------------------------------------------------------------------------
% Construction
%----------------------------------------------------------------------------------------
\subsubsection{Basic Construction}
We use a reduction from \emph{One-in-three 3SAT}.
Refer to Fig.~\ref{fig:nphard:simple:multiclause} for the overall construction. 
Fig.~\ref{fig:rawvariable} shows the building block for a variable gadget, formed by a zig-zagging path, marked in 
blue; this path enforces its covering cycle to go through the box (gray) of the variable, because a double coverage of itself would be too expensive because of the zig-zagging part.
It can be seen that this covering cycle takes the upper path if the corresponding variable is set to {\tt true} and takes the lower path if the variable is set to {\tt false}.

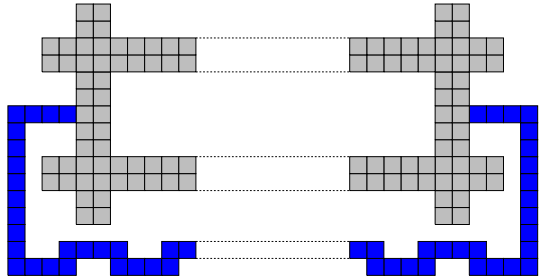
\begin{figure}[h!]
	\centering
		\newcommand{\tikzsquare}[3]{\draw[#3] (#1-0.5,#2-0.5)--(#1+0.5,#2-0.5)--(#1+0.5,#2+0.5)--(#1-0.5,#2+0.5)--(#1-0.5,#2-0.5);}
	\newcommand{\tikzsquares}[5]{ \foreach \a in {#1,...,#3}{ \foreach \b in {#2,...,#4}{ \tikzsquare{\a}{\b}{#5} } } }
	\centering
	\resizebox{0.5\textwidth}{!}{%
		\begin{tikzpicture}
			\tikzsquares{1}{1}{1}{10}{fill=\fillfieldblue}
			\tikzsquares{2}{10}{4}{10}{fill=\fillfieldblue}
			\tikzsquares{2}{1}{4}{1}{fill=\fillfieldblue}
			\tikzsquares{4}{2}{7}{2}{fill=\fillfieldblue};
			\tikzsquares{7}{1}{10}{1}{fill=\fillfieldblue}
			\tikzsquares{10}{2}{11}{2}{fill=\fillfieldblue}

			\tikzsquares{5}{4}{6}{16}{fill=\fillfieldgray}
			\tikzsquares{3}{13}{11}{14}{fill=\fillfieldgray}
			\tikzsquares{3}{6}{11}{7}{fill=\fillfieldgray}

			\tikzsquares{31}{1}{31}{10}{fill=\fillfieldblue}
			\tikzsquares{30}{10}{28}{10}{fill=\fillfieldblue}
			\tikzsquares{30}{1}{28}{1}{fill=\fillfieldblue}
			\tikzsquares{28}{2}{25}{2}{fill=\fillfieldblue};
			\tikzsquares{25}{1}{22}{1}{fill=\fillfieldblue}
			\tikzsquares{22}{2}{21}{2}{fill=\fillfieldblue}

			\tikzsquares{27}{4}{26}{16}{fill=\fillfieldgray}
			\tikzsquares{29}{13}{21}{14}{fill=\fillfieldgray}
			\tikzsquares{29}{6}{21}{7}{fill=\fillfieldgray}

			\draw[dashed, very thick] (11.5, 14.5) -- (20.5, 14.5);
			\draw[dashed, very thick] (11.5, 12.5) -- (20.5, 12.5);

			\draw[dashed, very thick] (11.5, 7.5) -- (20.5, 7.5);
			\draw[dashed, very thick] (11.5, 5.5) -- (20.5, 5.5);

			\draw[dashed, very thick] (11.5, 2.5) -- (20.5, 2.5);
			\draw[dashed, very thick] (11.5, 1.5) -- (20.5, 1.5);
		\end{tikzpicture}
	}%
	\caption[Raw variable]{The construction of a variable gadget. The blue part has to be covered by letting the cycle corresponding to this variable go through the gray part. If the path takes the upper path, the variable is set to {\tt true}; if it takes the lower path, the variable is set to {\tt false}.}
	\label{fig:rawvariable}
\end{figure}

The basic idea for representing literals in clauses is shown in Fig.~\ref{fig:gg:nphard:basic}; this basic element is employed in sets of three for each clause, arranged in a horizontal vector of ``connectors'' to the corresponding variable paths, making use of appropriate geometric extensions to account for (logical) parity, which are described and motivated in the following subsection; of critical importance is the central ``crucial'' pixel marked in red in all figures.  

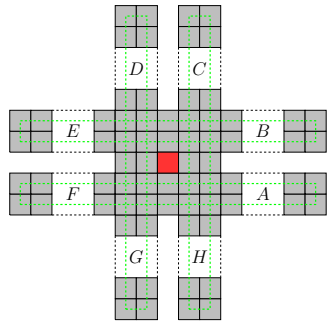
\begin{figure}[h]
	\newcommand{\tikzsquare}[3]{ \draw[#3] (#1-0.5,#2-0.5) -- (#1+0.5,#2-0.5) -- (#1+0.5,#2+0.5) -- (#1-0.5, #2+0.5) -- (#1-0.5,#2-0.5); }
	\newcommand{\tikzsquares}[5]{ \foreach \a in {#1,...,#3}{ \foreach \b in {#2,...,#4}{ \tikzsquare{\a}{\b}{#5} } } }
	\centering
	%\includegraphics[width=.45\columnwidth]{./figures/fig7_basicblock}
	%\old{
	\resizebox{0.3\textwidth}{!}{%
		\begin{tikzpicture}
			\tikzsquares{1}{6}{2}{7}{fill=\fillfieldgray}
			\draw[dashed] (2.5,5.5) -- (4.5, 5.5);
			\node at (3.5, 6.5) {\Huge{$F$}};
			\draw[dashed] (2.5,7.5) -- (4.5, 7.5);
			\tikzsquares{5}{6}{11}{7}{fill=\fillfieldgray}
			\draw[dashed] (11.5,5.5) -- (13.5, 5.5);
			\node at (12.5, 6.5) {\Huge{$A$}};
			\draw[dashed] (11.5,7.5) -- (13.5, 7.5);
			\tikzsquares{14}{6}{15}{7}{fill=\fillfieldgray}

			\tikzsquares{1}{9}{2}{10}{fill=\fillfieldgray}
			\draw[dashed] (2.5,8.5) -- (4.5, 8.5);
			\node at (3.5, 9.5) {\Huge{$E$}};
			\draw[dashed] (2.5,10.5) -- (4.5, 10.5);
			\tikzsquares{5}{9}{11}{10}{fill=\fillfieldgray}
			\draw[dashed] (11.5,8.5) -- (13.5, 8.5);
			\node at (12.5, 9.5) {\Huge{$B$}};
			\draw[dashed] (11.5,10.5) -- (13.5, 10.5);
			\tikzsquares{14}{9}{15}{10}{fill=\fillfieldgray}

			\tikzsquares{6}{1}{7}{2}{fill=\fillfieldgray}
			\draw[dashed] (5.5,2.5) -- (5.5, 4.5);
			\node at (6.5, 3.5) {\Huge{$G$}};
			\draw[dashed] (7.5,2.5) -- (7.5, 4.5);
			\tikzsquares{6}{5}{7}{11}{fill=\fillfieldgray}
			\draw[dashed] (5.5,11.5) -- (5.5, 13.5);
			\node at (6.5, 12.5) {\Huge{$D$}};
			\draw[dashed] (7.5,11.5) -- (7.5, 13.5);
			\tikzsquares{6}{14}{7}{15}{fill=\fillfieldgray}

			\tikzsquares{9}{1}{10}{2}{fill=\fillfieldgray}
			\draw[dashed] (8.5,2.5) -- (8.5, 4.5);
			\node at (9.5, 3.5) {\Huge{$H$}};
			\draw[dashed] (10.5,2.5) -- (10.5, 4.5);
			\tikzsquares{9}{5}{10}{11}{fill=\fillfieldgray}
			\draw[dashed] (8.5,11.5) -- (8.5, 13.5);
			\node at (9.5, 12.5) {\Huge{$C$}};
			\draw[dashed] (10.5,11.5) -- (10.5, 13.5);
			\tikzsquares{9}{14}{10}{15}{fill=\fillfieldgray}

			\tikzsquare{8}{8}{fill=\fillfieldred}

			\draw[green, ultra thick, dashed] (1,6) -- (15,6) -- (15,7) -- (1,7) -- cycle;
			\draw[green, ultra thick, dashed] (1,9) -- (15,9) -- (15,10) -- (1,10) -- cycle;
			\draw[green, ultra thick, dashed] (6,1) -- (7,1) -- (7,15) -- (6,15) -- cycle;
			\draw[green, ultra thick, dashed] (9,1) -- (10,1) -- (10, 15) -- (9, 15) -- cycle;

		\end{tikzpicture}
		%}%end old
	}%
	\caption[Basic logic element]{Basic block for the hardness proof. Only the red ``crucial'' pixel is of algorithmic interest, as described in the text.}
	\label{fig:gg:nphard:basic}
\end{figure}

\subsubsection{Logical Ideas and Details of the Construction}
The key idea of the reduction is that information can only be encoded in pixels for which we do not know in which orientation they are 
crossed without turning (otherwise matching techniques can be used, as utilized by our approximation algorithm).
This means that, e.g., pixels on the boundary are easy to deal with.
Also, pixels for which we know that there is a turn or that are crossed straight in both orientations are not critical:
in a solution, they may only be used indirectly for information propagation.  
Based on these observations, we consider variants of the construction in Fig.~\ref{fig:gg:nphard:basic} as fundamental information elements.
The optimality of the green cycles is argued in the next section, and thus the corresponding pixels are trivially covered.

Covering a crucial pixel is more difficult.
Covering it by an additional cycle or modifying the green cycles costs always at least $4$ turns.
Assume we create a path from one entry ($A,B,C,D,E,F,G,H$) to another and cover the crucial pixel by this path.
For example, a path entering at $A$ and leaving at $E$ would allow us to cover the crucial pixel at a cost of $2$ while a path entering at $A$ and leaving at $D$ has at least $3$ turns if also covering the crucial pixel, see Fig.~\ref{fig:gg:hard:coverred}.
Actually, we can cover the crucial pixel with a cost of $2$ exactly with S- and U-turns, i.e., paths that do not change the orientation.
\begin{figure}
	\newcommand{\tikzsquare}[3]{ \draw[#3] (#1-0.5,#2-0.5) -- (#1+0.5,#2-0.5) -- (#1+0.5,#2+0.5) -- (#1-0.5, #2+0.5) -- (#1-0.5,#2-0.5); }
	\newcommand{\tikzsquares}[5]{ \foreach \a in {#1,...,#3}{ \foreach \b in {#2,...,#4}{ \tikzsquare{\a}{\b}{#5} } } }
	\centering
	%\begin{subfigure}[b]{0.30\textwidth}
		%\subfloat[]{%
		\centering
		%\includegraphics[width=\columnwidth]{./figures/sturn}
		%\old{
		\resizebox{0.3\textwidth}{!}{%
			\begin{tikzpicture}
				\draw[dashed] (3.5,5.5) -- (4.5, 5.5);
				\node at (3.5, 6.5) {\Large{$F$}};
				\draw[dashed] (3.5,7.5) -- (4.5, 7.5);
				\tikzsquares{5}{6}{11}{7}{fill=\fillfieldgray}
				\draw[dashed] (11.5,5.5) -- (12.5, 5.5);
				\node at (12.5, 6.5) {\Large{$A$}};
				\draw[dashed] (11.5,7.5) -- (12.5, 7.5);

				\draw[dashed] (3.5,8.5) -- (4.5, 8.5);
				\node at (3.5, 9.5) {\Large{$E$}};
				\draw[dashed] (3.5,10.5) -- (4.5, 10.5);
				\tikzsquares{5}{9}{11}{10}{fill=\fillfieldgray}
				\draw[dashed] (11.5,8.5) -- (12.5, 8.5);
				\node at (12.5, 9.5) {\Large{$B$}};
				\draw[dashed] (11.5,10.5) -- (12.5, 10.5);

				\draw[dashed] (5.5,3.5) -- (5.5, 4.5);
				\node at (6.5, 3.5) {\Large{$G$}};
				\draw[dashed] (7.5,3.5) -- (7.5, 4.5);
				\tikzsquares{6}{5}{7}{11}{fill=\fillfieldgray}
				\draw[dashed] (5.5,11.5) -- (5.5, 12.5);
				\node at (6.5, 12.5) {\Large{$D$}};
				\draw[dashed] (7.5,11.5) -- (7.5, 12.5);

				\draw[dashed] (8.5,3.5) -- (8.5, 4.5);
				\node at (9.5, 3.5) {\Large{$H$}};
				\draw[dashed] (10.5,3.5) -- (10.5, 4.5);
				\tikzsquares{9}{5}{10}{11}{fill=\fillfieldgray}
				\draw[dashed] (8.5,11.5) -- (8.5, 12.5);
				\node at (9.5, 12.5) {\Large{$C$}};
				\draw[dashed] (10.5,11.5) -- (10.5, 12.5);

				\tikzsquare{8}{8}{fill=\fillfieldred}

				\draw[green, ultra thick, dashed] (3,6) -- (13,6) -- (13,7) -- (3,7) -- cycle;
				\draw[green, ultra thick, dashed] (3,9) -- (13,9) -- (13,10) -- (3,10) -- cycle;
				\draw[green, ultra thick, dashed] (6,3) -- (7,3) -- (7,13) -- (6,13) -- cycle;
				\draw[green, ultra thick, dashed] (9,3) -- (10,3) -- (10, 13) -- (9, 13) -- cycle;

				\draw[red, ultra thick] (4,10) -- (8,10) -- (8,6) -- (12,6);

			\end{tikzpicture}
		}%
		%}%end old
		%\caption{S-Turn}
		%\end{subfigure}
		%}
	~%
	%\begin{subfigure}[b]{0.3\textwidth}
		%\subfloat[]{%
		\centering
		%\includegraphics[width=\columnwidth]{./figures/uturn}
		%\old{
		\resizebox{0.3\textwidth}{!}{%
			\begin{tikzpicture}
				\draw[dashed] (3.5,5.5) -- (4.5, 5.5);
				\node at (3.5, 6.5) {\Large{$F$}};
				\draw[dashed] (3.5,7.5) -- (4.5, 7.5);
				\tikzsquares{5}{6}{11}{7}{fill=\fillfieldgray}
				\draw[dashed] (11.5,5.5) -- (12.5, 5.5);
				\node at (12.5, 6.5) {\Large{$A$}};
				\draw[dashed] (11.5,7.5) -- (12.5, 7.5);

				\draw[dashed] (3.5,8.5) -- (4.5, 8.5);
				\node at (3.5, 9.5) {\Large{$E$}};
				\draw[dashed] (3.5,10.5) -- (4.5, 10.5);
				\tikzsquares{5}{9}{11}{10}{fill=\fillfieldgray}
				\draw[dashed] (11.5,8.5) -- (12.5, 8.5);
				\node at (12.5, 9.5) {\Large{$B$}};
				\draw[dashed] (11.5,10.5) -- (12.5, 10.5);

				\draw[dashed] (5.5,3.5) -- (5.5, 4.5);
				\node at (6.5, 3.5) {\Large{$G$}};
				\draw[dashed] (7.5,3.5) -- (7.5, 4.5);
				\tikzsquares{6}{5}{7}{11}{fill=\fillfieldgray}
				\draw[dashed] (5.5,11.5) -- (5.5, 12.5);
				\node at (6.5, 12.5) {\Large{$D$}};
				\draw[dashed] (7.5,11.5) -- (7.5, 12.5);

				\draw[dashed] (8.5,3.5) -- (8.5, 4.5);
				\node at (9.5, 3.5) {\Large{$H$}};
				\draw[dashed] (10.5,3.5) -- (10.5, 4.5);
				\tikzsquares{9}{5}{10}{11}{fill=\fillfieldgray}
				\draw[dashed] (8.5,11.5) -- (8.5, 12.5);
				\node at (9.5, 12.5) {\Large{$C$}};
				\draw[dashed] (10.5,11.5) -- (10.5, 12.5);

				\tikzsquare{8}{8}{fill=\fillfieldred}

				\draw[green, ultra thick, dashed] (3,6) -- (13,6) -- (13,7) -- (3,7) -- cycle;
				\draw[green, ultra thick, dashed] (3,9) -- (13,9) -- (13,10) -- (3,10) -- cycle;
				\draw[green, ultra thick, dashed] (6,3) -- (7,3) -- (7,13) -- (6,13) -- cycle;
				\draw[green, ultra thick, dashed] (9,3) -- (10,3) -- (10, 13) -- (9, 13) -- cycle;

				\draw[red, ultra thick] (12,10) -- (8,10) -- (8,6) -- (12,6);

			\end{tikzpicture}
		}%
		%}%end old
		%\caption{U-Turn}
	%\end{subfigure}
	%}
~%
	%\begin{subfigure}[b]{0.3\textwidth}
		%\subfloat[]{
		\centering
		%\includegraphics[width=\columnwidth]{./figures/badturn}
		%\old{
		\resizebox{0.3\textwidth}{!}{%
			\begin{tikzpicture}
				\draw[dashed] (3.5,5.5) -- (4.5, 5.5);
				\node at (3.5, 6.5) {\Large{$F$}};
				\draw[dashed] (3.5,7.5) -- (4.5, 7.5);
				\tikzsquares{5}{6}{11}{7}{fill=\fillfieldgray}
				\draw[dashed] (11.5,5.5) -- (12.5, 5.5);
				\node at (12.5, 6.5) {\Large{$A$}};
				\draw[dashed] (11.5,7.5) -- (12.5, 7.5);

				\draw[dashed] (3.5,8.5) -- (4.5, 8.5);
				\node at (3.5, 9.5) {\Large{$E$}};
				\draw[dashed] (3.5,10.5) -- (4.5, 10.5);
				\tikzsquares{5}{9}{11}{10}{fill=\fillfieldgray}
				\draw[dashed] (11.5,8.5) -- (12.5, 8.5);
				\node at (12.5, 9.5) {\Large{$B$}};
				\draw[dashed] (11.5,10.5) -- (12.5, 10.5);

				\draw[dashed] (5.5,3.5) -- (5.5, 4.5);
				\node at (6.5, 3.5) {\Large{$G$}};
				\draw[dashed] (7.5,3.5) -- (7.5, 4.5);
				\tikzsquares{6}{5}{7}{11}{fill=\fillfieldgray}
				\draw[dashed] (5.5,11.5) -- (5.5, 12.5);
				\node at (6.5, 12.5) {\Large{$D$}};
				\draw[dashed] (7.5,11.5) -- (7.5, 12.5);

				\draw[dashed] (8.5,3.5) -- (8.5, 4.5);
				\node at (9.5, 3.5) {\Large{$H$}};
				\draw[dashed] (10.5,3.5) -- (10.5, 4.5);
				\tikzsquares{9}{5}{10}{11}{fill=\fillfieldgray}
				\draw[dashed] (8.5,11.5) -- (8.5, 12.5);
				\node at (9.5, 12.5) {\Large{$C$}};
				\draw[dashed] (10.5,11.5) -- (10.5, 12.5);

				\tikzsquare{8}{8}{fill=\fillfieldred}

				\draw[green, ultra thick, dashed] (3,6) -- (13,6) -- (13,7) -- (3,7) -- cycle;
				\draw[green, ultra thick, dashed] (3,9) -- (13,9) -- (13,10) -- (3,10) -- cycle;
				\draw[green, ultra thick, dashed] (6,3) -- (7,3) -- (7,13) -- (6,13) -- cycle;
				\draw[green, ultra thick, dashed] (9,3) -- (10,3) -- (10, 13) -- (9, 13) -- cycle;

				\draw[red, ultra thick] (4,10) -- (8,10) -- (8,6) -- (9,6) -- (9,4);

			\end{tikzpicture}
		}%
		%}%end old
		%\caption{Bad turn}
	%\end{subfigure}
	%}
\caption[Coverage properties of basic logic element]{Three ways of covering the crucial pixel. The first two keep the orientation and have a cost of $2$. The third one changes the orientation, and thus has a cost of at least $3$.}
\label{fig:gg:hard:coverred}
\end{figure}
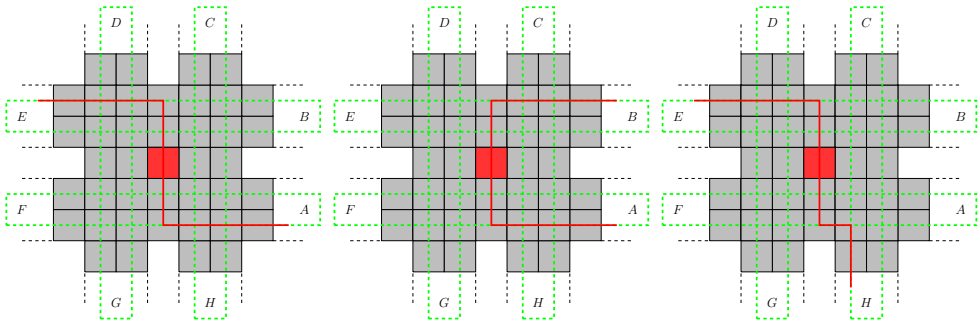
In addition, it can be seen that the crucial pixel always needs at least two turns within this construction.

Now consider the construction shown in Fig.~\ref{fig:gg:hard:triple} that uses a combination of elements as described  above.
Assume all turns made outside the construction are free and all six ``exits'' $A,B,C,D,E,F$ are connected from outside.
This means that every red highlighted pixel can be covered by a cycle of cost two that traverses the exits above and below (see blue paths in Fig.~\ref{fig:gg:hard:triple}).
The green cycles are still necessary; employing them for covering the crucial pixel is too expensive, so 
we can concentrate on how to cover the red highlighted pixels by additional cycles.
It can be seen that there are only five potential cycles that involve not more than a cost of two turns per crucial pixel:
these are the three mentioned red cycles that use the exterior and the two
interior cycles marked in blue in Fig.~\ref{fig:gg:hard:triple} (the cycles can
be slightly shifted without changing them in a useful manner).  This results in
only three potential optimal solutions (refer to
Fig.~\ref{fig:nphard:simple:singleclause}):
\begin{enumerate}
	\item Using the upper red cycle to cover the upper crucial pixel and the lower blue cycle to cover the middle crucial pixel and lower crucial pixel.
	\item Using the lower red cycle to cover the lower crucial pixel and the upper blue cycle to cover the middle crucial pixel and the upper crucial pixel.
	\item Using all three red cycles to cover each crucial pixel separately.
\end{enumerate}

For constructing a clause $x_1+x_2+x_3=1$ with $x_1$ being above $x_2$ and $x_3$ being below $x_2$ we now place the upper crucial pixel of Fig.~\ref{fig:gg:hard:triple} on the {\tt false} path of $x_1$, the middle pixel on the {\tt true} path of $x_2$ and the lower crucial pixel on the {\tt false} path of $x_3$ as illustrated in Fig.~\ref{fig:nphard:simple:singleclause}.
If we take the {\tt true} path of $x_2$, we now have to take the {\tt false} paths of $x_1$ and $x_3$.
If we take the {\tt false} path of $x_2$ we need to block the {\tt false} path of either $x_1$ or $x_3$ (but not both) by a red cycle.
The other variable paths (without the crucial pixel) do not need any turns and are hence always preferred if the red pixel does not have to be covered.
Note that any external cycle that enters the construction (e.g., through $A$) needs at least two turns to leave it again through another exit.
Hence, passing through this construction should always be combined with
covering an uncovered red pixel, for which we have only one choice with two
turns, in order to account for the two necessary turns.  By concatenating these
clauses, we can form any  \emph{One-in-three 3SAT} formula, see
Fig.~\ref{fig:nphard:simple:multiclause}.

\begin{figure}[h!]
	\centering
	\includegraphics[width=0.35\textwidth]{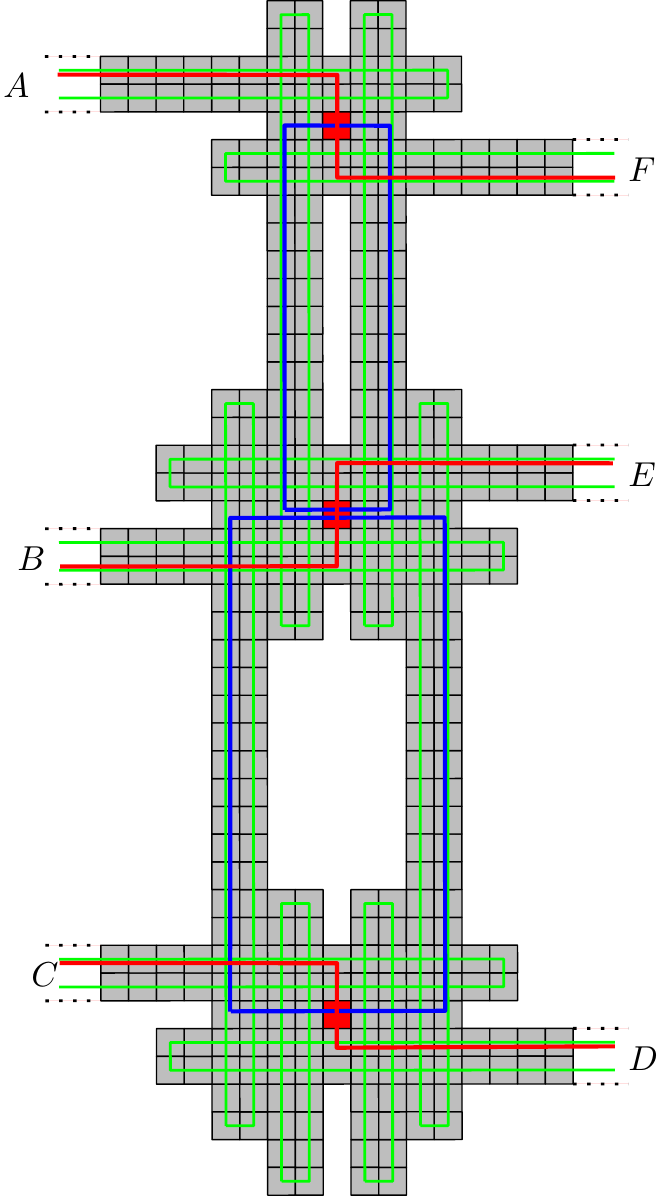}
	\caption[Building a formula]{
		Assuming that all turns made outside are free, there are only five potential cycles to cover the crucial pixels: the three exterior red cycles and the two interior blue cycles.
		Assigning the usage of the upper red cycle the Boolean variable $x_1$ (and analogous for the middle and the lower $x_2$ and $x_3$, resp.)  leads to the Boolean formula $(x_1 \vee x_3) \wedge (x_2 \leftrightarrow x_1 \wedge x_3)$ or $\overline{x_1}+x_2+\overline{x_3}=1$.
	}
	\label{fig:gg:hard:triple}
\end{figure}

%----------------------------------------------------------------------------------------
% A Tight Lower Bound
%----------------------------------------------------------------------------------------
\subsubsection{A Tight Lower Bound}

We now give a tight lower bound on the number of turns in a cycle cover.
In order to achieve it, the turns have to be made at very specific positions.
For this purpose it is sufficient to limit our view onto local components and state that in each of these disjunct components a specific number of turns has to be made.
The global lower bound is then the sum of all local lower bounds.
Later we show that these turns are only sufficient if the corresponding formula is satisfiable.

%We first provide a trivial but important observation and two simple lemmas that help us to deduce necessary turn positions.
%\begin{observation}
%	\label{observation:gg:hard:corver}
%	There has to be a turn on every corner pixel.
%\end{observation}
\begin{lemma}
	\label{lemma:gg:hard:stripeven}
	On every \emph{full strip} (i.e., a maximal connected collinear set of pixels) in the grid graph, a cycle cover has to have an even number of turns (u-turns counting twice).
\end{lemma}
\begin{proof}
	For every turn that enters the strip, there needs to be a turn for leaving it.
	Turns that do not enter/leave the strip are u-turns.
\end{proof}
\begin{lemma}
	\label{lemma:gg:hard:guard}
	Given any pixel $p$ of a grid graph and a corresponding cycle cover $\mathcal{C}$.
	If there is no turn on $p$, then there has to be a turn left of $p$ and another turn right of $p$ (on the horizontal full strip through $p$), or there has to be a turn above $p$ and another below $p$ (on the vertical full strip through $p$).
	If there is a turn at $p$, there are at least three turns on the star consisting of the horizontal and vertical full strip through $p$.
\end{lemma}
\begin{proof}
	If there is no turn on $p$, there is a straight part of a cycle going through it.
	At either end of this straight part, it needs to turn in order to close the cycle.
	If there is a turn at $p$, the claim follows by Lemma~\ref{lemma:gg:hard:stripeven}.
\end{proof}

Now we can deduce the turn positions for all essential parts of the construction, as stated in the next lemma.
\begin{lemma}
	\label{lemma:gg:hard:green}
	In Fig.\ref{fig:gg:hard:partslb}, the essential parts of the construction are displayed (possibly reflected). For these parts we can give the following lower bounds on the necessary number of interior turns and constraints on the positions of turns if this bound is tight:
	 In a part as shown in Fig.~\ref{fig:gg:hard:partslb:2}/\ref{fig:gg:hard:partslb:4}/\ref{fig:gg:hard:partslb:5}/\ref{fig:gg:hard:partslb:6} (excluding the blue area), we need at least 10/18/14/10 interior turns. To achieve exactly this lower bound, there has to be exactly one turn at every black dot and the remaining two turns have to be in the green area.
	\begin{figure}
		\newcommand{\tikzsquare}[3]{ \draw[#3] (#1-0.5,#2-0.5) -- (#1+0.5,#2-0.5) -- (#1+0.5,#2+0.5) -- (#1-0.5, #2+0.5) -- (#1-0.5,#2-0.5); }
		\newcommand{\tikzsquares}[5]{ \foreach \a in {#1,...,#3}{ \foreach \b in {#2,...,#4}{ \tikzsquare{\a}{\b}{#5} } } }
		\newcommand{\tikzforcedturn}[2]{ \draw[fill=black] (#1, #2) circle (0.3);}

			\centering
			%\old{

		~%
		%\subfloat[]{%
		%\begin{subfigure}[b]{0.21\textwidth}
			\centering
			\resizebox{0.3\textwidth}{!}{%
				\begin{tikzpicture}
					\tikzsquares{1}{6}{11}{7}{fill=\fillfieldgray}
					\draw[dashed, ultra thick] (0.5, 5.5) -- (-1, 5.5);
					\draw[dashed, ultra thick] (0.5, 7.5) -- (-1, 7.5);
					\tikzsquares{3}{3}{13}{4}{fill=\fillfieldgray}
					\draw[dashed, ultra thick] (13.5, 4.5) -- (15, 4.5);
					\draw[dashed, ultra thick] (13.5, 2.5) -- (15, 2.5);
					\tikzsquares{5}{1}{6}{11}{fill=\fillfieldgray}
					\draw[dashed, ultra thick](4.5, 11.5) -- (4.5, 13);
					\draw[dashed, ultra thick] (6.5, 11.5) -- (6.5, 13);
					\tikzsquares{8}{1}{9}{11}{fill=\fillfieldgray}
					\draw[dashed, ultra thick](7.5, 11.5) -- (7.5, 13);
					\draw[dashed, ultra thick] (9.5, 11.5) -- (9.5, 13);
					\tikzsquares{7}{3}{7}{7}{fill=green!80}
					\tikzsquares{5}{5}{9}{5}{fill=green!80}
					\tikzsquare{7}{5}{fill=\fillfieldred};
					\tikzforcedturn{5}{1}
					\tikzforcedturn{6}{1}
					\tikzforcedturn{8}{1}
					\tikzforcedturn{9}{1}
					\tikzforcedturn{3}{3}
					\tikzforcedturn{3}{4}
					\tikzforcedturn{11}{6}
					\tikzforcedturn{11}{7}
				\end{tikzpicture}
			}%
			\caption{\num{8} unique turns and \num{2} turns in green area}
				\label{fig:gg:hard:partslb:2}
		%\end{subfigure}
		%}
		~%
			%\subfloat[]{
				%\begin{subfigure}[b]{0.21\textwidth}
				\centering
				\resizebox{0.3\textwidth}{!}{%
					\begin{tikzpicture}
						\tikzsquares{1}{5}{15}{6}{fill=\fillfieldgray}
						\draw[dashed, ultra thick] (0.5, 4.5) -- (-1, 4.5);
						\draw[dashed, ultra thick] (0.5, 6.5) -- (-1, 6.5);
						\tikzsquares{3}{8}{17}{9}{fill=\fillfieldgray}
						\draw[dashed, ultra thick] (17.5, 7.5) -- (19, 7.5);
						\draw[dashed, ultra thick] (17.5, 9.5) -- (19, 9.5);
						\tikzsquares{5}{3}{6}{13}{fill=\fillfieldgray}
						\draw[dashed, ultra thick] (4.5, 13.5) -- (4.5, 16);
						\draw[dashed, ultra thick] (6.5, 13.5) -- (6.5, 16);
						\tikzsquares{7}{1}{8}{11}{fill=\fillfieldgray}
						\tikzsquares{10}{1}{11}{11}{fill=\fillfieldgray}
						\tikzsquares{12}{3}{13}{13}{fill=\fillfieldgray}
						\draw[dashed, ultra thick] (11.5, 13.5) -- (11.5, 16);
						\draw[dashed, ultra thick] (13.5, 13.5) -- (13.5, 16);
						\tikzsquares{9}{5}{9}{9}{fill=green!80}
						\tikzsquares{5}{7}{6}{7}{fill=green!80}
						\tikzsquares{12}{7}{13}{7}{fill=green!80}
						\tikzsquare{9}{7}{fill=\fillfieldred}

						\tikzforcedturn{15}{6}
						\tikzforcedturn{15}{5}
						\tikzforcedturn{3}{8}
						\tikzforcedturn{3}{9}
						\tikzforcedturn{5}{3}
						\tikzforcedturn{6}{3}
						\tikzforcedturn{7}{1}
						\tikzforcedturn{7}{11}
						\tikzforcedturn{8}{11}
						\tikzforcedturn{8}{1}
						\tikzforcedturn{10}{1}
						\tikzforcedturn{10}{11}
						\tikzforcedturn{11}{11}
						\tikzforcedturn{11}{1}
						\tikzforcedturn{12}{3}
						\tikzforcedturn{13}{3}

					\end{tikzpicture}
				}%
				\caption{\num{16} unique turns and \num{2} turns in green area}
				\label{fig:gg:hard:partslb:4}
			%\end{subfigure}
			%}
		~%
		%\subfloat[]{
		%\begin{subfigure}[b]{0.21\textwidth}
			\centering
			\resizebox{0.3\textwidth}{!}{%
				\begin{tikzpicture}
					\tikzsquares{5}{1}{6}{12}{fill=\fillfieldgray}
					\draw[dashed, ultra thick] (4.5, 0.5) -- (4.5, -1);
					\draw[dashed, ultra thick] (6.5, 0.5) -- (6.5, -1);
					\tikzsquares{12}{1}{13}{12}{fill=\fillfieldgray}
					\draw[dashed, ultra thick] (11.5, 0.5) -- (11.5, -1);
					\draw[dashed, ultra thick] (13.5, 0.5) -- (13.5, -1);
					\tikzsquares{7}{4}{8}{15}{fill=\fillfieldgray}
					\draw[dashed, ultra thick] (6.5, 15.5) -- (6.5, 17);
					\draw[dashed, ultra thick] (8.5, 15.5) -- (8.5, 17);
					\tikzsquares{10}{4}{11}{15}{fill=\fillfieldgray}
					\draw[dashed, ultra thick] (9.5, 15.5) -- (9.5, 17);
					\draw[dashed, ultra thick] (11.5, 15.5) -- (11.5, 17);
					\tikzsquares{3}{6}{17}{7}{fill=\fillfieldgray}
					\draw[dashed, ultra thick] (17.5, 5.5) -- (19, 5.5);
					\draw[dashed, ultra thick] (17.5, 7.5) -- (19, 7.5);
					\tikzsquares{1}{9}{15}{10}{fill=\fillfieldgray}
					\draw[dashed, ultra thick] (0.5, 8.5) -- (-1, 8.5);
					\draw[dashed, ultra thick] (0.5, 10.5) -- (-1, 10.5);

					\tikzsquares{9}{6}{9}{10}{fill=green!80}
					\tikzsquares{5}{8}{13}{8}{fill=green!80}
					\tikzsquare{9}{8}{fill=\fillfieldred}

					\tikzforcedturn{6}{12}
					\tikzforcedturn{5}{12}
					\tikzforcedturn{13}{12}
					\tikzforcedturn{12}{12}
					\tikzforcedturn{7}{4}
					\tikzforcedturn{8}{4}
					\tikzforcedturn{10}{4}
					\tikzforcedturn{11}{4}
					\tikzforcedturn{3}{6}
					\tikzforcedturn{3}{7}
					\tikzforcedturn{15}{10}
					\tikzforcedturn{15}{9}
				\end{tikzpicture}
			}%
			\caption{\num{12} unique turns and \num{2} turns in green area}
				\label{fig:gg:hard:partslb:5}
		%\end{subfigure}
		%}
				~%
		%\begin{subfigure}[b]{0.21\textwidth}
			%\subfloat[]{
			\centering
			\resizebox{0.3\textwidth}{!}{%
				\begin{tikzpicture}
					\tikzsquares{1}{1}{1}{10}{fill=\fillfieldblue}
					\tikzsquares{2}{10}{4}{10}{fill=\fillfieldblue}
					\tikzsquares{2}{1}{4}{1}{fill=\fillfieldblue}
					\tikzsquares{4}{2}{7}{2}{fill=\fillfieldblue};
					\tikzsquares{7}{1}{10}{1}{fill=\fillfieldblue}
					\tikzsquares{10}{2}{11}{2}{fill=\fillfieldblue}

					\tikzsquares{5}{4}{6}{16}{fill=\fillfieldgray}
					\tikzsquares{3}{13}{11}{14}{fill=\fillfieldgray}
					\tikzsquares{3}{6}{11}{7}{fill=\fillfieldgray}

					\draw[dashed, ultra thick] (11.5, 14.5) -- (13, 14.5);
					\draw[dashed, ultra thick] (11.5, 12.5) -- (13, 12.5);

					\draw[dashed, ultra thick] (11.5, 7.5) -- (13, 7.5);
					\draw[dashed, ultra thick] (11.5, 5.5) -- (13, 5.5);

					\draw[dashed, ultra thick] (11.5, 2.5) -- (13, 2.5);
					\draw[dashed, ultra thick] (11.5, 1.5) -- (13, 1.5);

					\tikzsquares{5}{13}{6}{14}{fill=green!80}
					\tikzsquares{5}{10}{6}{10}{fill=green!80}
					\tikzsquares{5}{6}{6}{7}{fill=green!80}

					\tikzforcedturn{5}{4}
					\tikzforcedturn{5}{16}
					\tikzforcedturn{6}{16}
					\tikzforcedturn{6}{4}
					\tikzforcedturn{3}{13}
					\tikzforcedturn{3}{14}
					\tikzforcedturn{3}{6}
					\tikzforcedturn{3}{7}
				\end{tikzpicture}
			}%
			\caption{\num{8} unique turns and \num{2} turns in green area.}
				\label{fig:gg:hard:partslb:6}
		%\end{subfigure}
		%}
		%}%end old
		%\includegraphics[width=.95\columnwidth]{./figures/fig12_Lowerbounds}
	\caption[Necessary turns]{Lower bounds on the number of turns in the different parts of the construction (possibly reflected). The black dots represent turn position necessary to meet the lower bound.}
	\label{fig:gg:hard:partslb}
\end{figure}
%This implies at least additional $42$ turns per clause and $20$ per variable plus one turn at each corner pixel in the blue highlighted area.

\end{lemma}
\begin{proof}
	There needs to be at least one turn in every corner pixel, making Fig.~\ref{fig:gg:hard:partslb:2} and \ref{fig:gg:hard:partslb:6} trivial.
	In Fig.~\ref{fig:gg:hard:partslb:4} and \ref{fig:gg:hard:partslb:5} we have some additional non-corner pixels for which we need easy additional arguments based on Lemma~\ref{lemma:gg:hard:stripeven} and~\ref{lemma:gg:hard:guard}.

	We give the details for the part in Fig.~\ref{fig:gg:hard:partslb:4}; the adaption to Fig.~\ref{fig:gg:hard:partslb:5} is straightforward. 
	The used arguments are highlighted in Fig.~\ref{fig:gg:hard:green}.
	First, there are already \num{12} corner pixels, so we have only \num{6} turns left to cover this local part.
	Note that we assume that we can appropriately place turns outside this
local part, so we are using only local arguments and the turn constraints
remain feasible in the global construction.  We need at least 6 additional
turns: one on each yellow area due to Lemma~\ref{lemma:gg:hard:stripeven} and
two on the green area due to Lemma~\ref{lemma:gg:hard:guard}.
	There can be no turn on the red pixel; otherwise we would need three of the six turns on the green and red pixels, contradicting the fact
that we already need at least four in the yellow locations.
	In order to remain locally optimal, turns on the gray and red pixel thus are impossible; however, these pixels still need to be covered,
	in particular the pixels marked by hollow circles.
	Lemma~\ref{lemma:gg:hard:guard} applied on these pixel forces us to place a turn on the adjacent yellow pixel above or below them, leaving only two further turns for covering the red pixel.
	Lemma~\ref{lemma:gg:hard:stripeven} additionally forbids turns on two pixels to the left and right of the red pixel: We need the ability to make turns outside in order to fulfill it.

	\begin{figure}[h]
		\centering
		%\old{
		\resizebox{0.4\textwidth}{!}{%
			\begin{tikzpicture}
				\newcommand{\tikzsquare}[3]{ \draw[#3] (#1-0.5,#2-0.5) -- (#1+0.5,#2-0.5) -- (#1+0.5,#2+0.5) -- (#1-0.5, #2+0.5) -- (#1-0.5,#2-0.5); }
				\newcommand{\tikzsquares}[5]{ \foreach \a in {#1,...,#3}{ \foreach \b in {#2,...,#4}{ \tikzsquare{\a}{\b}{#5} } } }
				\newcommand{\tikzforcedturn}[2]{ \draw[fill=black] (#1, #2) circle (0.3);}
					\tikzsquares{1}{5}{15}{6}{fill=\fillfieldgray}
					\draw[dashed, ultra thick] (0.5, 4.5) -- (-1, 4.5);
					\draw[dashed, ultra thick] (0.5, 6.5) -- (-1, 6.5);
					\tikzsquares{3}{8}{17}{9}{fill=\fillfieldgray}
					\draw[dashed, ultra thick] (17.5, 7.5) -- (19, 7.5);
					\draw[dashed, ultra thick] (17.5, 9.5) -- (19, 9.5);
					\tikzsquares{5}{3}{6}{13}{fill=\fillfieldgray}
					\draw[dashed, ultra thick] (4.5, 13.5) -- (4.5, 16);
					\draw[dashed, ultra thick] (6.5, 13.5) -- (6.5, 16);
					\tikzsquares{7}{1}{8}{11}{fill=\fillfieldgray}
					\tikzsquares{10}{1}{11}{11}{fill=\fillfieldgray}
					\tikzsquares{12}{3}{13}{13}{fill=\fillfieldgray}
					\draw[dashed, ultra thick] (11.5, 13.5) -- (11.5, 16);
					\draw[dashed, ultra thick] (13.5, 13.5) -- (13.5, 16);
					\tikzsquares{9}{5}{9}{9}{fill=green!80}
					\tikzsquares{5}{7}{13}{7}{fill=green!80}
					\tikzsquare{9}{7}{fill=\fillfieldred}

					%At least one due to even #turns on every max strip
				\tikzsquares{5}{3}{8}{3}{fill=yellow}
				\tikzsquares{5}{11}{8}{11}{fill=yellow}
				\tikzsquares{10}{3}{13}{3}{fill=yellow}
				\tikzsquares{10}{11}{13}{11}{fill=yellow}

					\tikzforcedturn{15}{6}
					\tikzforcedturn{15}{5}
					\tikzforcedturn{3}{8}
					\tikzforcedturn{3}{9}
					\tikzforcedturn{5}{3}
					\tikzforcedturn{7}{1}
					\tikzforcedturn{8}{11}
					\tikzforcedturn{8}{1}
					\tikzforcedturn{10}{1}
					\tikzforcedturn{10}{11}
					\tikzforcedturn{11}{1}
					\tikzforcedturn{13}{3}

					\newcommand{\tikzguard}[2]{\draw[very thick] (#1, #2) circle (0.3);}
					\tikzguard{12}{4}
					\tikzguard{6}{4}
					\tikzguard{11}{10}
					\tikzguard{7}{10}

			\end{tikzpicture}
			}%

			\caption[Necessary turns proof example]{Auxiliary figure for the proof of Lemma~\ref{lemma:gg:hard:green}.}
		\label{fig:gg:hard:green}
	\end{figure}
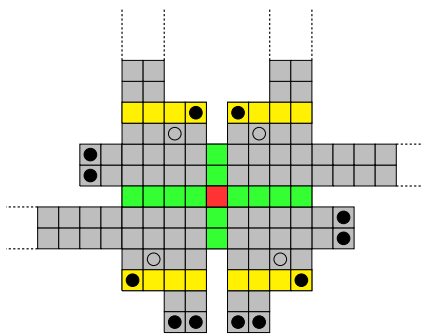

\end{proof}

%We now can state the lower bound:
\begin{lemma}
	\label{th:gg:hard:lb}
	Given a construction, as described above, for an arbitrary formula with $v$ variables and $c$ clauses.
	Let $k$ be the number of corner pixels in the zig-zagging paths of the variable gadgets (highlighted in blue).
	Then every cycle cover has to have at least $20v+42c+k$ turns.
\end{lemma}
\begin{proof}
	There is no interference between these local bounds, because in the
	proof we did not charge for exterior (not inside considered local part) turns
	made by possible cycles.  
	Thus, the claim follows by summing over the local bounds for all local parts using Lemma~\ref{lemma:gg:hard:green}.
\end{proof}

%----------------------------------------------------------------------------------------
% Upper Bound for Satisfiable 
%----------------------------------------------------------------------------------------
\subsubsection{An Upper Bound for Satisfiable Formulas}

We can construct a solution that matches the lower bound given in Lemma~\ref{th:gg:hard:lb} if the corresponding formula of the grid-graph construction is satisfiable.

\begin{lemma}
	\label{th:gg:hard:satub}
	Given a construction for an arbitrary formula with $v$ variables and $c$ clauses, as described above.
	Let $k$ be the number of corner pixels in the zig-zagging paths of the variable gadgets (highlighted in blue).
	If the formula is satisfiable, there exists a cycle cover with $20v+42c+k$ turns.
\end{lemma}
\begin{proof}
	It is straightforward to deduce this from the description of the construction and Fig.~\ref{fig:nphard:simple:multiclause}, so
	we only sketch the details. 
	Cover the gray pixels using only simple green 4-turn cycles on the explicit turns (black dots), as shown in Fig.~\ref{fig:gg:hard:partslb}.
	Select an arbitrary but fixed satisfiable solution for the formula.
	If a variable assignment is {\tt true}, add a cycle that traverses all blue pixels of the variable and all upper crucial pixels of the variable with a minimum number of turns. 
	This cycle has a turn at every blue corner pixel belonging to the variable, two turns per crucial pixel, and two further turns at each end to connect to the blue pixels.
	If a variable assignment is {\tt false}, proceed analogously, but for the lower row of crucial pixels belonging to this variable.
	As the assignment is feasible, there are either none or exactly two remaining crucial red pixels per clause gadget.
	If there are two, they are vertically adjacent and we can cover them by a simple blue 4-turn cycle.
	As Lemma~\ref{lemma:gg:hard:green} leaves us with exactly $2$ turns to select per crucial pixel, this matches the lower bound.
\end{proof}

%----------------------------------------------------------------------------------------
% Lower Bound for Unsatisfiable 
%----------------------------------------------------------------------------------------
\subsubsection{Lower Bound for Unsatisfiable Formulas}

If the formula corresponding to the grid graph construction is not satisfiable, we can deduce that the number of turns given by Lemma~\ref{th:gg:hard:lb} is not sufficient for a full coverage.
We first show that in a cycle cover that matches the lower bound, we can separate a cycle cover for the gray pixels from the rest.

\begin{lemma}
	\label{lemma:gg:hard:separategray}
	Given an optimal cycle cover that matches the lower bound of Lemma~\ref{th:gg:hard:lb}, 
	then we can modify the solution (without increasing the cost) such that it contains a cycle cover of exactly the gray pixel (the green cycles in Fig.~\ref{fig:nphard:simple:multiclause}).
	%Then we can modify the solution such that the gray pixels are covered by a set of cycles that uses exactly the black dotted pixels in Fig.~\ref{fig:gg:hard:partslb}.
\end{lemma}
\begin{proof}
	The cycle cover matches the lower bound of Lemma~\ref{th:gg:hard:lb} and hence has to fulfill the restriction on the position of turns as in Lemma~\ref{lemma:gg:hard:green}.
	We not only know that there has to be a turn at the exact same locations as the green cycles in Fig.~\ref{fig:nphard:simple:multiclause}, marked by black dots in Fig.~\ref{fig:gg:hard:partslb}, but also that these turns have to be exactly the same simply due to a lack of alternatives: if there are only potential partner turns in two directions, we can only make a turn between these two (u-turns are already prohibited by Lemma~\ref{th:gg:hard:lb}).
	If two such partnered turns are not directly connected, we known that there are exactly two turns on pixel that are highlighted in green in Fig.~\ref{fig:gg:hard:partslb} in between them.
	We can simply connect these two turns on green pixels directly as well as the two turns on black dots, as shown in Fig.~\ref{fig:gg:hard:unmerge}.
	If the lower bound is matched, we can therefore always separate the green cycles as in Fig.~\ref{fig:nphard:simple:multiclause}.
	
	%Due to the strict limitations of Lemma~\ref{lemma:gg:hard:green}, the only alternative turns are in the green areas.
	%However, there is the restriction that we can only use one turn in one coherent green area (Lemma~\ref{lemma:gg:hard:guard}) and further every turn on a green pixel needs a horizontal and a vertical partner turn on a green pixel to fulfill Lemma~\ref{lemma:gg:hard:stripeven}.
	%This results in two turns in green pixels only being encapsulated \todoi{Was??} by two turns on black dotted pixels.
	%As the possible turns reachable from a black dotted pixel are strictly limited, we only have the base case (two green pixels between two black dots) as in Fig.~\ref{fig:gg:hard:unmerge}, for which we can simply split the cycle into two.
	%This is sufficient, because for a cycle to have turns on the green
	%pixels and the black dots, there has to be a subsequent ``green'' turn after a
	%``black dot'' turn, which we can all eliminate by the previous
	%argument.  Note that the part as in Fig.~\ref{fig:gg:hard:partslb:6} is
	%slightly different with respect to Lemma~\ref{lemma:gg:hard:stripeven} due to
	%an additional turn in the blue area, but the argument works analogously.
	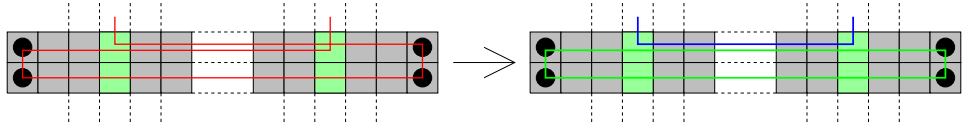
\begin{figure}[!ht]
		\centering
		\resizebox{0.9\textwidth}{!}{%
			\begin{tikzpicture}
				\newcommand{\tikzsquare}[3]{ \draw[#3] (#1-0.5,#2-0.5) -- (#1+0.5,#2-0.5) -- (#1+0.5,#2+0.5) -- (#1-0.5, #2+0.5) -- (#1-0.5,#2-0.5); }
				\newcommand{\tikzsquares}[5]{ \foreach \a in {#1,...,#3}{ \foreach \b in {#2,...,#4}{ \tikzsquare{\a}{\b}{#5} } } }
				\tikzsquares{0}{1}{5}{2}{fill=\fillfieldgray}
				\tikzsquares{8}{1}{13}{2}{fill=\fillfieldgray}
				\draw[dashed] (5.5, 2.5) -- (7.5, 2.5);
				\draw[dashed] (5.5, 0.5) -- (7.5, 0.5);
				\draw[dashed] (5.5, 1.5) -- (7.5, 1.5);
				\tikzsquares{3}{1}{3}{2}{fill=green!40}
				\tikzsquares{10}{1}{10}{2}{fill=green!40}
				\draw[dashed] (1.5, 2.5) -- (1.5, 3.5);
				\draw[dashed] (2.5, 2.5) -- (2.5, 3.5);
				\draw[dashed] (3.5, 2.5) -- (3.5, 3.5);
				\draw[dashed] (4.5, 2.5) -- (4.5, 3.5);
				\draw[dashed] (1.5, 0.5) -- (1.5, -0.5);
				\draw[dashed] (2.5, 0.5) -- (2.5, -0.5);
				\draw[dashed] (3.5, 0.5) -- (3.5, -0.5);
				\draw[dashed] (4.5, 0.5) -- (4.5, -0.5);
				\draw[dashed] (8.5, 2.5) -- (8.5, 3.5);
				\draw[dashed] (9.5, 2.5) -- (9.5, 3.5);
				\draw[dashed] (10.5, 2.5) -- (10.5, 3.5);
				\draw[dashed] (11.5, 2.5) -- (11.5, 3.5);
				\draw[dashed] (8.5, 0.5) -- (8.5, -0.5);
				\draw[dashed] (9.5, 0.5) -- (9.5, -0.5);
				\draw[dashed] (10.5, 0.5) -- (10.5, -0.5);
				\draw[dashed] (11.5, 0.5) -- (11.5, -0.5);
				\draw[fill=black] (0, 1) circle (0.3);
				\draw[fill=black] (0, 2) circle (0.3);
				\draw[fill=black] (13, 1) circle (0.3);
				\draw[fill=black] (13, 2) circle (0.3);
				\draw[red, very thick] (3, 3) -- (3,2.1) -- (13, 2.1) -- (13, 1) -- (0, 1) -- (0, 1.9) -- (10, 1.9) -- (10, 3);

				\draw[thick] (14, 1.5) -- (16,1.5);
				\draw[thick] (15, 2) -- (16,1.5) -- (15, 1);

				\tikzsquares{17}{1}{22}{2}{fill=\fillfieldgray}
				\tikzsquares{25}{1}{30}{2}{fill=\fillfieldgray}
				\draw[dashed] (22.5, 2.5) -- (24.5, 2.5);
				\draw[dashed] (22.5, 0.5) -- (24.5, 0.5);
				\draw[dashed] (22.5, 1.5) -- (24.5, 1.5);
				\tikzsquares{20}{1}{20}{2}{fill=green!40}
				\tikzsquares{27}{1}{27}{2}{fill=green!40}
				\draw[dashed] (18.5, 2.5) -- (18.5, 3.5);
				\draw[dashed] (19.5, 2.5) -- (19.5, 3.5);
				\draw[dashed] (20.5, 2.5) -- (20.5, 3.5);
				\draw[dashed] (21.5, 2.5) -- (21.5, 3.5);
				\draw[dashed] (18.5, 0.5) -- (18.5, -0.5);
				\draw[dashed] (19.5, 0.5) -- (19.5, -0.5);
				\draw[dashed] (20.5, 0.5) -- (20.5, -0.5);
				\draw[dashed] (21.5, 0.5) -- (21.5, -0.5);
				\draw[dashed] (25.5, 2.5) -- (25.5, 3.5);
				\draw[dashed] (26.5, 2.5) -- (26.5, 3.5);
				\draw[dashed] (27.5, 2.5) -- (27.5, 3.5);
				\draw[dashed] (28.5, 2.5) -- (28.5, 3.5);
				\draw[dashed] (25.5, 0.5) -- (25.5, -0.5);
				\draw[dashed] (26.5, 0.5) -- (26.5, -0.5);
				\draw[dashed] (27.5, 0.5) -- (27.5, -0.5);
				\draw[dashed] (28.5, 0.5) -- (28.5, -0.5);
				\draw[fill=black] (17, 1) circle (0.3);
				\draw[fill=black] (17, 2) circle (0.3);
				\draw[fill=black] (30, 1) circle (0.3);
			\draw[fill=black] (30, 2) circle (0.3);
			\draw[green, ultra thick] (17,1) -- (17,1.9) -- (30, 1.9) -- (30, 1) -- cycle;
			\draw[blue, ultra thick] (20,3) -- (20, 2.1) -- (27, 2.1) -- (27, 3);
		\end{tikzpicture}
	}%
	\caption[Separating cycles]{Lemma~\ref{lemma:gg:hard:stripeven} enforces the turns on the green pixels to be at the same $y$-coordinates. The cycle part can easily be separated to a cycle that covers the gray pixels and a connection between the turns in the green pixels.}
	\label{fig:gg:hard:unmerge}
\end{figure}
\end{proof}

This allows us to only concentrate on covering the crucial and the blue pixels, because all gray pixels are already covered by the ``black dot'' cycles.
We now show that every cycle cover that matches the lower bound results in a variable assignment (the cycle that covers the blue pixels either selects the {\tt true} line or the {\tt false} line and does not switch in between).

\begin{lemma}
	\label{lemma:gg:hard:variableline}
	In a cycle cover that matches the lower bound of Lemma~\ref{th:gg:hard:lb}, the cycle that covers the blue pixels of a variable also covers either all crucial pixels belonging to a {\tt true} assignment of this variable and no other crucial pixel or all crucial pixels belonging to a {\tt false} assignment of this variable and no other crucial pixel.
\end{lemma}
\begin{proof}
	%Blue can be seen as path
	It is impossible to cover a blue pixel twice without exceeding the lower bound, because we only have free turns in the green pixels. Therefore, the ``variable'' cycle has to go through the gray and crucial pixels.
	We are only allowed to make turns in the green pixels as marked in Fig.~\ref{fig:gg:hard:partslb} and only at most one in every connected green area.
	The variable cycle can take the upper or the lower path (Fig.~\ref{fig:gg:hard:partslb:6}), but cannot switch between them, as there is only one possible ``partner'' turn to achieve an even number of turns on every strip.
	The only possible choice for a partner turn is the next turn of the cycle after applying Lemma~\ref{lemma:gg:hard:separategray}. 
	This forces the path to continue until it closes the cycle.
	See Fig.~\ref{fig:gg:hard:sticktovariableline} for an illustration.
	\begin{figure}[h]
		\newcommand{\tikzsquare}[3]{ \draw[#3] (#1-0.5,#2-0.5) -- (#1+0.5,#2-0.5) -- (#1+0.5,#2+0.5) -- (#1-0.5, #2+0.5) -- (#1-0.5,#2-0.5); }
		\newcommand{\tikzsquares}[5]{ \foreach \a in {#1,...,#3}{ \foreach \b in {#2,...,#4}{ \tikzsquare{\a}{\b}{#5} } } }
		\centering
		%\old{
		\resizebox{0.7\textwidth}{!}{%
			\begin{tikzpicture}
				\tikzsquares{1}{1}{1}{10}{fill=\fillfieldblue}
				\tikzsquares{2}{10}{4}{10}{fill=\fillfieldblue}
				\tikzsquares{2}{1}{4}{1}{fill=\fillfieldblue}
				\tikzsquares{4}{2}{7}{2}{fill=\fillfieldblue};
				\tikzsquares{7}{1}{10}{1}{fill=\fillfieldblue}
				\tikzsquares{10}{2}{11}{2}{fill=\fillfieldblue}

				\tikzsquares{5}{4}{6}{16}{fill=\fillfieldgray}
				\tikzsquares{3}{13}{24}{14}{fill=\fillfieldgray}
				\tikzsquares{3}{6}{11}{7}{fill=\fillfieldgray}

				\tikzsquares{16}{16}{39}{17}{fill=\fillfieldgray}
				\tikzsquares{18}{11}{19}{19}{fill=\fillfieldgray}
				\tikzsquares{21}{11}{22}{19}{fill=\fillfieldgray}
				\tikzsquares{31}{13}{50}{14}{fill=\fillfieldgray}
				\tikzsquares{37}{11}{36}{19}{fill=\fillfieldgray}
				\tikzsquares{34}{11}{33}{19}{fill=\fillfieldgray}
				\tikzsquares{22}{15}{18}{15}{fill=green!60}
				\tikzsquares{20}{13}{20}{17}{fill=green!60}
				\tikzsquares{33}{15}{37}{15}{fill=green!60}
				\tikzsquares{35}{13}{35}{17}{fill=green!60}
				\tikzsquare{20}{15}{fill=\fillfieldred}
				\tikzsquare{35}{15}{fill=\fillfieldred}

				\draw[dashed] (50.5, 12.5) -- (52, 12.5);
				\draw[dashed] (50.5, 14.5) -- (52, 14.5);

				\draw[dashed] (11.5, 7.5) -- (13, 7.5);
				\draw[dashed] (11.5, 5.5) -- (13, 5.5);

				\draw[dashed] (11.5, 2.5) -- (13, 2.5);
				\draw[dashed] (11.5, 1.5) -- (13, 1.5);

				\tikzsquares{5}{13}{6}{14}{fill=green!60}
				\tikzsquares{5}{10}{6}{10}{fill=green!60}
				\tikzsquares{5}{6}{6}{7}{fill=green!60}

			\end{tikzpicture}
		}%
		%}%end old
		%\includegraphics[width=.75\columnwidth]{./figures/fig15_cycleup}
		\caption[Covering crucial pixels 1]{If the cycle uses the upper path, Lemma~\ref{lemma:gg:hard:stripeven} forces us to cover exactly all crucial pixels along this way, as we have only one choice for the position of the ``partner'' turn.}
		\label{fig:gg:hard:sticktovariableline}
	\end{figure}
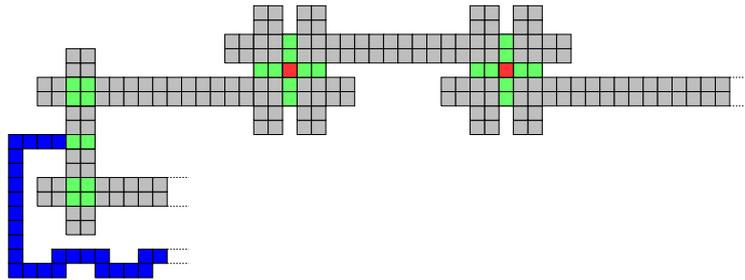
\end{proof}

Now only the crucial pixels not covered by the variable assignment cycles are left.
We have only two turns for each of them, hence we need at least two still uncovered crucial pixels in any additional cycle; this cycle must have at most twice the number of turns than the number of its newly covered crucial pixels.
This, however, is impossible for unsatisfied clauses, so we must exceed the lower bound.

\begin{lemma}
	\label{th:gg:hard:unsatub}
	Given a 2-dimensional grid graph produced by the above procedure for an arbitrary \emph{One-in-three 3SAT} formula.
	Let $\mathcal{C}$ be a set of cycles with minimum turn costs.
	Let $v$ be the number of variables in the corresponding formula, $c$ the number of clauses, and $k$ be the number of blue corner pixels.
	If the formula is not satisfiable, the cycle cover has to have at least $20v+42c+k+1$ turns.
\end{lemma}
\begin{proof}
	%TODO: Make this better. Only a quick short description because the old one was too long.
	Assume there exists a cycle cover with only $20v+42c+k$ turns (matching the lower bound of Lemma~\ref{th:gg:hard:lb}).
	We now show that this number of turns cannot suffice to cover all pixels.
	By separating the cycles on the black dotted pixels using Lemma~\ref{lemma:gg:hard:separategray}, only $4v+18c+k$ turns remain for the red and blue pixels.

	Due to Lemma~\ref{lemma:gg:hard:variableline} we have to cover one line of crucial pixels per variable by a blue variable assignment cycle.
	The remaining crucial pixels can only be covered by the red cycles known from the construction.
	However, the use of red cycles in a way that no crucial pixel is covered multiple times (which would exceed the budget) enforces the logic described in the construction.
	A valid selection of red cycles is therefore only possible if the corresponding variable assignment is feasible.
	Because this is not the case, $20v+42c+k$ turns cannot be sufficient.
\end{proof}

%----------------------------------------------------------------------------------------
% NP-Hardness Conclusion
%----------------------------------------------------------------------------------------
\subsubsection{NP-Hardness}
Finally, we can state the main theorem: the bound on the number of turns can only be met if the formula is satisfiable.
\begin{theorem}
	Given a construction, as described above, for an arbitrary formula with $v$ variables and $c$ clauses.
	Let $k$ be the number of blue highlighted corner pixels.
	Then the formula is satisfiable if and only if there exists a cycle cover with $20v+42c+k$ turns.
\end{theorem}
\begin{proof}
	This follows from Lemma~\ref{th:gg:hard:satub} and Lemma~\ref{th:gg:hard:unsatub}.
\end{proof}
This concludes the proof of the NP-hardness of full coverage cycle cover (Theorem~\ref{th:gg:hard:nphard}).

\subsection{Subset Coverage in Thin Grid Graphs}
For subset cover we can also show hardness for {\em thin} grid graphs.
Arkin et al.~\cite{arkin2001optimal,arkin2005optimal}
exploits the structure of these graphs to compute an optimal minimum-turn cycle cover in polynomial time.
If we only have to cover a subset of the vertices, the problem becomes NP-hard again. 
The proof is inspired by the construction of Aggarwal et al.~\cite{aggarwal2000angular} for the angular-metric cycle cover problem and significantly simpler than the one for full coverage.

\begin{theorem}
	\label{th:gg:hard:nphardsubset}
				The minimum-turn subset cycle cover problem is NP-hard, even in thin grid graphs.
\end{theorem}
\begin{proof}
\begin{figure}[h]
	\centering
	\newcommand{\tikzsquare}[3]{
	\draw[#3] (#1,#2) -- (#1+1,#2) -- (#1+1,#2+1) -- (#1, #2+1) -- (#1,#2);
}
\newcommand{\tikzsquares}[5]{
	\foreach \a in {#1,...,#3}{
		\foreach \b in {#2,...,#4}{
			\tikzsquare{\a}{\b}{#5}
		}
	}
}
\newcommand{\labelsize}[1]{\large{#1}}
\newcommand{\labelsizevar}[1]{\LARGE{#1}}
\centering
\resizebox{.8\columnwidth}{!}{%
	\begin{tikzpicture}
	%variables
		\tikzsquares{1}{1}{1}{29}{fill=\fillfieldgray}
	\tikzsquares{44}{1}{44}{29}{fill=\fillfieldgray}
	\tikzsquares{1}{29}{44}{29}{fill=\fillfieldgray}
	\tikzsquares{1}{27}{44}{27}{fill=\fillfieldgray}
	\tikzsquares{1}{1}{44}{1}{fill=\fillfieldred}
	\node[scale=2.5, align=left] at (46.5,28.5) {\labelsizevar{$x_1$}};

	\tikzsquares{3}{3}{3}{24}{fill=\fillfieldgray}
	\tikzsquares{42}{3}{42}{24}{fill=\fillfieldgray}
	\tikzsquares{3}{24}{42}{24}{fill=\fillfieldgray}
	\tikzsquares{3}{22}{42}{22}{fill=\fillfieldgray}
	\tikzsquares{3}{3}{42}{3}{fill=\fillfieldred}
	\node[scale=2.5, align=left] at (46.5,23.5) {\labelsizevar{$x_2$}};

	\tikzsquares{5}{5}{5}{19}{fill=\fillfieldgray}
	\tikzsquares{40}{5}{40}{19}{fill=\fillfieldgray}
	\tikzsquares{5}{19}{40}{19}{fill=\fillfieldgray}
	\tikzsquares{5}{17}{40}{17}{fill=\fillfieldgray}
	\tikzsquares{5}{5}{40}{5}{fill=\fillfieldred}
	\node[scale=2.5, align=left] at (46.5,18.5) {\labelsizevar{$x_3$}};

	\tikzsquares{7}{7}{7}{14}{fill=\fillfieldgray}
	\tikzsquares{38}{7}{38}{14}{fill=\fillfieldgray}
	\tikzsquares{7}{14}{38}{14}{fill=\fillfieldgray}
	\tikzsquares{7}{12}{38}{12}{fill=\fillfieldgray}
	\tikzsquares{7}{7}{38}{7}{fill=\fillfieldred}
	\node[scale=2.5, align=left] at (46.5,13.5) {\labelsizevar{$x_4$}};

	%Clauses
		\tikzsquares{10}{9}{10}{32}{fill=\fillfieldgray}
	\tikzsquares{12}{9}{12}{32}{fill=\fillfieldgray}
	\tikzsquares{14}{9}{14}{32}{fill=\fillfieldgray}
	\tikzsquares{10}{9}{14}{9}{fill=\fillfieldgray}
	\tikzsquares{10}{32}{14}{32}{fill=\fillfieldgray}
	\tikzsquare{12}{32}{fill=\fillfieldred}
	\tikzsquare{12}{9}{fill=\fillfieldred}
	\node[scale=2.5] at (12.5,33.5) {\labelsize{$x_1\vee x_2 \vee x_3$}};
	%x1
		\tikzsquare{10}{29}{fill=\fillfieldred}
	\tikzsquare{12}{27}{fill=\fillfieldred}
	\tikzsquare{14}{27}{fill=\fillfieldred}
	%x2
		\tikzsquare{10}{22}{fill=\fillfieldred}
	\tikzsquare{12}{24}{fill=\fillfieldred}
	\tikzsquare{14}{22}{fill=\fillfieldred}
	%x3
		\tikzsquare{10}{17}{fill=\fillfieldred}
	\tikzsquare{12}{17}{fill=\fillfieldred}
	\tikzsquare{14}{19}{fill=\fillfieldred}

	\tikzsquares{17}{9}{17}{32}{fill=\fillfieldgray}
	\tikzsquares{19}{9}{19}{32}{fill=\fillfieldgray}
	\tikzsquares{21}{9}{21}{32}{fill=\fillfieldgray}
	\tikzsquares{17}{9}{21}{9}{fill=\fillfieldgray}
	\tikzsquares{17}{32}{21}{32}{fill=\fillfieldgray}
	\tikzsquare{19}{32}{fill=\fillfieldred}
	\tikzsquare{19}{9}{fill=\fillfieldred}
	\node[scale=2.5] at (19.5,33.5) {\labelsize{$\overline{x_1}\vee \overline{x_2} \vee \overline{x_3}$}};
	%x1
		\tikzsquare{17}{27}{fill=\fillfieldred}
	\tikzsquare{19}{29}{fill=\fillfieldred}
	\tikzsquare{21}{29}{fill=\fillfieldred}
	%x2
		\tikzsquare{17}{24}{fill=\fillfieldred}
	\tikzsquare{19}{22}{fill=\fillfieldred}
	\tikzsquare{21}{24}{fill=\fillfieldred}
	%x3
		\tikzsquare{17}{19}{fill=\fillfieldred}
	\tikzsquare{19}{19}{fill=\fillfieldred}
	\tikzsquare{21}{17}{fill=\fillfieldred}

	\tikzsquares{24}{9}{24}{32}{fill=\fillfieldgray}
	\tikzsquares{26}{9}{26}{32}{fill=\fillfieldgray}
	\tikzsquares{28}{9}{28}{32}{fill=\fillfieldgray}
	\tikzsquares{24}{9}{28}{9}{fill=\fillfieldgray}
	\tikzsquares{24}{32}{28}{32}{fill=\fillfieldgray}
	\tikzsquare{26}{32}{fill=\fillfieldred}
	\tikzsquare{26}{9}{fill=\fillfieldred}
	\node[scale=2.5] at (26.5,33.5) {\labelsize{$x_2\vee\overline{x_3}\vee x_4$}};
	%x2
		\tikzsquare{24}{24}{fill=\fillfieldred}
	\tikzsquare{26}{22}{fill=\fillfieldred}
	\tikzsquare{28}{22}{fill=\fillfieldred}
	%x3
		\tikzsquare{24}{19}{fill=\fillfieldred}
	\tikzsquare{26}{17}{fill=\fillfieldred}
	\tikzsquare{28}{19}{fill=\fillfieldred}
	%x4
		\tikzsquare{24}{12}{fill=\fillfieldred}
	\tikzsquare{26}{12}{fill=\fillfieldred}
	\tikzsquare{28}{14}{fill=\fillfieldred}

	\tikzsquares{31}{9}{31}{32}{fill=\fillfieldgray}
	\tikzsquares{33}{9}{33}{32}{fill=\fillfieldgray}
	\tikzsquares{35}{9}{35}{32}{fill=\fillfieldgray}
	\tikzsquares{31}{9}{35}{9}{fill=\fillfieldgray}
	\tikzsquares{31}{32}{35}{32}{fill=\fillfieldgray}
	\tikzsquare{33}{32}{fill=\fillfieldred}
	\tikzsquare{33}{9}{fill=\fillfieldred}
	\node[scale=2.5] at (33.5,33.5) {\labelsize{$\overline{x_1}\vee x_2 \vee \overline{x_4}$}};
	%x1
		\tikzsquare{31}{27}{fill=\fillfieldred}
	\tikzsquare{33}{29}{fill=\fillfieldred}
	\tikzsquare{35}{29}{fill=\fillfieldred}
	%x2
		\tikzsquare{31}{22}{fill=\fillfieldred}
	\tikzsquare{33}{24}{fill=\fillfieldred}
	\tikzsquare{35}{22}{fill=\fillfieldred}
	%x4
		\tikzsquare{31}{14}{fill=\fillfieldred}
	\tikzsquare{33}{14}{fill=\fillfieldred}
	\tikzsquare{35}{12}{fill=\fillfieldred}

	\end{tikzpicture}
}%
	\caption[NP-hardness construction for subset coverage]{Example for the NP-hardness reduction using the \emph{One-in-three 3SAT} formula $(x_1\vee x_2 \vee x_3)\wedge(\overline{x_1}\vee \overline{x_2} \vee \overline{x_3})\wedge(x_2\vee\overline{x_3}\vee x_4)\wedge(\overline{x_1}\vee x_2 \vee \overline{x_4})$.}
				\label{fig:gg:cc:mtsccp:reductionexample}
\end{figure}
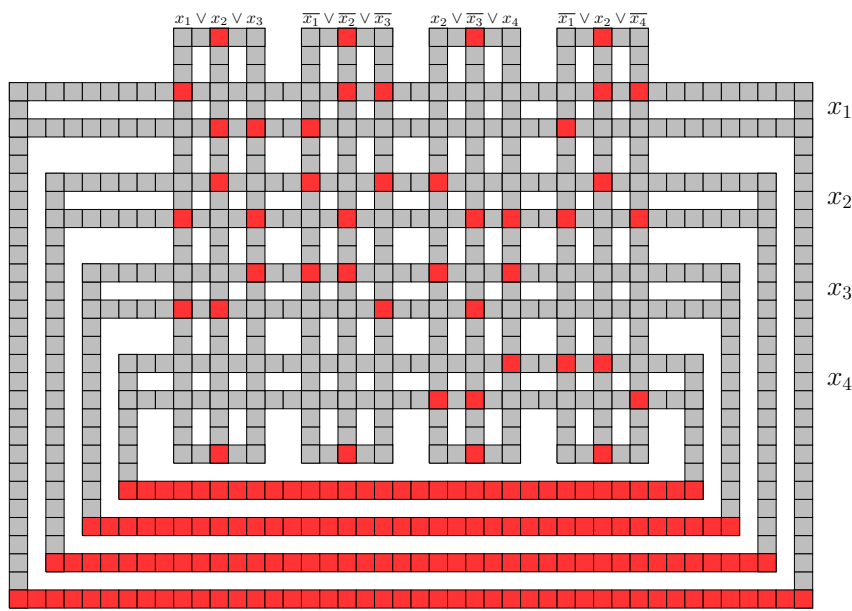

The proof is inspired by the construction of Aggarwal et al.~\cite{aggarwal2000angular} for the angular-metric cycle cover.
See Fig.~\ref{fig:gg:cc:mtsccp:reductionexample} for an overview of the construction;
due to its relative simplicity, we only sketch the details.
Every variable consists of a U-shape, with two rows connecting the ends of the U.
The bottom of the U has to be covered in any case, so we are free to either cover the upper or the lower row without additional turn cost.
Covering the upper row means setting the variable to {\tt true}, covering the lower row to {\tt false}.
These two rows intersect with the constructions for the clauses.
Every clause construction consists of three vertical columns (intersecting with the variable constructions) and a horizontal connection at each end.
At the top and the bottom of each is a pixel that has to be covered.
The cheapest way to cover these two pixels is by a cycle that goes through the columns.
Thus, we are able to cover two of three columns per clause construction for free.
The two covered columns represent the two literals of the clause that have to be {\tt false}.

We mark additional pixels at the intersections of the variable constructions and the clause constructions,
so that these pixels must also be covered.
This is done in a such way that if we set a variable to a specific value for
all clauses that now become {\tt true}, the two columns to be covered are automatically
enforced.  In the example, if we set $x_1$ to {\tt true} and cover the upper row, we
have to use the two lines on the right in clause $x_1\vee x_2 \vee x_3$, because
otherwise the two lower pixels in the intersection are not covered.
On the other hand, this enforces $x_2$ and $x_3$ to be {\tt false}, as the
corresponding two left pixels to be covered can no longer be covered by the
clause's cycle.  
The cycle cover represents a valid solution for the \emph{One-in-three 3SAT} formula,
if and only if we do not need any additional cycles. 
\end{proof}
Arkin et al.~\cite{arkin2001optimal,arkin2005optimal}
showed how the structure of thin grid graphs can be exploited for computing an optimal minimum-turn cycle cover in polynomial time.
If we only have to cover a subset of the vertices, the problem is NP-complete.

\section{Approximation Algorithms}
\label{sec:approx}
\subsection{Cycle Cover}
Now we describe a $2\omega$-approximation algorithm for the semi-quadratic (penalty) cycle cover problem with $\omega$ atomic strips per point if the edge weights satisfy Eq.~(\ref{eq:triang}).
%As discussed, the minimum turn cycle covers in grid graphs can be modelled with $\omega=2$ as well as the discretized angular cycle cover, optionally as linear combination with distance costs.
We focus on the full coverage version, as the penalty variant can be modelled in full coverage (with equal $\omega$ and while still satisfying Eq.~(\ref{eq:triang})), 
by adding for every point $p\in P$ two further points that have a zero cost cycle only including themselves and a cycle that also includes $p$ with the cost of the penalty.
%Hence, we only need to consider the full coverage variant.

% ## Description of CC Algorithm
\ifx\figapproximationexample\undefined%
\def\figapproximationexample{}
\begin{figure}[t]
\vspace*{-3mm}
	\centering
	\includegraphics[width=0.8\columnwidth]{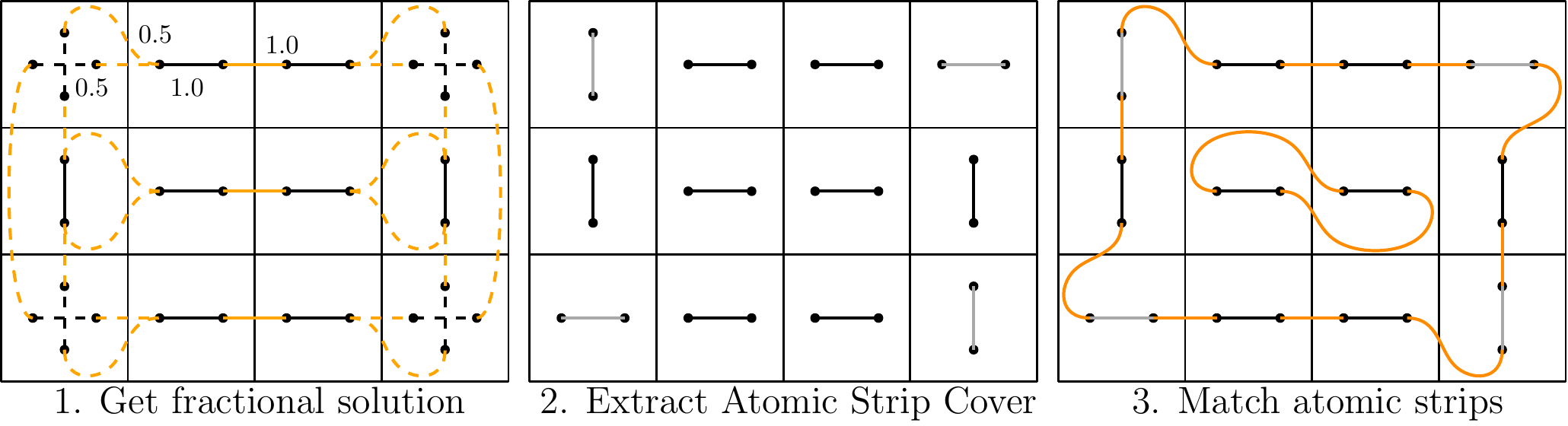}
	\caption{Example of the approximation algorithm for a simple full cycle cover instance in a grid graph. First the fractional solution of the integer program~(\ref{eq:ip:obj})-(\ref{eq:ip:inout}) is computed. Strips and edges with value $0$ are omitted, while dashed ones have value $0.5$. Then the dominant (i.e., highest valued) atomic strips of this solution are selected. 
%The grey atomic strips are ambiguous, i.e., we could have also chosen the other one. 
Finally, a minimum weight perfect matching on the ends of the atomic strips is computed. 
(Recall that atomic strips only have an but no length, so the curves in the corner indicate simple $90^\circ$ turns.)}
	\label{fig:apx:example}
\vspace*{-3mm}
\end{figure}
\fi%
%fig:apx:example If you comment this out, it is automatically added in the appendix.
Our approximation algorithm proceeds as follows.
We first determine an {atomic strip cover} via linear programming. %, as discussed in Sec.~\ref{sec:short_strip_cover}.
Computing an optimal {atomic strip cover} is NP-hard; we can show that choosing the \emph{dominant}
strips for each pixel in the fractional solution, i.e., those with the highest value, suffices to obtain provable good solutions.
As a next step, we connect the {atomic strips} to a {cycle cover}, using a minimum-weight perfect matching.
See Fig.~\ref{fig:apx:example} for an illustration.
%\todo[inline]{Make sure the terms are written uniformly. Capitalized or non-capitalized.}
%Here we can utilize properties of a cycle cover in a grid graph to obtain a matching graph of linear instead of quadratic size.

%\subsubsection{Atomic Strip Covers}
%\label{sec:short_strip_cover}

% Abstract Problem
We now describe the integer program whose linear programming relaxation is solved to select the dominant atomic strips.
It searches for an optimal atomic strip cover that yields a perfect matching of minimum weight. % (which is the optimal solution).
%Let $P$ be the set of pixel for which each $p\in P$ has a set of {atomic strips} $O_p$, in this case a horizontal and a vertical one.
%Each atomic strip $vw\in O_p$ yields two vertices $v$ and $w$; the set of all such vertices is denoted by $V_O$.
%The vertices $V_O$ induce a graph that may contain loop edges (but not multi-edges), with edges denoted by $E_O$.
%These edges $E_O$ are the superset of edges for the matching and are undirected, i.e., $vw=wv\in E_O$.
To satisfy Eq.~(\ref{eq:triang}), transitive edges may need to be added, especially loop-edges (which are not used in the final solution).
The IP does not explicitly enforce cycles to contain at least two points:  all small cycles consist only of transitive edges 
that implicitly contain at least one further atomic strip/point.
%The loop-edges are only necessary for Eq.~(\ref{eq:triang}) and are never actually used.
%The cost of such an edge $e\in E_O$ is denoted by $c(e)\in \mathbb{Q}^+_0$.
% IP
%This allows us to express the problem as the following integer program.
For the usage of a matching edge $e=vw\in E_O$, we use the Boolean variable $x_e=x_{vw}$.
For the usage of an {atomic strip} $o=vw\in O_p, p\in P$, we use the Boolean variable $y_o=y_{vw}$.
\begin{eqnarray}
	\min & \displaystyle \sum_{e\in E_O} \text{cost}(e)x_e & \label{eq:ip:obj}\\
	\text{s.t.} & \displaystyle \sum_{vw\in O_p} y_{vw} = 1 & p\in P\label{eq:ip:cover}\\
	& \displaystyle 2x_{vv} + \sum_{\substack{e\in E_O(v)\\ e\not= vv}} x_e =2x_{ww} + \sum_{\substack{e\in E_O(w)\\ e\not= ww}} x_e = y_{vw} & p\in P, vw\in O_p \label{eq:ip:inout}%\\
	%& x_e, y_o \in \mathbb{B} & e\in E_o, p\in P, o\in O_p \label{eq:ip:vars}
\end{eqnarray}
We minimize the cost of the used edges, with Eq.~(\ref{eq:ip:cover}) forcing the
selection of one {atomic strip} per pixel (atomic strip cover) and Eq.~(\ref{eq:ip:inout}) ensuring that
exactly the vertices (endpoints) of the selected {atomic strips} are matched, with
loop edges counting double due to their two ends.
%Our approximation algorithm selects the dominant strip, i.e., the strip with the highest value in the fractional solution of the linear programming relaxation.
%(which allows partial selection of edges and atomic strips but can be computed efficiently contrary to the integer program).
%See Fig.~\ref{fig:apx:example} (left) for how such a fractional solution could look like and Fig.~\ref{fig:apx:example} (middle) for a corresponding atomic strip cover.

% # Matching these yields 4-apx
%\todoi{All our approximation factors are for turn and distance cost (linear combination). Is this sufficiently expressed everywhere? The paper is on turn costs, do we really need to state it everywhere?}
\begin{theorem}
	\label{th:ccapx}
	%A minimum-weight perfect matching of the endpoints of the dominant {atomic strips} of the fractional solution of the integer program \eqref{eq:ip:obj}-\eqref{eq:ip:inout} yields 
Assuming edge weights that satisfy Eq.~(\ref{eq:triang}), there is 
 a $2\omega$-approxima\-tion for semi-quadratic (penalty) cycle cover. 
\end{theorem}
\begin{proof}
	We prove this theorem by showing that there exists a matching on the dominant {atomic strips} of the fractional solution with at most $2\omega$ the cost of the fractional solution.
	This implies that the minimum weight perfect matching on these strips has to be at least that good.
  %Let us now consider the details of this proof.

	% 1. Dominant strips
	Due to the constraints of type Eq.~(\ref{eq:ip:cover}), the dominant strip $o\in O_p$ for every $p\in P$ has a value of 
	$y_o\geq 1/\omega$ in any feasible solution of the LP-relaxation.
	Let $o_p$ be this dominant strip for $p\in P$ and $V'_O\subseteq V_O$ the set of the endpoints of all of them.
	% 2. Set $p\not \subset V'$ to zero via for loop
	There may also be further $p \in P, o\in O_p$ with $y_o > 0$.
	We can set these to zero using the operations of
	Fig.~\ref{fig:reduceusages} without increasing the objective value and without
	changing any $y_{o_p}$ while also satisfying Eq.~(\ref{eq:ip:inout}),
	because of the constraint in Eq.~(\ref{eq:triang}).
		 \begin{figure}
			 \centering
			 \includegraphics[width=1\textwidth]{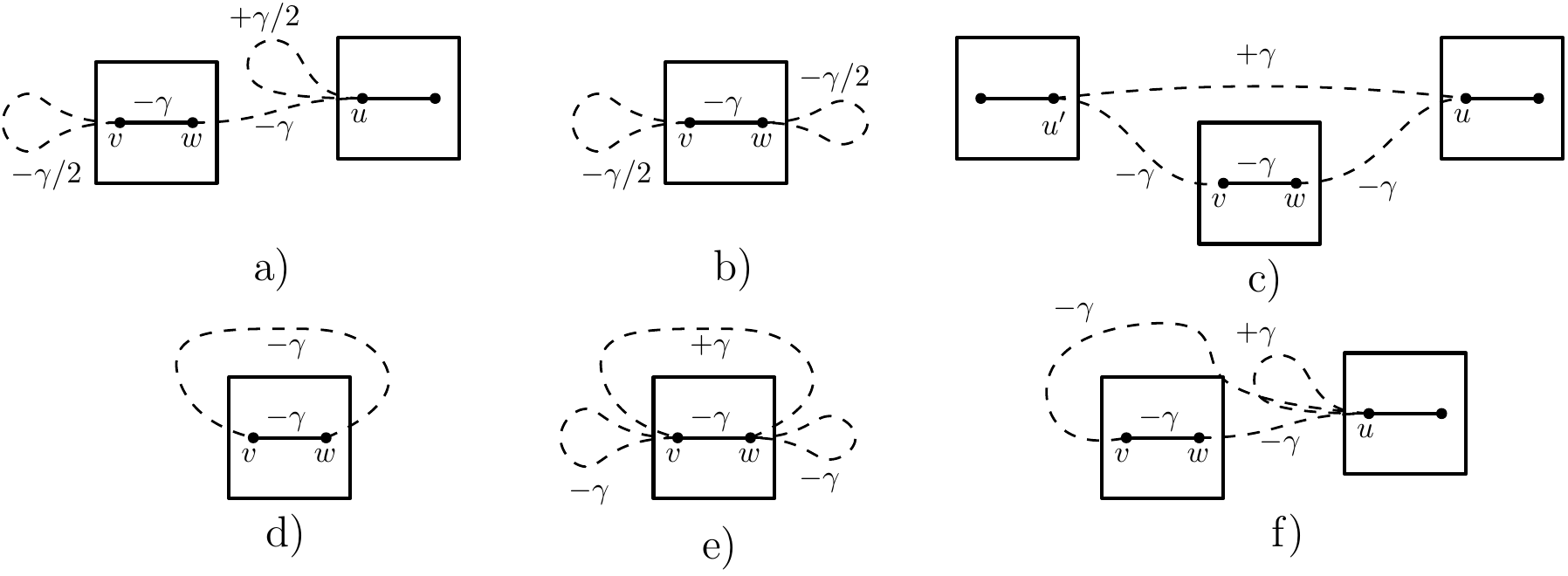}
			 \caption{Operations for reducing the usage of a strip $vw$.
			 Strips are shown by thick lines, while edges are indicated by dashed lines. These
			 costs do not increase (for
			 $\gamma>=0$), due to Ineq.~\ref{eq:triang}. For some cases as a) one needs a
			 repetitive application of this inequality: first replace the $\gamma/2*vv$
			 loop-edge by $\gamma/2 * vu$ while reducing $uw$ by $\gamma/2$ and then replace
			 the $\gamma/2*wu$ and the $\gamma/2*vu$ by $\gamma/2*uu$. As long as the usage
			 is greater than zero, at least one of the operations is possible.}
			 \label{fig:reduceusages}
		 \end{figure}
	The resulting solution only violates Eq.~(\ref{eq:ip:cover}) by relaxing it to $\sum_{vw\in O_p} y_{vw} \geq 1/\omega$, but we fix this later by a simple multiplication.
	Now we only have edges from $E'_O=E_O\cap (V'_O\times V'_O)$ in our fractional solution.
	The corresponding objective value is still a lower bound on the optimal solution of the integer program.

	% 3. Multiply Solution by $\max\{|S_i| \mid S_i\in S\}$ -> $x_{p'_i} \geq 1$
	We can multiply the solution by $\omega$ resulting in $y_{o_p} \geq 1$ for all $p\in P$.
	We can apply the same procedure to reduce all $y_{o_p} > 1$ to $=1$.
	The new solution has a cost of at most $\omega$ times the optimal
	solution and all variables $y_o, o\in O_p, p\in P$ are now Boolean; however, the edge
	usages can still be fractional.

	% 5. Fix values of $p'_i$ and remove unnecessary vertices -> matching polytope
	If we now fix the values for $y_o, o\in O_p, p\in P$ and remove all vertices (and their incident edges) that are not in $V'_O$, we have a minimum perfect matching polytope with loop-edges as stated below.
	\begin{eqnarray}
		\min & \sum_{e\in E'_O} c(e)*x_e & \\
		\text{s.t.} &\sum_{e=vw\in E'_O} x_e + 2*x_{vv} = 1 & \forall v\in V'_O\\
		& 0\leq x_e \leq 1 & \forall e\in E'_O
	\end{eqnarray}
	Because the previous solution is still feasible in this formulation, the optimal solution of this formulation is also at most $\omega$ the costs of the optimal solution of the problem.

	% 6. Matching Polytope is HalfIntegral
	We can show that the new polytope contains a half-integral optimal solution.
	The proof works analogously to the proof of the half-integrality of the classic matching polytope (see, e.g., Theorem~{$6.13$} in the book of Cook et al.~\cite{cook1997combinatorial}).
	We create a bipartite graph $G'_{\text{bip}}(V'_O\times 2, E'_{\text{bip}})$ in which for every $v\in V'_O$ there are two vertices $v,v'\in V'_O\times 2$ and for every edge $e=vw\in E'_O$ there are the edges $vw'$ and $v'w$ in $E'_{\text{bip}}$, see Fig.~\ref{fig:gg:aa:bipartdouble}.
	\begin{figure}
		\centering
		\includegraphics[width=0.7\textwidth]{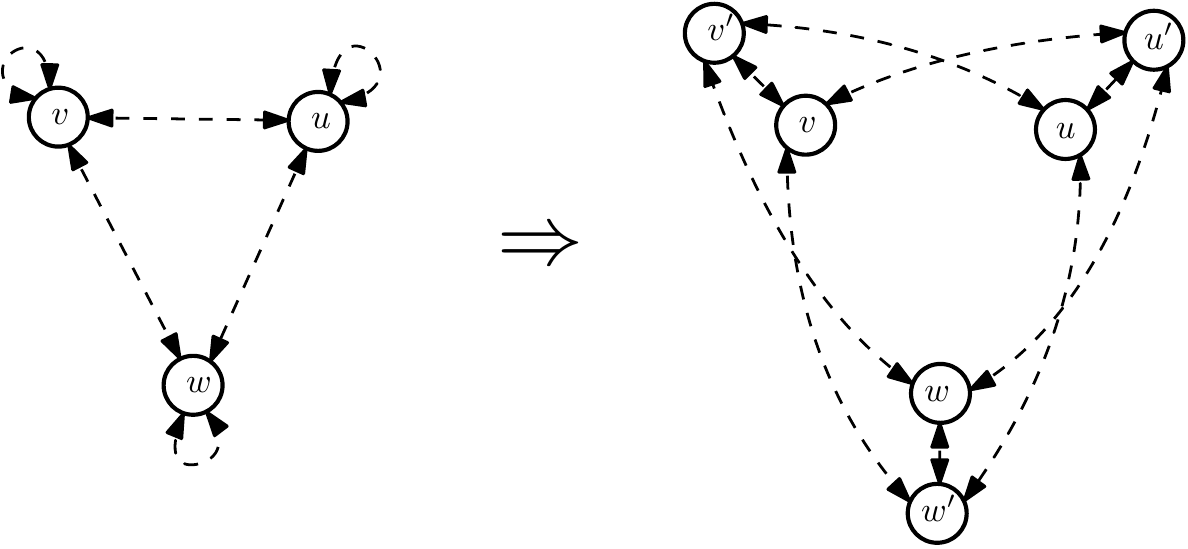}
		\caption{Converting a graph to a bipartite graph. Every edge is doubled except the loop-edges which become regular edges. The matching $vw'$,$wu'$ $uv'$ would be converted to the matching $0.5vw, 0.5wu, 0.5uv$. There always exists a perfect matching in  the original graph for our problem.}
		\label{fig:gg:aa:bipartdouble}
	\end{figure}
	A loop-edge $vv\in E'_O$ is simply replaced by a single edge $vv'\in E'_{\text{bip}}$.
	Because there are no edges within the individual copies of $V'_O$ in $G'_\text{bip}$, the graph is bipartite; the bipartite matching polytope is known to be integral.
	By assigning the edges in $G'_{\text{bip}}$ the same values as for the solution on $G'_O(V'_O, E'_O)$ (doubling for loop-edges), we obtain a feasible solution with twice the costs on $G'_{\text{bip}}$.
	The optimal solution in $G'_{\text{bip}}$, hence, at most twice as expensive as the optimal solution in $G'_O$.
	The optimal solution in $G'_{\text{bip}}$ (which is integral) can be transformed to a solution of half the costs for $G'_O$ by assigning each edge half the sum of the corresponding two edges (or one edge for loop-edges) in the solution for $G'_{\text{bip}}$.
	This solution is half-integral and optimal.
	Hence, we have an optimal solution where regular edges are selected by $0$, $0.5$, or $1$ and loop-edges are selected by $0$ or $0.5$.

	% 7. Multiply by two to get {1,2}-Polytope
	If we double such a solution, we obtain an integral solution with at most $2*\omega$ the costs of the original linear program.
	Every vertex has either two regular edges selected with $1$ each, a single regular edge selected twice, or a loop-edge selected with $1$.
	We can remove the loop-edges without cost increase as shown in Fig.~\ref{fig:removeloopedges} by also reducing the strip usage to one (a loop-edge always implies a double usage).
	For removing other double usages we can use nearly the same technique as for fractional in Fig.~\ref{fig:reduceusages}.
	We only have to make sure never to create fractional loop-edges (we do not need the cases in which this can happen because these cases all have loop-edges in the beginning which we removed before).
	Integral loop-edges can immediately be removed as in Fig.~\ref{fig:removeloopedges}.
	\begin{figure}
		\centering
		\includegraphics[width=0.9\textwidth]{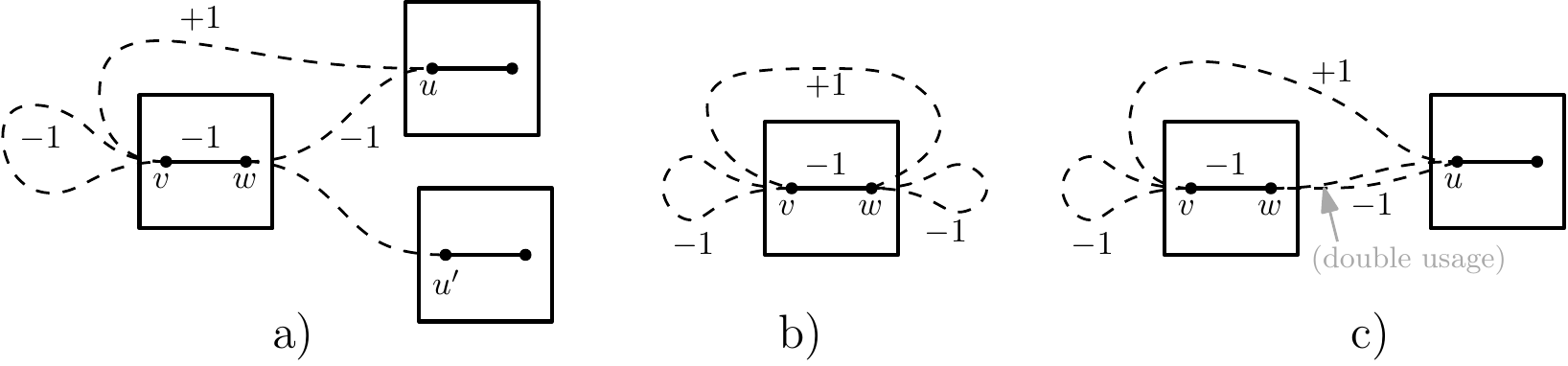}
		\caption{Removing integral loop-edges. A loop-edge always implies a double usage of the corresponding strip. The costs do not increase due to Ineq~\ref{eq:triang}. A `$+1$' indicates a new edge, a `$-1$' the removal of an edge.}
		\label{fig:removeloopedges}
	\end{figure}

	% 9. Solution is smaller than $2*\max\{|S_i| \mid S_i\in S\}$ times LP
	In the end, we have a solution that is a matching for $G'_O$ and has at most $2\omega$ the cost of the objective value of the linear relaxation.
	As the linear relaxation provides a natural lower bound, this concludes the proof.

\end{proof}
\ifx\figapproximationexample\undefined%
\def\figapproximationexample{}
\begin{figure}[t]
\vspace*{-3mm}
	\centering
	\includegraphics[width=0.8\columnwidth]{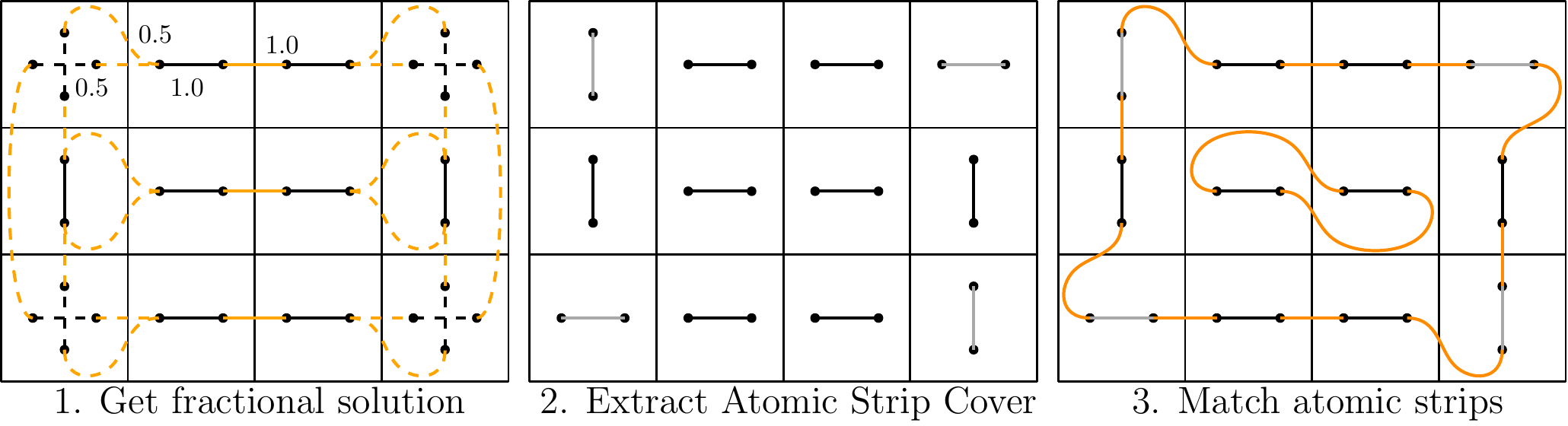}
	\caption{Example of the approximation algorithm for a simple full cycle cover instance in a grid graph. First the fractional solution of the integer program~(\ref{eq:ip:obj})-(\ref{eq:ip:inout}) is computed. Strips and edges with value $0$ are omitted, while dashed ones have value $0.5$. Then the dominant (i.e., highest valued) atomic strips of this solution are selected. 
%The grey atomic strips are ambiguous, i.e., we could have also chosen the other one. 
Finally, a minimum weight perfect matching on the ends of the atomic strips is computed. 
(Recall that atomic strips only have an but no length, so the curves in the corner indicate simple $90^\circ$ turns.)}
	\label{fig:apx:example}
\vspace*{-3mm}
\end{figure}
\fi%

\begin{figure}
	\centering
	\includegraphics[width=0.4\textwidth]{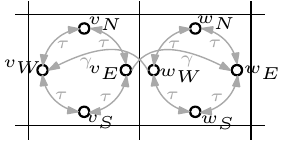}
	\caption{We can easily compute the shortest path with turn costs in grid graphs by replacing every point (that only encodes a position) by vertices for all headings of interest (four for grid graphs). The turn costs ($\tau$) then are encoded in the costs of switching between these vertices. A vertex heading north allows only to go north such that the rotation edges have to be used to make a rotation. $\gamma$ denotes the distance cost.
	%Note that this is actually not needed for the cycle cover approximation algorithm as it is indirectly done in the advanced matching technique in Sec.~\ref{sec:matching}.
	}
	\label{fig:apx:directedshortestpath}
\end{figure}

\subsection{Tours}

%%At this point we have seen how to approximate the cycle cover problems.
%In this section we see how to obtain approximations for the tour variants based on the cycle cover approximations.

A given cycle cover approximation can be turned into a tour approximation at the expense of an additional constant factor.
Because every cycle involves at least two points and a full rotation, we can use classic tree techniques known for TSP variants to connect the cycles and charge the necessary turns to the involved cycles.
We sketch the basic ideas and elaborate the details for the most difficult case afterwards in Theorem~\ref{th:gg:aa:penalty:tour2d} (the other cases are analogous).

For the classic Traveling Salesman Problem with triangle inequality, minimum spanning trees are trivial lower bounds, as any tour
must contain a spanning tree. %removing an edge of a tour provides a tree that is cheaper than the tour but cannot be cheaper than the cheapest tree.
This is also true for the penalty TSP and the prize-collecting Steiner tree; note that ``penalty'' and ``prize-collecting'' variants
are completely equivalent. 
Doubling optimal trees yields trivial 2-approximations. (The prize-collecting Steiner tree is NP-hard, 
but there is a 2-approximation~\cite{goemans1995general}.)

This is not directly possible with turn costs, because it matters from where a vertex is entered.
However, if we already have a cycle cover and we aim to connect them, this gets significantly easier.
If there is a path between two cycles, we can double it and merge the cycles,
requiring not more than an additional $180^\circ$ turns at the ends of the paths,
regardless of the direction from which the path hits the cycle. %(trivial geometric
%argument on disjunct sectors of a half circle, which is true for grid as well
%as purely geometric instances even in 3D).  
Thus we can merge two cycles with a
cost of two times the cheapest path between them, plus $360^\circ$ for
connecting the elements.
The costs of the paths can be charged to the optimal tour, analogous to the classic TSP.
On the other hand, the $360^\circ$ can be charged to one of the cycles, because
every cycle needs to do at least a full rotation, and there are fewer such merge
processes than cycles in the initial cycle cover.  This directly yields a
method for approximating subset tours when a cycle cover approximation is
given, see Fig.~\ref{fig:mst_subset_connect}.

In order to connect the cycles of a penalty cycle cover to a penalty tour, we
cannot simply use the doubled minimum spanning tree technique like we did for subset
coverage; e.g., if the penalty cycle cover consists of two distant cycles, it
may be cheaper not to connect them, but select the better cycle as a tour and
discard the other cycle.
Instead, we double a prize-collecting Steiner tree instead of a minimum
spanning tree, in which every vertex not in the tree results in a penalty (or
every vertex in the tree provides a prize). Without loss of generality, we
may assume that all cycles are disjoint: otherwise, we connect two
crossing cycles at a cost of at most $360^\circ$, which can be charged to
the merged cycle.
Then the penalty for not including a cycle is the sum of penalties of all its pixels.
%(we directly connect all crossing cycles and thus assume all cycles to be
%disjunct, this costs a $360^\circ$ turn per merge step which can be charged
%onto the vanished cycle).  
As a result, we get a constant-factor approximation for penalty tours.

\begin{theorem}
Assuming validity of Eq.~\ref{eq:triang} 
we can establish the following approximation factors for tours.
\begin{enumerate}
	\item Full tours in regular grid graphs:  $6$-approximation.
	\item Full tours in generalized grid graphs: $4\omega$-approximation.
	\item Subset tours in (generalized) grid graphs:  $(4\omega+2)$-approximation.
	\item Geometric full tours: $(4\omega+2)$-approximation.
	\item Penalty tours (in grid graphs and geometric):  $(4\omega+4)$-approximation.
\end{enumerate}
These results also hold for  objective functions that are linear combinations of length and turn costs.
\end{theorem}
\begin{proof}
	Crucial are that (1) a cycle always has a turn cost of at least $360^\circ$, (2) two intersecting cycles can be merged with a cost of at most $360^\circ$, and (3) two cycles intersecting on a $180^\circ$ turn can be merged without additional cost.
	\begin{enumerate}
		\item For full tours in grid graphs, greedily connecting cycles provides a tour with at most $1.5$ times the
			turn cost of the cycle cover, while a local optimization  can be exploited
			to limit the length to $4$ times the optimum, as shown by Arkin et al.~\cite{arkin2001optimal}.
		\item In a cycle cover for (generalized) grid graphs, there are always at least two cycles with a distance of one, while every cycle has a length of at least $2$; otherwise the cycle cover is already a tour. This allows iteratively merging cycles at cost at most as much as a cheapest cycle; the total number of merges is less than the number of cycles.
		\item and (iv) For subset coverage in grid graphs or full
coverage in the geometric case, we need to compute the cheapest paths between
any two cycles, ignoring the orientations at the ends. First connect all
intersecting cycles, charging the cost on the vanishing cycles.
			The minimum spanning tree on these edges is a lower bound on the cost of the tour.
			Doubling the MST connects all cycles with the cost of twice the MST, the cost of the cycle cover, and the turn costs at the end of the MST edges, which can be charged to the cycles.
\addtocounter{enumi}{1}
		\item Penalty tours can be approximated in a similar manner.
%The price for each cycle is the sum of the penalties of all its points. 
Instead of an MST, we use a Price-Collecting Steiner Tree, which is a lower bound on an optimal penalty tour.
%(which is the MST analogue to a penalty tour). 
We use a $2$-approximation for the PCST~\cite{goemans1995general}, as it is NP-hard. 
We achieve a cost of twice the 2-approximation of the PCST, the cost of the
penalty cycle cover, and the cost of its cycles again for charging the
connection costs. The penalties of the points not in the cycle cover are
already paid by the penalty cycle cover.  
\end{enumerate}
%\qed
\end{proof}
%\todo[inline]{Aboves theorem and proof sketch are using some unelegant language. Recall that the geometric problem is no called discretized angular-metric TSP.}

For completeness, we describe the details for 2-dimensional grid graphs; other cases are analogous.
\begin{figure}
	\centering
	\includegraphics[width=0.4\textwidth]{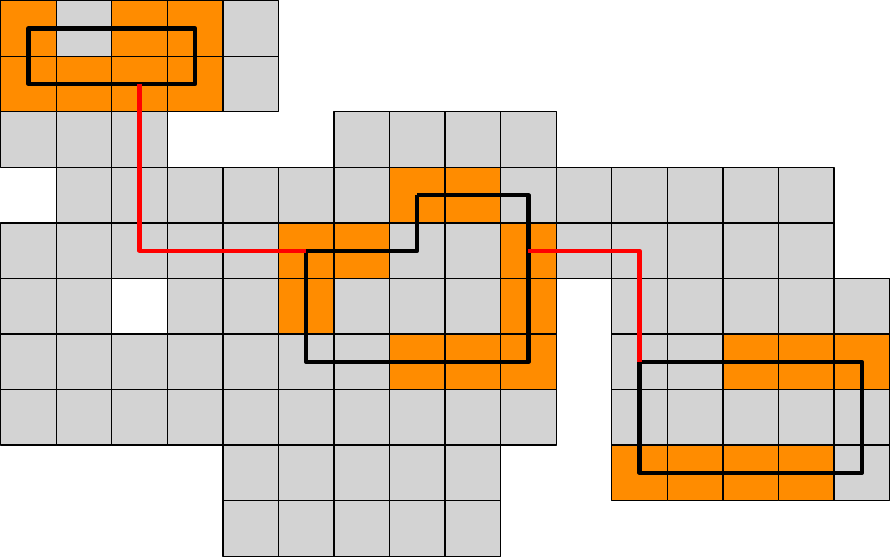}
	\caption{Connecting subset cycles (cycles in black, subset pixel in orange) by a minimum spanning tree (red edges) on the components/cycles.}
	\label{fig:mst_subset_connect}
\end{figure}
%Tour
\begin{theorem}
	\label{th:gg:aa:penalty:tour2d}
	There is a $12$-approximation algorithm for penalty tours  in grid graphs.
	% with a runtime of $O((\omega^2|V|)^{3.5})$ for penalty tours.
\end{theorem}
\begin{proof}
	The factor of 12 results of $4\text{OPT}$ for a penalty cycle
	cover, $2\cdot 2\text{OPT}$ for a 2-approximation of a prize-collecting Steiner
	tree, and again the cost of the penalty cycle cover as an upper bound on the
	necessary additional turns.  Just like for full coverage, the prize-collecting
	Steiner tree is computed on an auxiliary graph based on the cycles.
	For simplicity, we directly merge all cycles that share a pixel.
	The cost for this can be charged to one of the cycles from the same
	budget as for the later connections without any interference, because we
	reduce the number of cycles used in the later part.  In the corresponding
	graph, every cycle is represented by a vertex. The penalty for points not covered by cycles is already paid
	for by the penalty cycle cover. There is an edge between any two
	cycles; its cost is the cheapest transition from one pixel of the first cycle
	to a pixel of the second cycle, ignoring the orientations at these end pixels.
	The penalty of each vertex is the sum of the penalties of all of its pixels.
	Doubling a prize-collecting Steiner tree and removing all cycles not in the tree results in a tour; the turns at the ends can be charged to the cycle cover,
	just like for full/subset coverage.

	It remains to be shown that the optimal prize-collecting Steiner tree on the cycle graph is a lower bound on the optimal tour.
	This can be seen by obtaining such a tree from the optimal penalty tour.
	Because the turns at the end points are free, as is moving within the cycles, the resulting tree is not more expensive than the tour.
	Further, visiting a single pixel of a cycle covers also all other pixels of the cycle, so also the penalties are cheaper.
	Finally, a 2-approximation for the prize-collecting Steiner tree can be provided by the algorithm of Goemans and Williamson~\cite{goemans1995general}.
	%with a time complexity of $O(|V|^2 \log |V|)$.
\end{proof}

%For subset tours and tours in other grid graphs
%we can connect the cycles with limited extra cost as follows.
%Consider a complete weighted graph $G_\mathcal{C}(\mathcal{C}, E_\mathcal{C})$
%on the cycles of a cycle cover $\mathcal{C}$.
%We assume that all intersecting cycles are already connected; if necessary, charge this to the cycle that gets merged.
%The edge weight $w(c_1,c_2)$ for the edge between cycle $c_1$ and cycle $c_2$
%is the cheapest connection from a pixel of $c_1$ to a pixel of $c_2$
%(symmetric).  Doubling the corresponding paths of an MST on $G_\mathcal{C}$
%connects the cycle cover with an additional cost of at most $2\text{OPT}$
%for the MST and four simple turns (i.e.,, $360^\circ$) for connecting each of the $|\mathcal{C}|-1$ doubled paths with its two cycles; this is bounded by the cost of the cycle cover.
%For full coverage in more general grid graphs, the cost of the MST edges and the turns can be fully charged to the merged cycle.
%
%For penalty tours we have to account for the possibility that cycles can also be removed instead of connected.
%Here we can use a 2-approximation of the prize-collecting Steiner tree on $G_\mathcal{C}$, where the prize of
%a cycle is the sum of the penalties of all its pixels.
%Doubling the corresponding paths of this prize-collecting Steiner tree connects the cycle cover with an
%additional cost of at most $4\text{OPT}$ for the tree and four turns per cycle in $\mathcal{C}$, which is bounded by the cost of the cycle cover.
%See appendix for more details (Sec.~\ref{sec:app:tourapx}).
\subsection{Geometric Instances}
\label{sec:app:tourgeom}

In this section we give some additional information on solving the 2-dimensional geometric instances with polygonal obstacles.
We consider an (individual) set of $\omega$ atomic strips per point $p\in P$ and search for a cycle cover/tour that covers all points of $P$
(by integrating one of the atomic strips) with minimal distance and turn angle sum.
If there are only at most $\omega$ orientations per point, there are at most $(2\omega|P|)^2$ many possible transitions in the tour.
The cheapest transition between two configurations can be computed by
considering the visibility graph (this is analogous to the situation for
Euclidean shortest paths, see Fig.~\ref{fig:geo:path:turnpoints:circle} and e.g., Chapter~$15$ in de Berg at
al.~\cite{debergcg2008}) and a simple graph transformation to integrate the
turn costs into edge costs, which requires $O(|E|)$ vertices and $O(|E|)$
edges, see Fig.~\ref{fig:geo:path:transformation}.  The visibility graph can be computed in $O(n^2)$, e.g., by the
algorithm of Asano et al.~\cite{asano1986visibility}. % or
%Welzl~\cite{welzl1985constructing}.
%(This is worst-case optimal, there also exists an output-sensitive algorithm by Ghosh and Mount~\cite{ghosh1991output} that runs
%in $O(n \log n+k)$, where $k$ is the number of edges in the visibility graph.)
Skipping an atomic strip does not increase the costs, hence we can use the Atomic Strip Matching technique described before.
The connecting strategies remain also the same; the only additional subroutine is the computation of the geometric
primitives, which remains polynomial.
\begin{figure}
	\centering
	\includegraphics[width=0.5\textwidth]{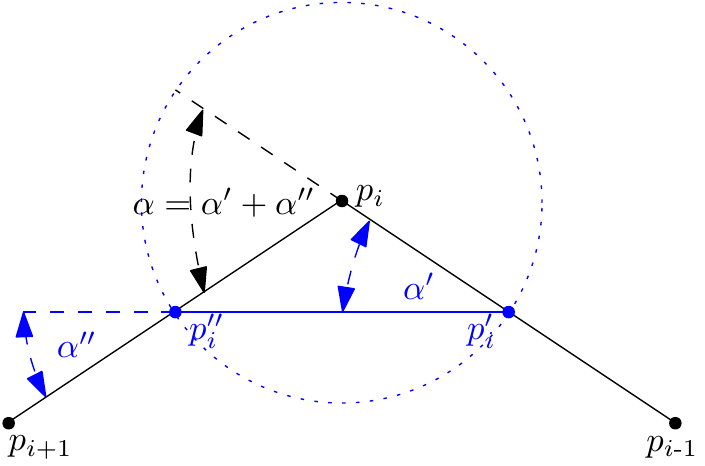}
	\caption[Illustration for minimum cost path proof]{If the turning point $p_i$ is not on a vertex of an obstacle, we can build a shorter tour with the same turn costs by replacing $p_i$ by $p'_i$ and $p''_i$ (blue).}
	\label{fig:geo:path:turnpoints:circle}
\end{figure}
\begin{figure}
	\centering
	\includegraphics[width=0.9\textwidth]{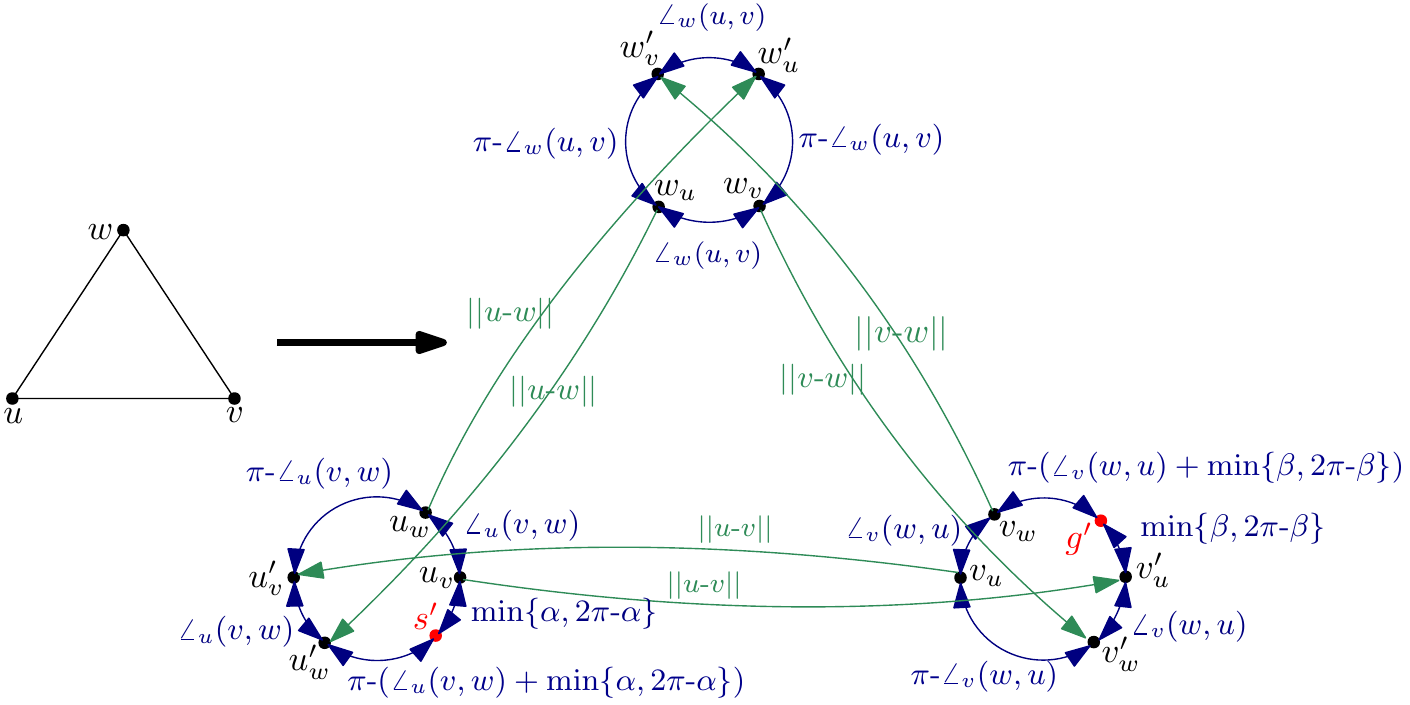}
	\caption[Transforming a graph with turn costs in a normal weighted digraph]{The transformation of the graph to integrate turn costs in the edge costs. The start configuration is in $u$ heading south-east ($s'$) and the goal configuration is in $v$ heading north-east ($g'$). $\angle_a(b,c)$ denotes the counterclockwise angle from segment $(a,b)$ to segment $(a,c)$.}
	\label{fig:geo:path:transformation}
\end{figure}

\section{Conclusions}
We have presented a number of theoretical results on finding optimal tours and cycle covers with turn costs.
In addition to resolving the long-standing open problem of complexity, we provided a generic framework to solve geometric (penalty) cycle cover and tours problems with turn costs.
%For grid graphs, this framework can engineered such that instances with over \num{300000} pixel can be computed as we show in a separate work~\cite{ALENEX19}.
%were able to demonstrate that LP/IP-based 
%solution methods do not only provide worst-case guarantees, but also allow practical performance evaluation when combined
%with appropriate algorithm engineering. We believe that this is an important and interesting perspective for other
%geometric optimization problems.

As described in~\cite{drone_vid}, the underlying problem is also of practical
relevance. As it turns out, our approach does not only yield polynomial-time
approximation algorithms; enhanced by an array of algorithm engineering techniques,
they can be employed for actually computing optimal and near-optimal solutions for instances of considerable size in grid graphs. 
%\todo[inline]{The solutions are better than Arkin et al.}
Further details on these algorithm engineering aspects will be provided in our
forthcoming paper~\cite{ALENEX19}.
%\todo[inline]{This conclusion needs some revision.}

%\subparagraph*{Acknowledgements.}
%
%I want to thank \dots

%%
%% Bibliography
%%
%\newpage
\bibliography{biblio}

\newpage
\appendix

\end{document}